\documentclass[fontsize=11pt, DIV=15]{scrartcl}

\usepackage[utf8]{inputenc}
\usepackage[T1]{fontenc}
\usepackage{lmodern}
\usepackage{microtype}

\usepackage{silence}

\usepackage{newunicodechar}
\newunicodechar{ρ}{\rho}
\newunicodechar{β}{\beta}
\newunicodechar{σ}{\sigma}
\newunicodechar{λ}{\lambda}
\newunicodechar{Λ}{\Lambda}
\newunicodechar{μ}{\mu}
\newunicodechar{ν}{\nu}
\newunicodechar{ψ}{\psi}
\newunicodechar{ϕ}{\varphi}
\newunicodechar{Φ}{\Phi}
\newunicodechar{φ}{\phi}
\newunicodechar{π}{\pi}
\newunicodechar{α}{\alpha}
\newunicodechar{ϵ}{{\epsilon}}
\newunicodechar{ε}{{\varepsilon}}
\newunicodechar{δ}{\delta}
\newunicodechar{ω}{\omega}
\newunicodechar{Δ}{\Delta}
\newunicodechar{Σ}{\Sigma}
\newunicodechar{∗}{\textasteriskcentered}

\usepackage{csquotes}

\usepackage{xparse}
\usepackage{xspace}
\usepackage{makecell}
\usepackage{titling}
\usepackage{appendix}
\usepackage{enumitem}

\usepackage{versions}
\includeversion{DRAFT}
\excludeversion{EXCLUDED}

\usepackage{amssymb}
\usepackage{amsmath}
\usepackage{mathrsfs}
\usepackage{mathtools}
\usepackage{optidef}
\usepackage{bbm}
\usepackage{bbold}
\usepackage{aligned-overset}
\usepackage{stmaryrd}
\SetSymbolFont{stmry}{bold}{U}{stmry}{m}{n} 

\usepackage{cool}
\usepackage[dvipsnames]{xcolor}
\usepackage{hyperref}
\hypersetup{
	colorlinks,
	linkcolor={red!50!black},
	citecolor={red!50!black},
	urlcolor={blue!80!black},
	linktoc=all
}
\usepackage{zref-clever}
\zcsetup{cap}
\usepackage{bookmark}

\usepackage[thmtools-compat]{keytheorems}
\newtheorem{theorem}{Theorem}[section]
\newtheorem{lemma}[theorem]{Lemma}
\newtheorem{corollary}[theorem]{Corollary}

\newtheorem{proposition}[theorem]{Proposition}
\newtheorem{example}[theorem]{Example}
\theoremstyle{definition}
\newtheorem{remark}[theorem]{Remark}
\newtheorem{definition}[theorem]{Definition}

\newtheorem*{definition*}{Definition}

\usepackage{graphicx}         
\setkeys{Gin}{width=\textwidth} %
\usepackage[multidot,spae]{grffile} 
\usepackage[compatibility=false]{caption}
\usepackage{subcaption}
\graphicspath{{autho./}{img/}}

\makeatletter
\let\x@caption\caption
\newcommand{\x@@caption}[2][\empty]{%
	\ifx\empty#1\relax\x@caption{#2}%
	\else\x@caption[#1]{\textbf{#1.} #2}%
	\fi}
\let\caption\x@@caption
\renewcommand{\captionabove}[2][\empty]{\captionsetup{position=above}%
	\x@@caption[#1]{#2}%
}
\renewcommand{\captionbelow}[2][\empty]{\captionsetup{position=below}%
	\x@@caption[#1]{#2}%
}
\makeatother

\usepackage{suffix}
\usepackage{xparse}
\usepackage{braket}

\usepackage{physics}

\renewcommand{\tr}{\Tr}

\usepackage{tensor}

\newcommand{\br}[1]{\left(#1\right)}

\NewDocumentCommand{\stared}{m m}{\IfBooleanTF{#2}{#1*}{#1}}
\NewDocumentCommand{\es}{m o}{
	\IfNoValueTF{#2}{
		\mathbb{\uppercase{#1}}^{\lowercase{#1\/}}
	}{
		\mathbb{\uppercase{#1}}^{#2\/}
	}
}

\newcommand{\reals}{\mathbb{R}}

\newcommand{\naturals}{\mathbb{N}}

\newcommand{\HS}{\mathcal{H}}

\NewDocumentCommand{\BO}{o}{\mathcal{B}\br{\IfValueTF{#1}{#1}{\HS}}}
\NewDocumentCommand{\BOsa}{o}{\mathcal{B}_{\mathrm{sa}}\br{\IfValueTF{#1}{#1}{\HS}}}

\NewDocumentCommand{\DM}{o}{\mathcal{D}\br{\IfValueTF{#1}{#1}{\HS}}}

\NewDocumentCommand{\pos}{o}{\mathscr{P}\!\br{\IfValueTF{#1}{#1}{\HS}}}

\DeclareMathOperator*{\Supp}{supp}
\NewDocumentCommand{\supp}{m}{\Supp(#1)}

\newcommand{\poly}{\mathrm{poly}}

\DeclarePairedDelimiterX{\infdivx}[2]{(}{)}{
	#1\;\delimsize\|\;#2
}

\DeclarePairedDelimiterX{\setx}[2]{\{}{\}}{\:#1\:\delimsize|\:#2\:}

\DeclareDocumentCommand{\qrd}{o o m m}{\ensuremath{\IfValueTF{#1}{#1}{D}\IfValueTF{#2}{_{#2}}{}\/(#3\|#4)}}

\NewDocumentCommand{\ent}{d()g}{\ensuremath{H\br{\IfValueTF{#1}{#1}{#2}}}}
\NewDocumentCommand{\qent}{d()g}{\ensuremath{H\br{\IfValueTF{#1}{#1}{#2}}}}

\newcommand*\ric[1]{\vphantom{#1}\smash{#1_{}\kern-\scriptspace}}

\renewcommand{\abs}[1]{\left\lvert#1\right\rvert}

\renewcommand{\norm}[1]{\left\lVert #1 \right\rVert}

\newcommand{\renyi}{R\'enyi\xspace}

\newcommand{\Sn}{\mathfrak{S}_n}
\NewDocumentCommand{\End}{o}{\mathrm{End}^{\Sn}(\IfValueTF{#1}{#1}{R^n B^n})}

\newcommand{\n}{^{\otimes n}}
\renewcommand{\k}{^{\otimes k}}

\renewcommand{\E}{\mathcal{E}}
\newcommand{\F}{\mathcal{F}}
\newcommand{\id}{\mathrm{id}}
\renewcommand{\D}{\mathbf{D}}

\usepackage{bm}

\renewcommand{\S}{\mathcal{S}}
\newcommand{\T}{\mathcal{T}}

\NewDocumentCommand{\cptp}{o}{\mathrm{CPTP}(\IfValueTF{#1}{#1}{A \to B})}

\newcommand{\X}{\mathcal{X}}
\newcommand{\Y}{\mathcal{Y}}

\newcommand{\1}{\IdentityMatrix}

\usepackage{centernot}

\newcommand{\B}{\mathcal{B}}
\NewDocumentCommand{\TC}{o}{\B_1\br{\IfValueTF{#1}{#1}{\HS}}}

\newlength\oversetwidth
\newlength\underwidth

\newcommand{\reg}{\mathrm{reg}}

\newcommand{\R}{R^\infty} 
\newcommand{\OR}{\tilde{\mathcal{O}}} 
\newcommand{\Ds}{\widetilde{D}_\alpha} 
\newcommand{\mE}{(\pi\cdot\mathcal{E}^{\otimes n})(\nu_n)} 
\newcommand{\mF}{(\pi\cdot\Hat{\mathcal{F}}_n')(\nu_n)} 
\NewDocumentCommand{\CQ}{o}{\mathrm{CQ}(\IfValueTF{#1}{#1}{\mathcal{\X} \to A})}
\numberwithin{equation}{section}

\usepackage{dsfont}

\usepackage{silence}
\usepackage[maxbibnames=99, maxalphanames=4, backend=biber]{biblatex} 

\AtEveryBibitem{
	\clearfield{note}
	\clearfield{issn}
	\clearfield{urlyear}
	\clearfield{urlmonth}
	\clearfield{eprintclass}
	\iffieldundef{journaltitle}{}{\clearfield{version}}
	\iffieldundef{journaltitle}{\iffieldundef{booktitle}{\iffieldundef{eventtitle}{\iffieldundef{eprint}{}{\clearfield{doi}}}{}}{}}{}
}

\renewbibmacro*{url}{%
	\iffieldundef{doi}{\iffieldundef{eprint}{%
			\printfield{url}%
		}{}}{}
}

\bibliography{references}
\bibliography{references_anirudh}

\DeclareMathOperator{\CPTP}{CQ}
\DeclareMathOperator{\spec}{spec}

\DeclareMathOperator{\Prob}{Prob}
\DeclareMathOperator{\SEP}{SEP}
\DeclareMathOperator{\BM}{\mathcal{B}}
\title{Generalized Quantum Stein's Lemma and Reversibility of Quantum Resource Theories for Classical-Quantum Channels}
\author{Bjarne Bergh\thanks{Department of Applied Mathematics and Theoretical Physics, University of Cambridge, United Kingdom} \and Nilanjana Datta\thanksmark{1} \and Anirudh Khaitan\thanksmark{1}}

\begin{document}

\maketitle
\begin{abstract}
We extend the recent proof of the Generalized Quantum Stein's Lemma by Hayashi and Yamasaki \emph{[arXiv:2408.02722]} to classical-quantum (c-q) channels. We analyze the composite hypothesis testing problem of testing a c-q channel $\E\n$ against a sequence of sets of c-q channels $(\S_n)_n$ (satisfying certain natural assumptions), under parallel strategies. We prove that the optimal asymptotic asymmetric error exponent is given by the regularization of Umegaki channel divergence, minimized over $\S_n$. This allows us to prove the reversibility of resource theories of classical-quantum channels in a natural framework, where the distance between channels (and hence also the notion of approximate interconvertibility of channels) is measured in diamond norm, and the set of free operations is the set of all asymptotically resource non-generating superchannels. The results we obtain are similar to the ones in the concurrent and independent work by Hayashi and Yamasaki \emph{[arXiv:2509.07271]}. However the proof of the direct part of the GQSL uses different arguments and techniques to deal with the challenges that arise from dealing with c-q channels. 
\end{abstract}
\tableofcontents

\section{Introduction}
In quantum information theory, information is encoded in quantum states and transmitted through quantum channels.  
Being able to reliably distinguish (or discriminate) between different states and channels is hence an essential requirement for the successful implementation of various quantum information-processing tasks. Such discrimination tasks amount to hypothesis testing problems. It is well-known that optimal error exponents for such hypothesis testing problems are often given by entropic expressions (see \zcref{sec:hypothesis_testing} below for an introduction to quantum hypothesis testing). For example, the optimal asymmetric error exponent of distinguishing independent and identically distributed (IID) copies of quantum states $ρ$ and $σ$ is given by the quantum relative entropy $D(ρ\|σ)$. Such entropic expressions also arise in the case of more complicated non-IID hypothesis testing tasks. In particular, one can consider the composite hypothesis testing problem where the {\em{null hypothesis}} is simple (i.e., given by a tensor power of a single state), whereas the \emph{alternate hypothesis} is given by a set of states. This was originally motivated in particular by the example of entanglement testing, where one wants to decide whether a given state on a bipartite Hilbert space $\HS_{A^nB^n}$ is either equal to $n$ copies of a promised entangled state, or given by some arbitrary separable state along the bipartition $A^n:B^n$. The Generalized Quantum Stein's Lemma (GQSL) states that the asymptotic asymmetric error exponent of this hypothesis testing task is similarly given by an expression involving the quantum relative entropy, but with an additional  infimum over the set of states and a and a regularization (see~\eqref{eq:GQSL_states} below for the exact statement). 
It is a highly nontrivial result, given that the states allowed in the set corresponding to the alternate hypothesis can be very highly non-IID. This result has found numerous applications (see e.g. \cite{berta_gap_2023} for an overview). However, in 2022, an error was discovered in the original proof of the GQSL \cite{berta_gap_2023}. This was recently resolved independently using two different approaches by Hayashi and Yamasaki~\cite{hayashi_generalized_2024} and Lami~\cite{lami_solution_2024}.

Distinguishing between quantum channels is a considerably more complex problem than distinguishing between states. This is because it involves an additional optimization over the input states to the channels, which may, in general, be entangled. Moreover, the input states could be chosen adaptively. 
In this paper, we prove an extension of the GQSL to an important subclass of quantum channels, namely, classical-quantum (c-q) channels. Similarly to the GQSL for states, we consider the null hypothesis to be simple (corresponding to an $n$-fold tensor power of a fixed c-q channel), and the alternate hypothesis to be composite, corresponding to the channel being a member of a set ${\mathcal{S}}_n$ of c-q channels satisfying certain assumptions (see \zcref{def:axioms} below). We restrict to parallel discrimination strategies, where the input state for the whole channel is picked before any measurements are performed on the outputs. We then prove that, similar to the GQSL for states, the optimal asymmetric error exponent for this hypothesis testing problem is given by a regularization of the corresponding Umegaki channel divergence minimized over the set ${\mathcal{S}}_n$. The assumptions we put on the sets $(\S_n)_n$ are natural analogues of the assumptions used in the proof of the GQSL for states in \cite{hayashi_generalized_2024}. Moreover, we impose an additional assumption that the sets are closed under permutations of the channels (see \zcref{def:axioms} for the explicit statement\footnote{An analogue of this assumption was also required in the original proof \cite{brandao_generalization_2010}, as well as in \cite{lami_solution_2024}.}). An explanation of the reason why we require this assumption 
is given in \zcref{remark:why_assumption_needed}. 

One of the main applications of the GQSL is that it allows one to show that certain quantum resource theories (QRTs) are reversible under a class of operations called asymptotically resource non-generating operations (ARNGs) \cite{brandao_generalization_2010, brandao_reversible_2015}. We show that if we quantify closeness of channels in diamond norm, a similar relation between the GQSL and reversibility in QRTs holds for c-q channels. In particular, we show that under natural assumptions on the set of free objects (in this case free c-q channels), the theory becomes reversible when the free operations (superchannels) are ARNG operations (where the resource generation is quantified through the $\max$ channel divergence; this is a natural analogue of the log robustness previously considered). 

Our approach differs from the approach taken in \cite{hayashi_generalized_2024}, where the GQSL for states was applied to Choi states of c-q channels to show reversibility in a different notion of QRTs for c-q channels.
In their approach, the distance between channels is quantified through the trace distance between their (normalized) Choi states, which corresponds to the distance of the output states of the channels averaged over all input states. Instead, we consider the distance between channels in diamond norm, which is the distance of the output states of the channels maximal over all input states. We consider the diamond norm to be a more natural notion of distance, since, in particular, it implies that if $\E_1 $ is approximately transformed to $\E_2$ via a free operation $\Theta$ then $\Theta(\E_1)$ and $\E_2$ are (approximately) indistinguishable, for any possible input state. In contrast, the Choi states of two different channels can be approximately equal, and yet the difference between the outputs of the two channels can be very different for a specific input state. Additionally, because of how Choi states behave under superchannels, the authors of \cite{hayashi_generalized_2024} have to put additional constraints on the set of free operations (i.e., free super-channels), which then becomes a subset of the ARNG operations (which is hard to specify explicitly). Such a restriction is unnecessary in our approach.

From a technical perspective, the general strategy for proving our GQSL for c-q channels is based on the ideas from \cite{hayashi_generalized_2024}, with the generalization to c-q channels relying on techniques from \cite{bergh_composite_2023}. In particular, the {\em{exchange lemma}}, \zcref{lem:rel_entropy_exchange} is extensively used throughout the proof. It allows us to exchange the supremum over input states with the infimum over the sets of channels. Additionally, the direct part of the proof of the GQSL exploits permutation invariance by suitably rephrasing the ideas from \cite{bergh_composite_2023} in \zcref{lem:perm_inputs} and \zcref{lem:perm_cov_channels}. These lemmas state that if the sets of channels are closed under permutations, then the extremizing channels can be chosen to be permutation covariant, with permutation-invariant optimal input states. The requirement that the channels are c-q allows us to often evaluate the supremum over input states in a simple way (see \zcref{lem:divergence_input_reduction}, \zcref{lem:input_reduce_diamond}), which also almost immediately yields the strong converse part of our GQSL.
\medskip

\noindent
\textbf{Note on concurrent work:} Similar results have recently also been obtained in the independent and concurrent work \cite{hayashi2025generalizedquantumsteinslemma}. The general outline of the arguments are similar. However the authors in \cite{hayashi2025generalizedquantumsteinslemma} employ a slightly different proof strategy for the direct part which also uses one less assumption than our approach. Even though their result can thus be considered to be slightly stronger, we still believe that our approach might be of interest. This is because the direct part of our proof of the GQSL makes use of the assumption that the channels are c-q in only one step of the proof (see the text below (\ref{S8_limsup_bound1})), and hence might be more straightforward to generalize to larger classes of channels (which are no longer c-q).

\medskip

\noindent
{\textbf{Layout of the paper:}} The paper is structured as follows: In \zcref{sec:intro} we introduce all the relevant concepts  and explain how they relate to previous work. In \zcref{sec:GQSL} we state and prove our first main result, namely, the Generalized Quantum Stein's Lemma for c-q channels (\Autoref{thm:GQSL}). \zcref{sec:QRT} contains our definition of QRTs for c-q channels and our second main result (\Autoref{thm:QRT}), namely, that such theories are reversible under ARNG operations. Finally, \zcref{sec:QRT_comparison} contains examples highlighting how our construction of QRTs differs from the previous construction involving Choi states of channels, and \zcref{sec:QRT_examples} gives multiple examples of QRTs that satisfy our assumptions and are thus reversible.

\section{Mathematical Introduction, Definitions and Previous Work}\label{sec:intro}
\subsection{Fundamentals of Quantum Information Theory}
Let $\HS$ denote a complex finite-dimensional Hilbert space, and $\BO$ be the set of linear operators acting on $\HS$. We write $\pos$ for the set of positive semi-definite operators acting on $\HS$. For $A, B \in \pos$, we write $A \ll B$ if $\supp{A} \subseteq \supp{B}$  and $A \not \ll B$ otherwise.
Let $\DM$ denote the set of density matrices, i.e., the set of positive semi-definite operators of unit trace. A quantum channel (denoted usually as $\E$ or $\F$) is a completely positive trace-preserving map 
acting on $\BO$. We will label different quantum systems with capital Roman letters ($A$, $B$, $C$, etc.) and often use these letters interchangeably with the corresponding Hilbert space or set of density matrices (i.e., we write $\rho \in \DM[A]$ instead of $\rho \in \DM[\HS_A]$ and $\E: A \to B$ instead of $\E: \BM(\HS_A)\to \BM(\HS_B)$). We will also concatenate these letters to denote tensor products of systems, i.e.,~we will write $\rho \in \DM[RA]$ for $\rho \in \DM[\HS_R \otimes \HS_A]$. We write $\cptp$ for the set of all completely positive trace-preserving maps from $\BM(\HS_A)$ to $\BM(\HS_B)$. The identity operator in $\BM(\HS)$ is denoted as $\IdentityMatrix$, whereas the identity map acting on $\BM(\HS_R)$, where ${\mathcal{H}}_R$ is the Hilbert state associated with a system $R$, is denoted as $\id_R$.

Given a specific orthonormal basis $\{\ket{e_i}\}_{i=1}^{d_A}$ of $\HS_{A}$, we denote the \emph{Choi state} of a channel $\E: A \to B$ as $J(\E) = J(\E)_{\tilde{A}B} \coloneqq (\id_{\tilde{A}} \otimes \E)(\Phi_{\tilde{A}A})$, where $\Phi_{\tilde{A}A} = \ketbra{\Phi}{\Phi}_{\tilde{A}A}$ and $\ket{\Phi}_{\tilde{A}A} = {1 \over \sqrt{d_A}} \sum_{i}\ket{e_i}_{\tilde{A}} \otimes \ket{e_i}_{A}$. The channel output for any input state $ρ \in \DM[A]$ can then be expressed as $\E(ρ) = d_A \Tr_{\tilde{A}}(J(\E)_{\tilde{A}B} (\rho_{\tilde{A}}^T \otimes \IdentityMatrix_B))$, where the transposition is taken in the chosen basis of $\HS_A$.  

The {\em{trace distance}} between any two states $ρ, σ \in \DM$, is given by ${1 \over 2}\norm{ρ - σ}_1$, where $\norm{A}_1 = \Tr(\abs{A}) = \Tr(\sqrt{A^\dagger A})$. The \emph{diamond (norm) distance} between two channels $\E$ and $\F \in \cptp[A \to B]$ is given by $\norm{\E - \F}_{\Diamond} \coloneqq \sup_{\nu \in \DM[RA]} \norm{(\id_{R} \otimes \E)(\nu) - (\id_{R} \otimes \F)(\nu)}_1$, where the supremum can be restricted to a reference system $R$ whose Hilbert space is isomorphic to that of $A$.

\subsection{Quantum Divergences}
    For $\rho \in \DM$ and $\sigma \in \pos$ the \emph{(Umegaki) quantum relative entropy} is defined as~\cite{umegaki_conditional_1962}
	\begin{equation}
		D(\rho\|\sigma) \coloneqq \Tr(\rho(\log \rho - \log \sigma)),
	\end{equation}
	if $\rho \ll \sigma$ and 
	$D(\rho\|\sigma) \coloneqq \infty$ if $\rho \not \ll \sigma$ (with the convention that $0 \log 0 = 0$, and we use the natural  logarithm throughout the paper). One of its most important properties is the {\em{data-processing inequality}}~\cite{lindblad_completely_1975}, which states that for every quantum channel $\E$ acting on $\BM(\HS)$,
	\begin{equation}
		D(\rho\|\sigma) \geq D(\E(\rho)\|\E(\sigma)) \,.
	\end{equation}
	More generally, we call a function of $\rho$ and $\sigma$ a (generalized) divergence if it satisfies the data-processing inequality.
    \smallskip
    We will also make use of the \emph{sandwiched \renyi divergence}
\parencite{muller-lennert_quantum_2013, wilde_strong_2014}, which is for $\rho \geq 0, σ > 0$, defined as
\begin{equation}
    \widetilde{D}_\alpha (\rho\|\sigma) := \frac{1}{\alpha -1} \log \tr \left(\sigma^{\frac{1-\alpha}{2 \alpha}}\rho \sigma^{\frac{1-\alpha}{2 \alpha}}\right)^\alpha\,,
\end{equation}
and can be defined for a $σ$ which is not of full rank by setting $\widetilde{D}_{α}(ρ\|σ) \coloneqq \lim_{ε \to 0} \widetilde{D}_{α}(ρ\|σ + ε \IdentityMatrix)$. Whether this limit is finite depends on whether $ρ \ll σ$ and whether $α$ is smaller or bigger than one, we will be exclusively interested in the case where $α > 1$ and in this case $\widetilde{D}_{α}(ρ\|σ) < \infty$ if and only if $ρ \ll σ$. 
\smallskip

Yet another quantum divergence we use is the so-called \emph{max-relative entropy}. It is the quantum analogue of the classical maximum log-likelihood ratio, and is defined as follows \parencite{datta_min-_2009}:
\begin{equation}
    D_{\max}(\rho\|\sigma) \coloneqq \log \inf\Set{\lambda \in \reals | \rho \leq \lambda \sigma}\,.
\end{equation}

A divergence $\D$ is said to satisfy the direct-sum property, if 
\begin{equation}
    \D\infdivx*{\bigoplus_{i=1}^n p_i \rho_i }{\bigoplus_{i=1}^n p_i \sigma_i}= \sum_{i=1}^n p_i \D(\rho_i\|\sigma_i)\,.
\end{equation}
whenever $\rho_i, \sigma_i \in \HS_i$ are two sets of density matrices and $\{p_i\}_{i = 1}^n$ is a probability distribution. This is satisfied in particular for the quantum relative entropy $D$. 

A divergence $\mathbf{D}$ is faithful if, for any density matrix $\rho$, we have that
\begin{eqnarray}
    \mathbf{D}(\rho\| \rho)=0\,.
\end{eqnarray}

{A divergence $\mathbf{D}$ is said jointly quasi-convex if, for any probability distribution $\{p_i\}_{i=1}^n$, we have that}
\begin{eqnarray}
    \mathbf{D}\left(\sum_{i=1}^n p_i\rho_i \bigg\| \sum_{i=1}^np_i\sigma_i \right)\leq \max_{1\leq i\leq n}\mathbf{D}(\rho_i\| \sigma_i)
\end{eqnarray}
for any density matrices $\rho_i,\sigma_i$ belonging to the same system.

When a quantum channel acts only on one part of a quantum state, we will often omit writing the identity map (which acts on the other part) explicitly, i.e., for a channel $\E: A \to B$ and a state $\nu_{RA} \in \DM[RA]$ we use the following notational simplification:
\begin{equation}
    (\id_R \otimes \E)(\nu_{RA}) \equiv \E(\nu_{RA}) \equiv \E(\nu)\,.
\end{equation}


For a divergence $\mathbf{D}$ acting on states (e.g.~$\mathbf{D} = D$, the quantum relative entropy), we define the corresponding channel divergence via an optimization over all possible input states (including a reference system). If $\E, \F \in \cptp[A \to B]$ are two quantum channels,
\begin{equation}\label{eq:def_channel_divergence}
\mathbf{D}(\E\|\F) \coloneqq \sup_{\nu_{RA} \in \DM[RA]} \D((\id_R \otimes \E)(\nu_{RA}) \|(\id_R \otimes \F)(\nu_{RA})) = \sup_{\nu_{RA} \in \DM[RA]} \D(\E(\nu)\|\F(\nu))
\end{equation}
where we used our abbreviated notation with implicit identities from above. In particular, for the choice $\mathbf{D}=D$, we refer to the divergence as the \emph{Umegaki channel divergence}. Using a purification argument, one can easily see that the data-processing inequality for $\D$ implies that the reference system $R$ can be chosen isomorphic to $A$ (see e.g. \cite{bergh_composite_2023} for an explicit argument).

A result that we will make use of multiple times is that for the $D_{\max}$ divergence, this optimization is known to be achieved for a maximally entangled state (\cite[Lemma 12]{wilde_amortized_2020}), i.e.
\begin{lemma}[\cite[Definition~19]{diaz_using_2018},  \cite[Lemma 12]{wilde_amortized_2020}] \label{lem:channel_divergence_maximally_entangled_state}
Let $\E, \F \in \cptp[A \to B]$ be two quantum channels, then for any maximally entangled state $\Phi$, i.e. for any basis $\{\ket{e_i}\}_{i=1}^{d_A}$ of $A$ and corresponding vector $\ket{\Phi}_{\tilde{A}{A}} =  {1 \over \sqrt{d_A}} \sum_{i} \ket{e_i}_{\tilde{A}} \otimes \ket{e_i}_{A}$ with state $\Phi =\ketbra{\Phi}{\Phi}_{\tilde{A}A}$, we have that
\begin{equation}
    D_{\max}(\E(\Phi)\|\F(\Phi)) = D_{\max}(\E\|\F)\,.
\end{equation}
\end{lemma}
\subsection{Classical-Quantum Channels}
We use the symbols $\X$ and $\Y$ to denote finite alphabets (or equivalently finite sets corresponding to discrete classical sample spaces or classical systems). We can embed such a classical system $\X$ into a quantum system by associating with it a finite-dimensional Hilbert space $\HS_{\X}$ that has an orthonormal basis, $\{\ket{x}\}_{x \in \X}$, labeled by the elements $x$ of $\X$. A classical-quantum (c-q) channel $\E: \X \to A$ can be seen as a channel from the classical set $\X$ , i.e., a channel that outputs the state of a quantum system $A$, when an input $x \in \X$ is sent through it. However, since we will employ many results stated for general quantum channels, we want to treat the set of c-q channels as a subclass of the set of all quantum channels. For this, we view a c-q channel as a fully quantum channel where the input quantum state is first subjected to a projective measurement in the basis $\{\ket{x}\}_{x \in \X}$ of $\HS_{\X}$.  With a slight abuse of notation, we also denote this channel as $\E:\X \to A$, and the set of c-q channels as $\CQ[\X \to A]$. Explicitly, this means that the set $\CQ[\X \to A]$ contains all channels of the form
\begin{equation}
    \E(ρ) = \sum_{x \in \X} \Tr[ρ\ketbra{x}{x}] ω_{x}
\end{equation}
for any $\rho \in \DM[\HS_{\mathcal{X}}]$, where $\{ω_{x} \in \DM[A]\}_{x}$ is a collection of quantum states. 
Furthermore, we often also write $\DM[\X] = \DM[\HS_{\X}]$ for the set of \emph{general} quantum states on $\HS_{\X}$, in particular, a state $\DM[R\X]$ can be entangled.

In many cases, optimizations of distance measures over input states for c-q channels are achieved in a particularly simple manner, where the intuition is that c-q channels cannot take advantage of entangled input states (in fact, they are a special case of entanglement-breaking channels). This is very useful in particular for the proof of the strong converse in \autoref{ch_strong_converse}, since the following lemmas imply that the optimal input state is always a tensor product state, when applied to the c-q channels $\E^{\n},\F_n:\mathcal{X}^n\rightarrow A^n$. For this, we prove the following two Lemmas in \zcref{Appendix}:
\getkeytheorem{stored:lem:divergence_input_reduction}

\getkeytheorem{stored:lem:diamond_input_reduction}

\subsection{Permutations and the Symmetric Group}
\label{ch:perm}
We write $\Sn$ for the symmetric group, i.e., the group of permutations on $n$ elements. We denote elements of this group as $\pi \in \Sn$, and their unitary representations on a Hilbert space $\HS\n$ -- acting by permuting the $n$ copies -- as $P_{\HS}(\pi)$.\\
\\
This means that $\Sn$ acts on a density operator $\nu_n \in \mathcal{D}(\mathcal{H}_A^{\otimes n})$ as
\begin{eqnarray}
    \nu_n \mapsto P_A(\pi)\nu_n P_A(\pi)^\dag 
\end{eqnarray}
We say that $\nu$ is \textit{permutation invariant} if $P_A(\pi)\nu_n P_A(\pi)^\dag=\nu_n$ for all permutations $\pi\in\Sn$. \\
\\
Similarly, $\Sn$ acts on a quantum channel $\F \in \CPTP(A^n\rightarrow B^n)$ to give another quantum channel $\F_{\pi}$, defined as 
\begin{eqnarray}
    \F_{\pi}(\nu_n):=P_B(\pi)^\dag\F(P_A(\pi)\nu_n P_A(\pi)^\dag)P_B(\pi)
\end{eqnarray}
The channel $\F$ is said to be \textit{permutation covariant} if $\F_{\pi}=\F$ for all $\pi\in\Sn$. \\
\\
Henceforth, we use the standard notation 
\begin{eqnarray}
    \mE:=P_A(\pi)^\dag\E\n(P_\X(\pi)\nu_n P_\X(\pi)^\dag)P_A(\pi)
\end{eqnarray}
for any $\nu_n\in\mathcal{D}(R\X^n)$, and $\pi\in\Sn$.
\subsection{Pinching Maps and the Pinching Inequality}
Let $\mathcal{H}_A$ be a finite-dimensional Hilbert space, and let $(E_j)_{j=1}^k$ be a (finite) set of orthogonal projections satisfying $\sum_{j}E_j=\mathds{1}_A$. Then we define the \textit{pinching} map $E:A\rightarrow A$ with respect to the projectors $\{E_j\}_j$ as
\begin{eqnarray}
    E(\rho):=\sum_{j=1}^k{E_j\rho E_j}\,,
\end{eqnarray}
which can easily be seen to be a quantum channel.

Similarly, we can define a pinching map with respect to a state $\sigma$. If the spectral decomposition of $\sigma$ is $σ = \sum_{j=1}^k{\lambda_jE_j}$ for distinct eigenvalues $\lambda_j$ and corresponding eigen-projectors $E_j$, then we define the pinching map \emph{with respect to $σ$} to be the pinching map with respect to the eigen-projectors $\{E_j\}_j$. 

One of the key properties of the pinching map is that, if $E$ is a pinching map with respect to $\sigma$, then for any state $\rho$, $E(\rho)$ and $\sigma$ commute. This is because the pinching makes $E(\rho)$ block diagonal with respect to the eigenspaces of $\sigma$, by deleting the relevant off-diagonal blocks of $\rho$. 

Furthermore, the \textit{pinching inequality} \cite{hayashi_optimal_2002} gives a relation between the original state and the pinched state. It states that, for any state $\rho$ and a pinching map $E$ with $k$ distinct mutually orthogonal projectors,
\begin{eqnarray}
    \label{eq:pinching_inequality}
    \rho \leq kE(\rho)\,.
\end{eqnarray}
In particular, when the pinching map is with respect to $σ$, then $k$ is the number of distinct spectral points of $σ$. There are quite a few cases in which this number 
can be conveniently bounded, for example if $σ_n \in \DM[\HS\n]$ is a permutation-invariant state (in particular also if it is a tensor-product state) then the number of spectral points is bounded as $k = \mathcal{O}(\poly(n))$ (this is well-known see e.g.\ \cite[Lemma II.1]{csiszar_method_1998}). \\
\\
The pinching inequality can also be used to obtain a bound on the quantum relative entropy which is shown in the following lemma,, and which we use in \Autoref{lem:main_iteration} below.
\begin{lemma}
    \label{lem:pinching_bound}
    Let $\rho,\sigma$ be density matrices acting on the same Hilbert space, and let $E$ denote the pinching map with respect to the spectral projectors of $\sigma$. Then
    \begin{eqnarray}
        D(\rho \| \sigma) \leq D(E(\rho) \| \sigma)+\log k
    \end{eqnarray}
    where $k=\abs{\spec \sigma}$.
\end{lemma}
\begin{proof}
    We have
    \begin{eqnarray}
        D(\rho \| \sigma) - D(E(\rho) \| \sigma) &=& \Tr[\rho\log\rho - \rho\log\sigma-E(\rho)\log E(\rho)+E(\rho)\log \sigma] \nonumber\\
        &\overset{(a)}{=}& \Tr[\rho\log\rho - E(\rho)\log\sigma-\rho\log E(\rho)+E(\rho)\log \sigma] \nonumber\\
        &=& \Tr[\rho\log\rho-\rho\log E(\rho)]
    \end{eqnarray}
    where (a) uses the fact that, for any Hermitian matrices $M,N$ and pinching $E_N$ with respect to the spectral projectors of $N$, we have that $\Tr[MN]=\Tr[E_N(M)N]$, and the fact that $σ$ and $\log σ$ have the same spectral projectors. \\
    \\
    We now use the pinching inequality (\ref{eq:pinching_inequality}) and the operator monotonicity of the logarithm to get
    \begin{eqnarray}
        \log \rho \leq \1\log k+\log E(\rho)
    \end{eqnarray}
    Hence,
    \begin{eqnarray}
        \Tr[\rho\log\rho-\rho\log E(\rho)] \leq \log k
    \end{eqnarray}
    as required.
\end{proof}

\subsection{Quantum Hypothesis Testing}\label{sec:hypothesis_testing}
Hypothesis testing deals with the question of asserting the truth of one of multiple possible hypotheses given some object or data. In the simplest case, the given object is a quantum state and the task is to identify which of two (fully specified) options it is. Generically, if there are two hypotheses, there are two ways of making an error (mistaking the first for the second, or mistaking the second for the first). Given a decision strategy, we will call the corresponding probabilities of making such an error the type-I and type-II error probabilities (we will also use the term type-I and type-II errors interchangeably). In the simple case where the task is to identify a given state as either $ρ$ or $σ$, these two probabilities are given by
\begin{align}
α &\coloneqq \mathbb{P}[\text{we claim the state is $σ$ } | \text{ state is actually $ρ$}]  \qquad \qquad \text{ the type-I error,}\\
β &\coloneqq \mathbb{P}[\text{we claim the state is $ρ$ } | \text{ state is actually $σ$}]  \qquad \qquad \text{ the type-II error.}
\end{align}

The most general way to arrive at such a decision in the case of distinguishing $ρ$ and $σ$ is by performing a binary (i.e., a two-outcome) POVM measurement, which is fully specified by one of its elements $0 \leq M \leq \IdentityMatrix$, and we use the convention that an outcome corresponding to the measurement $M$ (resp.~$\IdentityMatrix - M$) leads to the inference that the state is $\rho$ (resp.~$\sigma$). In that case the corresponding error probabilities are given by
\begin{align}
α &= \Tr((\IdentityMatrix - M) ρ) = 1 - \Tr(Mρ)\\
β &= \Tr(M σ) \,.
\end{align}

We can then optimize over $M$ to find the optimal measurement in a certain sense. In the so-called asymmetric setting of hypothesis testing the aim is to minimize the type-II error probability given the constraint that the type-I error is below a chosen threshold $ε \in [0,1]$:
\begin{equation}
    \min_{M:\, 0 \leq M \leq \IdentityMatrix_{\HS}}\Set{\Tr(M σ) | 1 - \Tr(M ρ) \leq ε} = \min_{\substack{M:\, 0 \leq M \leq \IdentityMatrix_{\HS} \\ \Tr(M \rho) \geq 1- \varepsilon}} \Tr(M σ)
\end{equation}

The negative logarithm of this is called \emph{hypothesis testing relative entropy} \parencite{wang_one-shot_2012}:
\begin{equation}\label{eq:hypothesis_testing_relative_entropy}
    D_H^{\varepsilon}(\rho\|\sigma) \coloneqq - \log \qty[\min_{\substack{M:\, 0 \leq M \leq \IdentityMatrix_{\HS} \\ \Tr(M \rho) \geq 1- \varepsilon}} \Tr(M σ)] = \max_{\substack{M:\, 0 \leq M \leq \IdentityMatrix_{\HS} \\ \Tr(M \rho) \geq 1- \varepsilon}} - \log(\Tr(M σ))\,.
\end{equation}
It has been given this name, since it shares some properties with the quantum relative entropy, in particular, it satisfies the data-processing inequality \parencite{wang_one-shot_2012}. 

One often considers the case in hypothesis testing where one can make use of many IID (independent and identically distributed) copies of the quantum states associated to the hypotheses. In that case, the asymmetric error probability will decay exponentially in the number of samples (usually denoted as $n$), and one can study the asymptotic rate of this decay, which is known as the \emph{asymmetric error exponent} and given by 
\begin{equation}
    \lim_{n \to \infty}{1 \over n} D_H^{ε}(ρ\n\|σ\n)\,,
\end{equation}
which is the optimal achievable decay rate of the type-II error, such that also the type-I error stays below $ε$. The \emph{quantum Stein's lemma} states that this limit is given by the relative entropy between $ρ$ and $σ$, i.e., for $ε \in (0,1)$ \parencite{hiai_proper_1991, ogawa_strong_2000}
\begin{equation}\label{eq:quantum-steins-lemma-weak-converse}
    \lim_{n \to \infty}{1 \over n} D_H^{ε}(ρ\n\|σ\n) = D(ρ\|σ)\,.
\end{equation}
The fact that the right-hand side does not depend on $ε$, and so the exponential decay rate does not depend on the allowed type-I error, is known as the strong-converse property. 
\subsection{Hypothesis Testing for Channels}
In quantum channel discrimination, the given objects and hypotheses are not quantum states, but quantum channels, i.e., one is given an unknown quantum channel (which can be thought of as a black-box which takes a quantum state as an input and outputs another quantum state) and the promise that the black-box acts like one of two possible quantum channels, and the task is to find out which one. This gives rise to the same notion of type-I and type-II error probabilities for the two ways of confusing the two hypotheses as above. In the actual decision-making process though, we now have the additional freedom of picking an input state for the channel, which in general could also be entangled with a reference system that is untouched by the channel. If we want to distinguish between the two quantum channels $\E, \F \in \cptp[A \to B]$, and if we are allowed to use the channel just a single time, we can optimize over both the input state and the measurement and obtain the following expression for the optimal type-II error probability with a type-I error constraint (similar to the state case above):
\begin{equation}
    \min_{\nu \in \DM[RA]} \min_{\substack{M: \, 0 \leq M \leq \IdentityMatrix_{RB} \\ \Tr(M \E(\nu)) \geq 1- \varepsilon}} \Tr(M \E(ν))
\end{equation}
where we used the notation involving implicit identities as explained above. We can use this to define a hypothesis-testing relative entropy for channels which happens to be also the channel divergence (as defined above) associated to the hypothesis-testing relative entropy for states:
\begin{equation}
    D_H^{ε}(\E\|\F) \coloneqq - \log \qty[ \min_{\nu \in \DM[RA]} \min_{\substack{M: \, 0 \leq M \leq \IdentityMatrix_{RB} \\ \Tr(M \E(\nu)) \geq 1- \varepsilon}} \Tr(M \E(ν))] = \max_{ν \in \DM[RA]} D_H^{ε}(\E(\nu)\|\F(\nu))
\end{equation}

If one is allowed to use the channel more than once, the situation becomes more interesting as one could employ an adaptive discrimination strategy and make the input states of subsequent channel uses depend on the output of previous channel uses. It is not hard to construct examples where one can show that this can lead to strictly better error probabilities for a finite number of channel uses \cite{harrow_adaptive_2010, salek_usefulness_2022}. Adaptive strategies can also lead to better asymptotic error exponents in some scenarios \cite{salek_usefulness_2022}, but for the asymmetric error exponent in particular (i.e., the optimal decay rate of the type-II error under a type-I error threshold), one can show that this is not the case, i.e., the optimal asymptotic exponent can also be achieved by a non-adaptive (also called parallel) strategy \cite{fang_chain_2020}. Such a parallel strategy fixes the input state (which can still be entangled between different channels inputs) and the measurement at the beginning, which leads to the logarithm of the optimal type-II error probability (given $n$ channel uses and a type-I error threshold of $ε$) being given by $D_H^{ε}(\E\n\|\F\n)$. The asymptotic limit of this is then given by \cite{wang_resource_2019}
\begin{equation}
    \lim_{ε \to 0} \lim_{n \to \infty} {1 \over n} D_H^{ε}(\E\n\|\F\n) = D^{\reg}(\E\|\F) \coloneqq \lim_{n \to \infty} {1 \over n} D(\E\n\|\F\n)\,.
\end{equation}
Note that it is not known whether a strong-converse property holds for this problem with general quantum channels (see also  \cite{fawzi_defining_2021, fang_towards_2022}), and so the best known result at the moment involves an additional limit $ε \to 0$ as stated.

If the given channels are c-q channels the situation simplifies a bit further, in that one can show using \zcref{lem:divergence_input_reduction} that
\begin{equation}
D^{\reg}(\E\|\F) = D(\E\|\F)
\end{equation}
and additionally the strong-converse property holds in this case, so a limit $ε \to 0$ is not needed \cite{wilde_amortized_2020}. Furthermore, in the c-q case, equality between adaptive and non-adaptive error exponents holds also for more than just the asymmetric error exponent (for example for Chernoff and Hoeffding exponents, where equality fails to hold for general quantum channels \cite{salek_usefulness_2022}).

\subsection{Composite Hypothesis Testing}\label{sec:composite_HT}
In composite hypothesis testing, the hypotheses need no longer be simple, which means they can include more than one possible option. For example, in the case of distinguishing states, instead of a hypotheses specifying that the state is exactly $ρ$, the hypothesis just specifies that the given state comes from a specific set $S \subset \DM[\HS]$. The task is still just to find the correct hypothesis, so we need not identify the state exactly, but just correctly identify which set it comes from. Since now, given one of the two hypotheses, there can be many states we could encounter from this hypothesis, we have to additionally specify how we define type-I and type-II errors, since the error probability (given a specific strategy) will in general depend on which state we do encounter. Here, we will be looking at \emph{worst-case error probabilities}, i.e., given a decision strategy (i.e., a measurement), the type-I and type-II error probabilities are given by the worst-case over all states coming from the sets corresponding to the hypotheses. Specifically, if the first hypothesis corresponds to the set $S \subset \DM$ and the second hypothesis to the set $T \subset \DM$, and we make a decision based on the measurement $0 \leq M \leq \IdentityMatrix_{\HS}$, we will arrive at type-I and type-II errors:
\begin{align}
    α(M,S) &\coloneqq \sup_{ρ \in S} \Tr((\IdentityMatrix - M) ρ) \\
    β(M,T) &\coloneqq \sup_{σ \in T} \Tr(Mσ) \,.
\end{align}
Proceeding similarly as above, this leads to the following notion of hypothesis testing relative entropy between sets
\begin{equation}
    D_H^{ε}(S\|T) \coloneqq -\log\left[\min_{\substack{0\leq M \leq \1 \\ \alpha(M,S)\leq \varepsilon}}\beta(M,T)\right] 
\end{equation}
One can show that if $S$ and $T$ are convex, then \cite[Lemma 31]{fang_generalized_2025}
\begin{equation}\label{eq:hypothesis_testing_relative_entropy_exchange}
    D_H^{ε}(S\|T) = \inf_{\substack{ρ \in S \\ σ \in T}} D_H^{ε}(ρ\|σ)\,.
\end{equation}
This also motivates the following definition for generalized divergences $\mathbf{D}$:
\begin{equation}
\mathbf{D}(S\|T) \coloneqq \inf_{\substack{ρ \in S\\σ \in T}}\mathbf{D}(ρ\|σ)\,.
\end{equation}
where we will also write $\D(ρ\|T) = \D(\{ρ\}\|T)$ if the set $S = \{ρ\}$ contains only a single element. 

The \emph{Generalized Quantum Stein's Lemma} (GQSL) deals with the asymptotic asymmetric error exponent in a composite hypothesis testing problem where only the alternate hypothesis (i.e., the second argument in the divergences) is composite and is given by a sequence of sets. In particular, let $(S_n \subset \DM[\HS\n])_n$ be a sequence of sets of states, such that the following holds:
\begin{enumerate}
\item Each set $S_n$ is closed and convex as a subset of $\DM[\HS\n]$.
\item The sets $(S_n)$ are closed under tensor products, i.e., if $σ_n \in S_n$ and $σ_m \in S_m$ then $σ_n \otimes σ_m \in S_{n + m}$. 
\item The set $S_1$ contains a full-rank state.
\end{enumerate}

Then for any $ρ \in \DM$ and any $ε \in (0, 1)$, the Generalized Quantum Stein's Lemma states that \cite{hayashi_generalized_2024} (see also \cite{brandao_generalization_2010, berta_gap_2023, lami_solution_2024}):
\begin{equation}
\lim_{n \to \infty} {1 \over n} D_H^{ε}(ρ\n\|S_n) = \lim_{n \to \infty} {1 \over n} D(ρ\n\|S_n)\,.
\label{eq:GQSL_states}
\end{equation}

Note that the original (incorrect) proof \cite{brandao_generalization_2010} and the proof in \cite{lami_solution_2024} additionally use the following two assumptions
\begin{enumerate}
    \item Tracing out any one of the $n$ systems in a state $σ_n \in S_{n}$ yields a state in $S_{n - 1}$
    \item For any state $σ_n \in S_{n}$ and a permutation $\pi \in \Sn$ with corresponding representation $P_{\HS}(\pi)$ that permutes the $n$ copies of $\HS$, also $P_{\HS}(\pi)σ_n P_{\HS}(\pi)^\dagger \in S_n$.
\end{enumerate}

\subsection{Composite Hypothesis Testing for Channels}\label{sec:composite_HT_channels}
Generalizing the previous section to quantum channels, we are now given again a channel as a black-box with the promise that it corresponds to a channel out of two possible sets $\S, \T \subset \cptp[A \to B]$. After choosing an input state $\nu \in \DM[RA]$ and a measurement $0 \leq M \leq \IdentityMatrix_{RB}$ we again introduce the worst-case notion of error probabilities
\begin{align}
    α(M,\mathcal{S},\nu) &\coloneqq \sup_{\E \in \S} \Tr((\IdentityMatrix - M) \E(\nu)) \\
    β(M,\mathcal{T},\nu) &\coloneqq \sup_{\F \in \T} \Tr(\F(\nu) M)
\end{align}
and corresponding hypothesis testing relative entropy
\begin{equation}
D_H^{ε}(\S\|\T)  \coloneqq -\log\left[\min_{\nu\in\mathcal{D}(RA)}\min_{\substack{0\leq M\leq\1 \\ \alpha(M,\mathcal{S},\nu)\leq \varepsilon }}\beta(M,\mathcal{T},\nu)\right].
\end{equation}
 It is easy to see that
\begin{equation}
    D_H^{ε}(\S\|\T) = \max_{\nu \in \DM[RA]} D_H^{ε}(\S[\nu]\|\T[\nu])
\end{equation}
where we used the notation $\S[\nu] \coloneqq \Set{\E(\nu) | \E \in \S}$ for all the output states of the channels in the set $\S$ given the input state $\nu$. Using \eqref{eq:hypothesis_testing_relative_entropy_exchange} one then also finds that if $\S$ and $\T$ are convex
\begin{equation}
    D_H^{ε}(\S\|\T) = \max_{\nu \in \DM[RA]} \inf_{\substack{\E \in \S\\ \F \in \T}} D_H^{ε}(\E(\nu)\|\F(\nu))\,.
\end{equation}
This again motivates the following shorthand notation for any divergence $\D$: 
\begin{equation}\label{eq:shorthand_composite_channels}
    \mathbf{D}(\S\|\T) \coloneqq \sup_{\nu \in \DM[RA]}\inf_{\substack{\E \in \S\\\F \in \T}}\mathbf{D}(\E(\nu)\|\F(\nu))
\end{equation}
We prove in \Autoref{lem:rel_entropy_exchange} below that the infimum and supremum can be exchanged for certain divergences. However, for many divergences, the condition of this lemma is not satisfied, and hence the order of the supremum and infimum matters. Similarly as for the state case we also write
\begin{equation}
    \D(\E\|\S) \coloneqq  \D(\{\E\}\|\S)\,.
\end{equation}

Given a sequence of sets $\S_n, \T_n \subset \cptp[A^n \to B^n]$, and under some additional assumptions on these sets (in particular, that the channels they contain are (convex combinations of) tensor products of channels), asymmetric error exponents for such channel discrimination tasks using adaptive and parallel strategies were explored in \cite{bergh_composite_2023}. 

\subsection{Quantum Resource Theories}
\label{sec:intro_QRT}
An important problem in quantum information theory is to determine the amount of resources needed to perform specific communication and information-processing tasks efficiently. The setting of quantum resource theories (QRTs) provides a powerful and unifying framework for addressing the above questions and, more broadly, for analyzing a wide range of phenomena in quantum physics. In this setting, for a task in question, one distinguishes between free resources and costly (or valuable) ones. One also identifies a restricted set of operations, called free operations, which are the physically allowed transformations that leave the set of free resources unchanged (in some exact sense to be specified). For example, in the QRT of bipartite entanglement, free resources are separable states, the valuable ones are entangled states, and free operations are chosen to contain at least local operations and classical communication (LOCC) between the two distant parties (say, Alice and Bob), each possessing a subsystem of the underlying bipartite quantum system. The QRT framework, which started with the consideration of static resources encoded in quantum states, has more recently been extended to the consideration of dynamic resources, namely, quantum channels and measurements \cite{gour_how_2019, }. In any QRT, a question of interest is that of convertibility of one valuable resource to another, under the allowed set of operations, and the determination of whether such a conversion is reversible. The question of reversibility is usually studied in the so-called asymptotic setting, in which an arbitrarily large number of the initial resource is given. The allowed set of operations are required to at least satisfy the property that, in the asymptotic limit, they do not generate a valuable resource when acting on a free resource. This set of operations is also referred to as {\em{asymptotically resource non-generating (ARNG) operations}}. Reversibility of the QRT is then characterized in terms of asymptotic conversion rates. For two given valuable resources $R_1$ and $R_2$, the asymptotic conversion rate  $r(R_1 \to R_2)$ quantifies the number of copies of $R_2$ that can be obtained per copy of $R_1$ with vanishing error in the limit $n \to \infty$, where $n$ denotes the number of identical copies of the resource $R_1$ that one starts with. The QRT is said to be reversible if, for any two valuable resources $R_1$ and $R_2$, $r(R_1 \to R_2) =  r(R_2 \to R_1)^{-1}$. Despite the large degree of freedom in how one defines free resources and free operations, unexpected similarities emerge among different QRTs in terms of resource measures and resource convertibility. 

The convertibility of valuable resources in a QRT can be phrased in terms of a second law, in analogy with the well-known second law of thermodynamics. The thermodynamic version is typically expressed as the principle that the entropy of a closed system never decreases. However, as made precise in the axiomatic formulation of Lieb and Yngvason~\cite{lieb_physics_1999} (see also~\cite{giles_mathematical_1964}), it can equivalently be understood as the existence of a total ordering of equilibrium thermodynamic states, which dictates which state transformations are possible under adiabatic processes: a state $X_1$ can be converted to a state $X_2$ under an adiabatic process if and only if $S(X_1) \leq S(X_2)$, where $S(\cdot)$ denotes the entropy. Thus, the convertibility of the states is characterized by a single function, namely, the entropy.

In the QRT of bipartite entanglement of pure states, a similar result is valid: A state $\ket{\psi_{AB}}$ can be converted to a state $\ket{\phi_{AB}}$ by LOCC asymptotically if and only if $E(\psi) \geq E(\phi)$, where $E(\psi)$ denotes the entropy of entanglement of the bipartite state $\ket{\psi_{AB}}$ , and is defined as the von Neumann entropy of its reduced state. Hence, the entropy of entanglement plays a role similar to that of the thermodynamic entropy. This result is valid in the asymptotic setting mentioned above, in which an arbitrarily large umber of identical copies of $\ket{\phi_{AB}}$ are shared between Alice and Bob.

More generally, a QRT is said to exhibit a second law if the convertibility of its valuable resources under a restricted set of operations is characterized by a single function of its resources. This happens if the asymptotic conversion rate (mentioned above) takes the form $r(R_1 \to R_2)= f(R_1)/f(R_2)$, for some function $f(\cdot)$. The latter is referred to as the {\em{relative entropy of resource}}. In thermodynamics $f$ is the thermodynamic entropy, while in the theory of bipartite pure state entanglement, it is the entropy of entanglement. The restricted set of operations are adiabatic processes in the former and LOCC operations in the latter.

Even though the interconversion of pure bipartite entangled states satisfies the reversibility criterion mentioned above, reversibility under LOCC does not hold for mixed entangled states. In 2010, Brandao and Plenio \cite{brandao_generalization_2010} showed that under the class of ARNG operations (which in the context of entanglement are also called asymptotically non-entangling operations) the theory of manipulation of mixed entangled states becomes reversible. They identified the unique function $f$ characterizing the reversibility under these operations to be the regularized relative entropy of entanglement. The main technical tool developed and employed by Brandao and Plenio to prove this result was the Generalized Quantum Stein's Lemma (GQSL), which, as is evident from its name, is a generalization of the quantum Stein's lemma of quantum hypothesis testing. 
It established that the optimal type-II error exponent for distinguishing multiple copies of an entangled state from an arbitrary sequence of separable states is given by the regularized relative entropy of entanglement.

Even though a gap in the proof of the GQSL was detected in~\cite{berta_gap_2023} in 2023, this issue was recently resolved, with complete proofs provided independently, through two different approaches, by Hayashi and Yamasaki~\cite{hayashi_generalized_2024} and Lami~\cite{lami_solution_2024}.

\section{A Generalized Quantum Stein's Lemma for C-Q Channels}\label{sec:GQSL}
Throughout this section, let $\mathcal{E}\in \CQ$, be a fixed c-q channel. 
We assume that $(\mathcal{S}_n \subset\CPTP(\mathcal{X}^n\rightarrow A^n))_n$ is a sequence of sets of c-q channels with the following properties:
\begin{definition}[Assumptions on the set of free c-q channels]\label{def:axioms}
    \mbox{}
\begin{enumerate}
    \item $\mathcal{S}_n$ is closed and convex as a subset of $ \CPTP(\mathcal{X}^n\rightarrow A^n)$ for each $n\in\mathbb{N}$.
    \item $\mathcal{S}_n$ is closed under permutations of the $n$ inputs and outputs for each $n\in\mathbb{N}$. That is, for every $n \in \naturals$, for each $\F_n \in \S_n$ and every permutation $\pi \in \Sn$ also the permuted channel $\omega \mapsto (\pi\cdot\F_n)(\omega)= P_A(\pi)^\dag \F_n(P_{\X}(\pi) \omega P_{\X}(\pi)^\dag) P_A(\pi)$ is an element of $\S_n$. 
    \item $(\mathcal{S}_n)_n$ is closed as a sequence under tensor products (i.e., for any $m,n\in \mathbb{N}$, if $\mathcal{F}_n\in\mathcal{S}_n$ and $\mathcal{F}_m\in \mathcal{S}_m$, then $\mathcal{F}_n \otimes\mathcal{F}_m\in \mathcal{S}_{n+m}$).
    \item There exists a channel $\mathcal{F}_* \in \mathcal{S}_1$, such that the Choi state of $\F_*$ has full rank. Due to  \zcref{lem:channel_divergence_maximally_entangled_state} this is equivalent to the statement that there exists a constant $C$ such that for every other channel $\E \in \cptp[\X \to A]$: $D_{\max}(\mathcal{E}\|\mathcal{F}_*) \leq C$.
\end{enumerate}
\end{definition}

\begin{remark}
    Note that for the GQSL (\zcref{thm:GQSL}) where the channel $\E$ is fixed, it is in fact sufficient to have assumption (4) in \zcref{def:axioms} only for this particular channel $\E$, i.e., that $D_{\max}(\E\|\S_1) < \infty$. We state the assumptions in \zcref{def:axioms} as written to be able to use the same set of assumptions for the free sets of quantum resource theories of c-q channels later on, which should be independent of any channel $\E$.
\end{remark}

\begin{remark}\label{remark:why_assumption_needed}
Our assumptions (1), (3) and (4) are natural channel analogues of the assumptions considered in the GQSL for quantum states in \cite{hayashi_generalized_2024}. The original argument \cite{brandao_generalization_2010}, and also the argument in \cite{lami_solution_2024} require further assumptions, which also include an analogue of our assumption (2). Our proof broadly follows the argument in \cite{hayashi_generalized_2024}, however we still also require assumption (2). This is because a key step in our proof is the use of the pinching inequality (in particular \zcref{lem:pinching_bound}), for which we require the number of spectral points (of the state we pinch with respect to) to be sub-exponential.  We guarantee this by using an argument involving permutation invariance which requires assumption (2).
In \cite{hayashi_generalized_2024} the authors use a different approach to bound the number of spectral points: They show that assumption (4) can be used to guarantee that certain states that they construct always have their minimal eigenvalues bounded from below, and they use this to show that a certain spectral bunching operation does not disturb the states too much. In the approach we pick to prove our GQSL using a similar approach is tricky. This is because we often use the channel together with a reference system, and for every channel $\F$, the smallest eigenvalue of $(\id_{R} \otimes \F)(\nu_{R\X})$ can be arbitrarily close to zero, for an appropriate choice of $\nu_{R\X}$. We do not consider the additional assumption (2) to be very restrictive, as it is satisfied by most interesting sets of free channels. Moreover, as mentioned above, a similar assumption was also included in the original argument of \cite{brandao_generalization_2010}.
\end{remark}




The channel version of the Generalized Quantum Stein's Lemma is very similar to the state version (see \zcref{sec:composite_HT}).
It relates an operational quantity, concretely the optimal asymptotic type-II error exponent of the composite hypothesis testing problem of discriminating the simple IID channel hypothesis $\E\n$ against the composite hypothesis $\S_n$, to an entropic expression. The latter is given by a regularization of the Umegaki channel divergence when minimized over the set $\S_n$:

\begin{theorem}[Generalized Quantum Stein's Lemma for Classical-Quantum Channels]\label{thm:GQSL}
Let $\varepsilon\in(0,1)$, $\E\in \CQ[\X \to A]$ be a c-q channel and let $(\mathcal{S}_n)_n$ be a sequence of c-q channels satisfying (1)-(4) in \Autoref{def:axioms}. Then, it holds that
\begin{equation}\label{eq:GQSL}
\lim_{n \to \infty} {1 \over n} D_H^{ε}(\E\n\|\S_n) = \lim_{n \to \infty} {1 \over n} D(\E\n\|\S_n).
\end{equation}
Writing out the shorthand notation (defined in \eqref{eq:shorthand_composite_channels}), this is equivalent to:
\begin{equation}\label{eq:GQSL_explicit}
    \lim_{n\rightarrow\infty}\frac{1}{n}\sup_{\nu_n\in\mathcal{D}(R\mathcal{X}^n)}\inf_{\mathcal{F}_n\in \mathcal{S}_n}{D^{\varepsilon}_H(\mathcal{E}^{\otimes n}(\nu_n)\|\mathcal{F}_n(\nu_n))} \\ =  \lim_{n\rightarrow\infty}\frac{1}{n}\sup_{\nu_n\in\mathcal{D}(R\mathcal{X}^n)}\inf_{\mathcal{F}_n\in \mathcal{S}_n}{D(\mathcal{E}^{\otimes n}(\nu_n)\|\mathcal{F}_n(\nu_n))} 
\end{equation}
\end{theorem}
\begin{remark}
    Note that by \zcref{lem:rel_entropy_exchange} on the right-hand side of \eqref{eq:GQSL_explicit} the supremum and infimum can be exchanged. We are not aware of any way to show a similar exchange for the left-hand side of \eqref{eq:GQSL_explicit}. However, note that, as explained in \zcref{sec:composite_HT_channels}, the left-hand side with the order of supremum and infimum as written, corresponds to the operational quantity of optimal achievable type-II error exponent.
\end{remark}
Note also that a similar Generalized Quantum Stein's Lemma for channels, but with the constraint of input states being IID tensor products (in which case the problem reduces fairly straightforwardly to the GQSL for states) was already studied in \cite{gour_how_2019}.

\begin{EXCLUDED}
In fact, we prove the GQSL for channels not just for c-q channels but for any class of channels for which the strong converse property is known (currently the most important class of this is the class of c-q channels):
\begin{theorem}\label{thm:GQSL_under_strong_converse}
    Let $\E \in \cptp[A \to B]$ be a quantum channel and let $ (\mathcal{S}_n)_n$ be a sequence of general quantum channels satisfying (1)-(4) in \Autoref{def:axioms} such that the strong-converse property holds for simple channel discrimination between $\E\n$ and any element of $\S_n$, that means that for any $n \in \naturals$, for any $\F_n \in \S_n$, and for any $\varepsilon \in (0,1)$ it holds that
    \begin{equation}
        \lim_{m \to \infty} \sup_{\nu \in \DM[RA^{nm}]} {1 \over m} D_H^{ε}(\E^{\otimes mn}(\nu)\|\F_n^{\otimes m}(\nu)) = \lim_{m \to \infty} \sup_{\nu \in \DM[RA^{nm}]} {1 \over m} D(\E^{\otimes mn}(\nu)\|\F_n^{\otimes m}(\nu))\,.
    \end{equation}
    Then, it also holds that
    \begin{multline}\label{eq:GQSL_under_strong_converse}
        \lim_{n \to \infty} {1 \over n} \sup_{ν_n \in \DM[RA^n]} \inf_{\F_n \in \S_n} D^{ε}_H(\E\n(ν_n)\|\F_n(ν_n)) \\ = \lim_{n \to \infty} {1 \over n} \sup_{\nu_n \in \DM[RA^n]} \inf_{\F_n \in \S_n} D(\E\n(ν_n)\|\F_n(ν_n))
    \end{multline}
\end{theorem}
\end{EXCLUDED}

\subsection{Important Lemmas}
Before we prove the above theorem we would like to prove some important lemmas first.

\medskip
We will make use many times of the following {\em{exchange lemma}} found in \cite[Prop. 19]{bergh_composite_2023} (see also \cite[Lemma 4.2.4]{bergh_discrimination_2025}, and \cite[Theorem 2]{gour_how_2019}).
\begin{lemma}[\cite{bergh_composite_2023, bergh_discrimination_2025}]\label{lem:rel_entropy_exchange}
    Let $\S,\T \subset \cptp$ be two closed, convex sets of channels. Let $\D$ be a quantum divergence that satisfies the data-processing inequality, is (jointly) lower semi-continuous, and also satisfies the direct-sum property. Then
		\begin{equation}
			\inf_{\substack{\E \in \S \\ \F \in \T}} \sup_{\nu \in \DM[R A]} \D(\E(\nu)\|\F(\nu)) = \sup_{\nu \in \DM[R A]} \inf_{\substack{\E \in \S \\ \F \in \T}}  \D(\E(\nu)\|\F(\nu)) .
		\end{equation}
\end{lemma}
\medskip

Next, we wish to prove that the limit on the right-hand side of equation \eqref{eq:GQSL} exists. To achieve this, we use \Autoref{lem:divergence_input_reduction} to reduce the supremum to over tensor product states, and then use the same method as \cite[Lemma S4]{hayashi_generalized_2024}. 
We use additivity of the quantum relative entropy to prove the sequence is subadditive, and then use Fekete's subadditive lemma to prove that the limit exists.
\begin{lemma}[Existence of Relative Entropy Limit]
\label{lem:lim_exists}
For any c-q channel $\E\in\CPTP(\mathcal{X}\rightarrow A)$, and any sequence of c-q channels $(\mathcal{S}_n)_n$ which satisfy the assumptions in \Autoref{def:axioms}, the following limit exists:
\begin{equation}
    \lim_{n \to \infty} {1 \over n} D(\E\n\|\S_n) = \lim_{n\rightarrow\infty}\frac{1}{n}\sup_{\nu_n\in\mathcal{D}(R\mathcal{X}^n)}\inf_{\mathcal{F}_n\in \mathcal{S}_n}{D(\mathcal{E}^{\otimes n}(\nu_n)\|\mathcal{F}_n(\nu_n))} 
\end{equation}
\end{lemma}
\begin{proof}
    We first choose $\Tilde{\mathcal{F}_n}\in\mathcal{S}_n$ as an extremizer of 
    \begin{equation}
        \inf_{\mathcal{F}_n\in\mathcal{S}_n} \sup_{x\in\mathcal{X}^n} D(\mathcal{E}(\ketbra{x}{x})\:\|\:\mathcal{F}(\ketbra{x}{x}))
        \label{S4Fchoice}
    \end{equation}
    for each integer $n$. By the assumptions in \Autoref{def:axioms}, we have that $\mathcal{F}_*^{\otimes n}\in\mathcal{S}_n$ and $\mathcal{S}_n$ is compact, so since $D$ is lower semi-continuous (this is
well-known see e.g. \cite{wehrl_general_1978}) and a supremum of lower semi-continuous functions is lower semi-continuous, the infimum must exist and be finite. \\
    \\
    Let $m, n \in\mathbb{N}$. Then,
    \begin{equation}
    \begin{split}
    &\sup_{\nu_{m+n}\in\mathcal{D}(R\mathcal{X}^{m+n})}\inf_{\mathcal{F}_{m+n}\in \mathcal{S}_{m+n}}{D(\mathcal{E}^{\otimes (m+n)}(\nu_{m+n})\|\mathcal{F}_{m+n}(\nu_{m+n}))} \\
    &\leq \sup_{\nu_{m+n}\in\mathcal{D}(R\mathcal{X}^{m+n})} {D(\mathcal{E}^{\otimes (m+n)}(\nu_{m+n})\|(\Tilde{\mathcal{F}}_{n}\otimes\Tilde{\mathcal{F}}_{m})(\nu_{m+n}))} \\
    &= \sup_{x\in\mathcal{X}^{m+n}} {D(\mathcal{E}^{\otimes (m+n)}(\ketbra{x}{x})\|(\Tilde{\mathcal{F}}_{n}\otimes\Tilde{\mathcal{F}}_{m})(\ketbra{x}{x}))}
    \end{split}
    \end{equation}
    where we used \Autoref{lem:divergence_input_reduction} to restrict to states of the form $\ketbra{x}{x}$. \\
    \\
    Now, for each $x\in\mathcal{X}^{m+n}$, we can write $x=x_1x_2$, where $x_1\in\mathcal{X}^n$ and $x_2\in\mathcal{X}^m$ have been concatenated. So, we have
    \begin{eqnarray}
        && \sup_{x\in\mathcal{X}^{m+n}} {D(\mathcal{E}^{\otimes (m+n)}(\ketbra{x}{x})\|(\Tilde{\mathcal{F}}_{n}\otimes\Tilde{\mathcal{F}}_{m})(\ketbra{x}{x}))} \nonumber\\
        & & \,\, \quad = \sup_{x\in\mathcal{X}^{m+n}} {D( \mathcal{E}^{\otimes n}(\ketbra{x_1}{x_1})\otimes \mathcal{E}^{\otimes m}(\ketbra{x_2}{x_2})\|\Tilde{\mathcal{F}}_{n}(\ketbra{x_1}{x_1}) \otimes \Tilde{\mathcal{F}}_{m}(\ketbra{x_2}{x_2})}\nonumber\\ 
        & & \,\, \quad \overset{(a)}{=} \sup_{x\in\mathcal{X}^{m+n}} \left[D( \mathcal{E}^{\otimes n}(\ketbra{x_1}{x_1})\|\Tilde{\mathcal{F}}_{n}(\ketbra{x_1}{x_1}) + D( \mathcal{E}^{\otimes m}(\ketbra{x_2}{x_2})\|\Tilde{\mathcal{F}}_{m}(\ketbra{x_2}{x_2}) \right]\nonumber\\\
        & & \,\, \quad\leq \sup_{x_1\in\mathcal{X}^{n}} D( \mathcal{E}^{\otimes n}(\ketbra{x_1}{x_1})\|\Tilde{\mathcal{F}}_{n}(\ketbra{x_1}{x_1})) + \sup_{x_2\in\mathcal{X}^{m}} D( \mathcal{E}^{\otimes m}(\ketbra{x_2}{x_2})\|\Tilde{\mathcal{F}}_{m}(\ketbra{x_2}{x_2})) \nonumber\\\
        & & \,\,\quad \overset{(b)}{=}\inf_{\mathcal{F}_n\in\mathcal{S}_n} \sup_{x_1\in\mathcal{X}^{n}} D( \mathcal{E}^{\otimes n}(\ketbra{x_1}{x_1})\|\mathcal{F}_n(\ketbra{x_1}{x_1}))
        \nonumber\\ & & \,\,\quad \quad 
        \quad + \inf_{\mathcal{F}_m\in\mathcal{S}_m} \sup_{x_2\in\mathcal{X}^{m}} D( \mathcal{E}^{\otimes m}(\ketbra{x_2}{x_2})\|\mathcal{F}_m(\ketbra{x_2}{x_2}))      
    \end{eqnarray}
    We used the additivity of the quantum relative entropy under tensor products in (a), and by the choice of $\Tilde{\mathcal{F}}_n$ in \eqref{S4Fchoice} we get (b). \\
    \\
    Now we use \Autoref{lem:divergence_input_reduction} to return the supremum to being over all states in $\mathcal{D}(R\mathcal{X}^n)$, and use \Autoref{lem:rel_entropy_exchange} to swap the order of the supremum and infimum to get
    \begin{multline}
        \sup_{\nu_{m+n}\in\mathcal{D}(R\mathcal{X}^{m+n})}\inf_{\mathcal{F}_{m+n}\in \mathcal{S}_{m+n}}{D(\mathcal{E}^{\otimes (m+n)}(\nu_{m+n})\|\mathcal{F}_{m+n}(\nu_{m+n}))}
       \\ \leq \sup_{\nu_{n}\in\mathcal{D}(R\mathcal{X}^{n})}\inf_{\mathcal{F}_{n}\in \mathcal{S}_{n}}{D(\mathcal{E}^{\otimes n}(\nu_{n})\|\mathcal{F}_{n}(\nu_{n}))}
         + \sup_{\nu_{m}\in\mathcal{D}(R\mathcal{X}^{m})}\inf_{\mathcal{F}_{m}\in \mathcal{S}_{m}}{D(\mathcal{E}^{\otimes m}(\nu_{m})\|\mathcal{F}_{m}(\nu_{m}))}
    \end{multline}
    Hence, the sequence is subadditive, so we can use Fekete's subadditive lemma to get that the limit exists.
\end{proof}
We now use \Autoref{lem:lim_exists} to break down the GQSL for c-q channels into two statements that we then prove separately. This proposition is analogous to \cite[Proposition S2]{hayashi_generalized_2024}.
\begin{proposition}
Suppose for any $\varepsilon\in(0,1)$, any c-q channel $\E\in\CPTP(\X\rightarrow A)$, and any sequence $(\mathcal{S}_n)_n$ satisfying the conditions in \Autoref{def:axioms}, that the following relations hold: \\ 
\\
\textbf{Strong Converse:}
\begin{equation}
    \limsup_{n \to \infty} {1 \over n} D_H^{ε}(\E\n\|\S_n) \leq \lim_{n \to \infty} {1 \over n} D(\E\n\|\S_n)
\end{equation}
\textbf{Direct:}
\begin{equation}
    \liminf_{n \to \infty} {1 \over n} D_H^{ε}(\E\n\|\S_n) \geq \lim_{n \to \infty} {1 \over n} D(\E\n\|\S_n)
\end{equation}
Then the Generalized Quantum Stein's Lemma for c-q channels (\Autoref{thm:GQSL}) holds.
\end{proposition}
\begin{proof}
    By \Autoref{lem:lim_exists}, the limit on the right-hand side of both inequalities exists. \\
    \\
    Now, for any sequence of real numbers $(a_n)_n$, we have that $\liminf_{n\rightarrow\infty}a_n\leq\limsup_{n\rightarrow\infty}a_n$. Hence, we get from the combination of both inequalities that the limit
    \begin{equation}
        \lim_{n \to \infty} {1 \over n} D_H^{ε}(\E\n\|\S_n)
    \end{equation}
    exists and is equal to 
    \begin{equation}
        \lim_{n \to \infty} {1 \over n} D(\E\n\|\S_n)
    \end{equation}
    as required.
\end{proof}
\subsection{Strong Converse} \label{ch_strong_converse}
First, we wish to bound the hypothesis testing relative entropy from above by the quantum relative entropy. The next lemma, which is based on \cite[Lemma S6]{hayashi_generalized_2024}, allows us to do this for a specific sequence of channels in $(\mathcal{S}_n)_n$. Then, by using a set of channels extremizing $\inf_{\F_n\in\mathcal{S}_n}D(\E\n \| \F_n)$, and applying \Autoref{lem:GQSL_converse_bound}, we obtain the strong converse in \Autoref{lem:GQSL_strong_converse}. 
\begin{lemma}
\label{lem:GQSL_converse_bound}
Fix any $m\in\mathbb{N}$, and $\Tilde{\mathcal{F}}_m\in\mathcal{S}_m$. With $\F_*$ the channel with full-rank Choi state from \zcref{def:axioms}, define $\hat{\mathcal{F}}_n:=\Tilde{\F}_m^{\otimes l}\otimes \mathcal{F}_*^{\otimes r}$, where $n=ml+r$, with $l,r\in\mathbb{Z}$, and $0\leq r<n$. 

Then, for any $\varepsilon\in[0,1)$, we have
\begin{equation}
    \limsup_{n\rightarrow\infty}\frac{1}{n}\sup_{\nu_n\in\mathcal{D}(R\mathcal{X}^n)}{D_H^{\varepsilon}(\mathcal{E}^{\otimes n}(\nu_n)\|\hat{\mathcal{F}}_n(\nu_n))}
    \leq \frac{1}{m}\sup_{\nu_m\in\mathcal{D}(R\mathcal{X}^m)}{D(\mathcal{E}^{\otimes m}(\nu_m)\|\Tilde{\mathcal{F}}_m(\nu_m))}\,.
\end{equation}
\end{lemma}
\begin{proof}
    The idea used in the proof is that, asymptotically, we can consider $D(\E\n\|\hat{\mathcal{F}}_n)$ to be the same as $D(\E^{\otimes ml}\|\Tilde{\mathcal{F}}_m^{\otimes l})$, because the extra part $D(\E^{\otimes r}\|\F_*^{\otimes r})$ is $O(1)$ as $n\rightarrow\infty$, so it does not contribute to the final expression.\\
    \\
    Note that by assumptions (3) and (4) in \Autoref{def:axioms}, we have $\hat{\mathcal{F}}_n\in\mathcal{S}_n$ for all $n$.
    We use \cite[Proposition 7.71]{khatri_principles_2020} which states that any two states $\rho, \sigma$ and any $α > 1$:
    \begin{equation}
        D^{\varepsilon}_H(\rho \|\sigma)\leq \Ds(\rho \|\sigma) + \frac{\alpha}{\alpha-1}\log\left(\frac{1}{1-\varepsilon}\right)
    \end{equation}
    Using this, we get,
    \begin{eqnarray}
        &&\sup_{\nu_n\in\mathcal{D}(R\mathcal{X}^n)}{D_H^{\varepsilon}(\mathcal{E}^{\otimes n}(\nu_n)\|\hat{\mathcal{F}}_n(\nu_n))} \\
        &\leq&\sup_{\nu_n\in\mathcal{D}(R\mathcal{X}^n)}{\Ds(\mathcal{E}^{\otimes n}(\nu_n)\|\hat{\mathcal{F}}_n(\nu_n))} + \frac{\alpha}{\alpha-1}\log\left(\frac{1}{1-\varepsilon}\right) \\
        &=&\sup_{x\in\mathcal{X}^n}{\Ds(\mathcal{E}^{\otimes n}(\ketbra{x}{x})\|\hat{\mathcal{F}}_n(\ketbra{x}{x}))} + \frac{\alpha}{\alpha-1}\log\left(\frac{1}{1-\varepsilon}\right)
        \label{tensor_product_channels_samepoint}
    \end{eqnarray}
    by using \Autoref{lem:divergence_input_reduction}. We can do this since by e.g.\ \cite[Chapter 7.5]{khatri_principles_2020}, the sandwiched \renyi relative entropy $\Ds$ is additive under tensor products and faithful, and is also jointly quasi-convex for $\alpha>1$. \\
    \\ 
    Now, choose $x_*\in\mathcal{X}^n$ to extremize 
    \begin{equation}
        \sup_{x\in\mathcal{X}^n}{\Ds(\mathcal{E}^{\otimes n}(\ketbra{x}{x})\|\hat{\mathcal{F}}_n(\ketbra{x}{x}))}
    \end{equation}
    We can write $\ketbra{x_*}{x_*}$ in the form
    \begin{eqnarray}
        \ketbra{x_*}{x_*}=\left(\bigotimes_{i=1}^l{\ketbra{y_i}{y_i}}\right)\otimes\bigotimes_{j=1}^n\ketbra{z_j}{z_j}
    \end{eqnarray}
    where $y_i\in\mathcal{X}^m$ and $z_j\in\mathcal{X}$. With this we find:
    \begin{align}
        \sup_{x\in\mathcal{X}^n}&{\Ds(\mathcal{E}^{\otimes n}(\ketbra{x}{x})\|\hat{\mathcal{F}}_n(\ketbra{x}{x}))} \\
        &=\Ds(\mathcal{E}^{\otimes n}(\ketbra{x_*}{x_*})\|\hat{\mathcal{F}}_n(\ketbra{x_*}{x_*}))\\
        &= \sum_{i=1}^l\Ds(\mathcal{E}^{\otimes m}(\ketbra{y_i}{y_i})\|\Tilde{\mathcal{F}}_m(\ketbra{y_i}{y_i}))
        +\sum_{j=1}^r\Ds(\mathcal{E}(\ketbra{z_j}{z_j})\|\mathcal{F}_*(\ketbra{z_j}{z_j})) \\
        &\leq l \sup_{x\in\mathcal{X}^m}\Ds(\mathcal{E}^{\otimes m}(\ketbra{x}{x})\|\Tilde{\mathcal{F}}_m(\ketbra{x}{x})) 
        + r \sup_{x\in\mathcal{X}}\Ds(\mathcal{E}(\ketbra{x}{x})\|\mathcal{F}_*(\ketbra{x}{x}))
    \end{align}
    where we have used the additivity of $\Ds$ under tensor products. \\
    \\
    Due to assumption (4) in \Autoref{def:axioms}, $D_{\max}(\mathcal{E}(\ketbra{x}{x})\| \mathcal{F}_*(\ketbra{x}{x}))<\infty$ for all $x\in\mathcal{X}$, so the second term is finite. This means that, when we take $\limsup_{n\rightarrow\infty}\frac{1}{n}$ of both sides, the second term is zero asymptotically, so we obtain 
    \begin{eqnarray}
        \limsup_{n\rightarrow\infty}\frac{1}{n}\sup_{\nu_n\in\mathcal{D}(R\mathcal{X}^n)}{D_H^{\varepsilon}(\mathcal{E}^{\otimes n}(\nu_n)\|\hat{\mathcal{F}}_n(\nu_n))} 
        \leq\frac{1}{m}\sup_{x\in\mathcal{X}^m}\Ds(\mathcal{E}^{\otimes m}(\ketbra{x}{x})\|\Tilde{\mathcal{F}}_m(\ketbra{x}{x}))
    \end{eqnarray}
    also noting that $\frac{\alpha}{n(\alpha-1)}\log(\frac{1}{1-\varepsilon})=o(1)$ as $n\rightarrow\infty$. \\
    \\
    Now, we take the limit $\alpha \rightarrow1^+$. Since the supremum is over a finite set, we can exchange the limit and supremum. This gives
    \begin{eqnarray}
        \limsup_{n\rightarrow\infty}\frac{1}{n}\sup_{\nu_n\in\mathcal{D}(R\mathcal{X}^n)}{D_H^{\varepsilon}(\mathcal{E}^{\otimes n}(\nu_n)\|\hat{\mathcal{F}}_n(\nu_n))}
        &\leq&\frac{1}{m}\sup_{x\in\mathcal{X}^m}D(\mathcal{E}^{\otimes m}(\ketbra{x}{x})\|\Tilde{\mathcal{F}}_m(\ketbra{x}{x})) \\
        &\leq&\frac{1}{m}\sup_{\nu\in\mathcal{D}(R\mathcal{X}^m)}D(\mathcal{E}^{\otimes m}(\nu_m)\|\Tilde{\mathcal{F}}_m(\nu_m))
    \end{eqnarray}
    as required.
\end{proof}
Finally, we prove the strong converse. This proof is based on \cite[Proposition S5]{hayashi_generalized_2024}. Here, the upper bound in \Autoref{lem:GQSL_converse_bound} is used for an extremal choice of $\tilde{\mathcal{F}}_m$ for the quantum relative entropy in equation (\ref{eq:GQSL_converse_F_choice}). Then we take the limit $m\rightarrow\infty$ to get the desired inequality. It is worth noting that this proof does not require the assumption that $\S_n$ is closed under permutations.
\begin{proposition}
\label{lem:GQSL_strong_converse}
Let $\E:\mathcal{X}\rightarrow A$ be a c-q channel, and $(\mathcal{S}_n\subset\CPTP(\mathcal{X}^n\rightarrow A^n))_n$ be a sequence of sets of c-q channels satisfying assumptions (1), (3) and (4) in \Autoref{def:axioms}. Then, the strong converse

\begin{equation}
    \limsup_{n \to \infty} {1 \over n} D_H^{ε}(\E\n\|\S_n) \leq \lim_{n \to \infty} {1 \over n} D(\E\n\|\S_n)
\end{equation}
holds.
\end{proposition}
\begin{proof}
    Fix $m\in\mathbb{N}$, and we choose $\tilde{\mathcal{F}}_k$ for each $k$ to extremise 
    \begin{equation}
        \label{eq:GQSL_converse_F_choice}
        \inf_{\mathcal{F}_k\in\mathcal{S}_k} \sup_{\nu_n\in\mathcal{D}(R\mathcal{X}^k)}D(\mathcal{E}^{\otimes k}(\nu_n)\|\mathcal{F}_k(\nu_n))
    \end{equation}
    This exists as $\mathcal{S}_k$ is compact, the quantum relative entropy is jointly lower semi-continuous (this is well-known see e.g.\ \cite{wehrl_general_1978}), and the supremum of lower semi-continuous functions is lower semi-continuous. \\
    \\
    Again, we define $\hat{\mathcal{F}}_n=\Tilde{\mathcal{F}}_m^{\otimes l}\otimes\mathcal{F}_*^{\otimes r}$ as before. Now, we get
    \begin{eqnarray}
        && \limsup_{n\rightarrow\infty}\frac{1}{n}\sup_{\nu_n\in\mathcal{D}(R\mathcal{X}^n)}\inf_{\mathcal{F}_n\in\mathcal{S}_n}D^{\varepsilon}_H(\mathcal{E}^{\otimes n}(\nu_n)\|\mathcal{F}_n(\nu_n)) \\ 
        &\leq&\limsup_{n\rightarrow\infty}\frac{1}{n}\sup_{\nu_n\in\mathcal{D}(R\mathcal{X}^n)}D^{\varepsilon}_H(\mathcal{E}^{\otimes n}(\nu_n)\|\hat{\mathcal{F}}_n(\nu_n))
        \\
        &\leq&\frac{1}{m} \sup_{\nu_m\in\mathcal{D}(R\mathcal{X}^m)}D(\mathcal{E}^{\otimes m}(\nu_m)\|\Tilde{\mathcal{F}}_m(\nu_m))
        \label{S5Bound1}
    \end{eqnarray}
    by \Autoref{lem:GQSL_converse_bound}. \\
    \\
    Additionally,
    \begin{eqnarray}
    \label{S5startineq}
        && \sup_{\nu_m\in\mathcal{D}(R\mathcal{X}^m)}\inf_{\mathcal{F}_m\in\mathcal{S}_m}D(\mathcal{E}^{\otimes m}(\nu_m)\|\mathcal{F}_m(\nu_m)) \\
        &\leq& \sup_{\nu_m\in\mathcal{D}(R\mathcal{X}^m)}D(\mathcal{E}^{\otimes m}(\nu_m)\|\Tilde{\mathcal{F}}_m(\nu_m)) \\
        &\overset{(a)}{=}& \inf_{\mathcal{F}_m\in\mathcal{S}_m}\sup_{\nu_m\in\mathcal{D}(R\mathcal{X}^m)}D(\mathcal{E}^{\otimes m}(\nu_m)\|\mathcal{F}_m(\nu_m)) \\
        &\overset{(b)}{=}& \sup_{\nu_m\in\mathcal{D}(R\mathcal{X}^m)}\inf_{\mathcal{F}_m\in\mathcal{S}_m}D(\mathcal{E}^{\otimes m}(\nu_m)\|\mathcal{F}_m(\nu_m))
        \label{S5endineq}
    \end{eqnarray}
    The equality in (a) is by the choice of $\tilde{\mathcal{F}}_m$, and the equality  in (b) is by \Autoref{lem:rel_entropy_exchange}. \\
    \\
    This means that all the inequalities in equations (\ref{S5startineq}) to (\ref{S5endineq}) are actually equalities. Hence,
    \begin{equation}
        \frac{1}{m} \sup_{\nu_m\in\mathcal{D}(R\mathcal{X}^m)}D(\mathcal{E}^{\otimes m}(\nu_m)\|\Tilde{\mathcal{F}}_m(\nu_m)) \\
        =\frac{1}{m} \sup_{\nu_m\in\mathcal{D}(R\mathcal{X}^m)}\inf_{\mathcal{F}_m\in\mathcal{S}_m}D(\mathcal{E}^{\otimes m}(\nu_m)\|\mathcal{F}_m(\nu_m))
        \label{S5Bound2}
    \end{equation}
    Combining equations (\ref{S5Bound1}) and (\ref{S5Bound2}) and taking the limit as $m\rightarrow\infty$ gives
    \begin{eqnarray}    
    && \limsup_{n\rightarrow\infty}\frac{1}{n}\sup_{\nu_n\in\mathcal{D}(R\mathcal{X}^n)}\inf_{\mathcal{F}_n\in\mathcal{S}_n}D^{\varepsilon}_H(\mathcal{E}^{\otimes n}(\nu_n)\|\mathcal{F}_n(\nu_n)) \\
    &\leq& \lim_{m\rightarrow\infty}\frac{1}{m} \sup_{\nu_m\in\mathcal{D}(R\mathcal{X}^m)}\inf_{\mathcal{F}_m\in\mathcal{S}_m}D(\mathcal{E}^{\otimes m}(\nu_m)\|\mathcal{F}_m(\nu_m))
    \label{S5end}
    \end{eqnarray}
    where the limit in (\ref{S5end}) exists by \Autoref{lem:lim_exists}.
\end{proof}
\subsection{Direct Part}
We start by giving a brief overview of the proof of the direct part.
The main construction (which we show in \Autoref{lem:main_iteration}), is that, given a sequence $(\F_n\in\mathcal{S}_n)_n$, we can find another sequence $(\F_n'\in\mathcal{S}_n)_n$ of channels satisfying
\begin{align}
    \liminf_{n\rightarrow\infty}\frac{1}{n}D(\mathcal{E}^{\otimes n}\|\mathcal{F}_n') -\liminf_{n\rightarrow\infty}\frac{1}{n}D_H^\varepsilon(\mathcal{E}^{\otimes n}\|\mathcal{S}_n) &\leq (1-\Tilde{\varepsilon})\left\{\liminf_{n\rightarrow\infty}\frac{1}{n}D(\mathcal{E}^{\otimes n}\|\mathcal{F}_n)-\liminf_{n\rightarrow\infty}\frac{1}{n}D_H^\varepsilon(\mathcal{E}^{\otimes n}\|\mathcal{S}_n)\right\}
\end{align}
This construction can be iterated to have the right-hand side go to zero, which gives a sequence of channels $(\F_{n,*})_n$ satisfying
\begin{eqnarray}
    \liminf_{n\rightarrow\infty}\frac{1}{n}D(\mathcal{E}^{\otimes n}\|\mathcal{F}_{n,*}) \leq\liminf_{n\rightarrow\infty}\frac{1}{n}D_H^\varepsilon(\mathcal{E}^{\otimes n}\|\mathcal{S}_n)\,,
\end{eqnarray}
which then proves the direct part. The tricky part of this construction is how to deal with the supremum over input states. We tackle this by deferring this supremum to the very end, i.e. we start by fixing a sequence of input states, then carry out our construction for this fixed sequence and take the supremum over all sequences of input states at the end. The fact that the constructed $\F_n'$ depend on this sequence of input states is not a problem because of \zcref{lem:rel_entropy_exchange}. 
Furthermore, we show that we can restrict to input states which are permutation invariant, and channels which are permutation covariant, which allows us to use the pinching inequality (and in particular \zcref{lem:pinching_bound}), even though the input (and thus also output) states we consider can be highly non-IID.
\\
\\
To perform the mentioned last step of taking the supremum over all sequences of input states and commuting this supremum with the limit, we will use the following Lemma:
\begin{lemma}
\label{lem:sup_over_sequences}
Let $(T_n)_n$ be a sequence of arbitrary sets, and define $D:=\{(n,x)\, |\, n\in\mathbb{N}, x\in T_n\}$. Let $\hat{T}$ denote the set of all sequences $(x_n)_n$, with $x_n\in T_n$ for every $n$. \\
\\
Let $f:D\rightarrow \mathbb{R}$ be a function, such that $\sup_{x_n \in T_n}f(n, x_n) < \infty$ for all $n \in \naturals$. Then
\begin{eqnarray}
    \sup_{(x_n)_n\in\hat{T}} \liminf_{n\rightarrow\infty}f(n,x_n) = \liminf_{n\rightarrow\infty} \sup_{x_n\in T_n}f(n,x_n)\,.
\end{eqnarray}
The analogous identity holds if the $\liminf$ is replaced with a $\limsup$ on both sides.
\end{lemma}
\begin{proof}
We will prove the statement for the $\liminf$ but the proof for the $\limsup$ statement is completely identical. 
    Fix $\varepsilon>0$. For each $n\in\mathbb{N}$, choose $x_n^* \in T_n$ such that
    \begin{eqnarray}
        \sup_{x_n\in T_n}f(n,x_n) \leq f(n,x_n^*)+\varepsilon
    \end{eqnarray}
    Then, we have
    \begin{eqnarray}
        \liminf_{n\rightarrow\infty} \sup_{x_n\in T_n}f(n,x_n) &\leq& \varepsilon+\liminf_{n\rightarrow\infty} f(n,x_n^*) \\
        &\leq& \varepsilon+ \sup_{(x_n)_n\in\hat{T}}\liminf_{n\rightarrow\infty} f(n,x_n)
    \end{eqnarray}
    Now, $\varepsilon$ is arbitrary, so we can take the limit $\varepsilon\rightarrow 0$ to get
    \begin{eqnarray}
        \liminf_{n\rightarrow\infty} \sup_{x_n\in T_n}f(n,x_n) &\leq&\sup_{(x_n)_n\in\hat{T}}\liminf_{n\rightarrow\infty} f(n,x_n)
    \end{eqnarray}
    Also, we have, for any sequence $(\Tilde{x}_n)_n\in \hat{T}$,
    \begin{eqnarray}
        \liminf_{n\rightarrow\infty} f(n,\Tilde{x}_n) &\leq&\liminf_{n\rightarrow\infty} \sup_{x_n\in T_n}f(n,x_n) 
    \end{eqnarray}
    Hence, 
    \begin{eqnarray}
        \sup_{(x_n)_n\in\hat{T}}\liminf_{n\rightarrow\infty} f(n,{x}_n) &\leq&\liminf_{n\rightarrow\infty} \sup_{x_n\in T_n}f(n,x_n) 
    \end{eqnarray}
    which gives the required result.
\end{proof}
Now, we will prove some technical lemmas on permutation covariant channels that we will use to prove the direct part. They allow us to deal only with permutation invariant states or permutation invariant channels. 
\getkeytheorem{stored:lem:perm_inputs}

\Autoref{lem:perm_inputs} and its proof are almost exactly the same as \cite[Lemma 24]{bergh_composite_2023} (see also \cite[Lemma 2.6.1]{bergh_discrimination_2025}), and hence the proof is given in \zcref{Appendix}. The difference between our statement and the statement in \cite{bergh_discrimination_2025} is that in our statement we additionally have infima over sets of permutation covariant channels on both sides of the equality. Nevertheless,   the proof is essentially identical.

\bigskip
The next lemma we prove here is a version of \cite[Lemma 23]{bergh_composite_2023} (see also \cite[Lemma 4.2.5]{bergh_discrimination_2025}), which is useful for the following reason: In the main achievability argument (\Autoref{lem:main_iteration} below) we will encounter expressions of the form
\begin{eqnarray}
    \inf_{\substack{\mathcal{E}\in\mathcal{S} \\ \mathcal{F}\in\mathcal{T}}} D(\mathcal{E}(\nu)\|\mathcal{F}(\nu))
\end{eqnarray}
where $\S$ and $\T$ contain channels that need not be permutation covariant. We show that if the input state is permutation invariant, we can restrict these optimizations to be over the subsets of permutation covariant channels.
\begin{lemma}
\label{lem:perm_cov_channels}
Let $\nu\in\mathcal{D}(A^n)$ be a permutation invariant state. Let $\mathcal{S}$, $\mathcal{T}$ be two sets of quantum channels in $\cptp[A^n\rightarrow B^n]$ that are closed under permutations (in the sense of (2) of \Autoref{def:axioms}), convex, and compact. Then
\begin{eqnarray}
    \inf_{\substack{\mathcal{E}\in\mathcal{S} \\ \mathcal{F}\in\mathcal{T}}} D(\mathcal{E}(\nu)\|\mathcal{F}(\nu)) =\inf_{\substack{\mathcal{E}\in\mathcal{S} \\ \mathcal{F}\in\mathcal{T} \\ \mathcal{E},\mathcal{F}\,perm.\,covariant}} D(\mathcal{E}(\nu)\|\mathcal{F}(\nu))
    \label{eq:GQSL_channel_perm_cov}
\end{eqnarray}
\end{lemma}
\begin{proof}
    The idea of the proof is to create new channels by averaging over the orbit of the action of $\Sn$ on the set of quantum channels, and showing that the quantum  relative entropy can only decrease under such an averaging.
    Let $\mathcal{E}\in\mathcal{S}$, $\mathcal{F}\in\mathcal{T}$ be arbitrary channels. Now, by closure under permutations and convexity, there exist channels  $\Bar{\mathcal{E}}\in\mathcal{S}$, $\Bar{\mathcal{F}}\in\mathcal{T}$, where
    \begin{eqnarray}
        \Bar{\mathcal{E}}(\rho)=\frac{1}{n!}\sum_{\pi\in \Sn}P_B(\pi)^\dag\mathcal{E}(P_A(\pi)\rho P_A(\pi)^\dag)P_B(\pi) \\
        \Bar{\mathcal{F}}(\rho)=\frac{1}{n!}\sum_{\pi\in \Sn}P_B(\pi)^\dag\mathcal{F}(P_A(\pi)\rho P_A(\pi)^\dag)P_B(\pi)
    \end{eqnarray}
   
   We can see these two new channels as the result of a permutation superchannel having acted on each channel, and it is known that such superchannels can only decrease any relative entropy where one optimizes over all input states (see e.g. \cite{gour_comparison_2019}). For our desired statement we do not have such an optimization over input states, but we do require our input state to be permutation invariant, and this will turn out to be enough. To see this, define $\mathcal{A}: B^n\rightarrow B^n$
    \begin{eqnarray}
        \mathcal{A}(\rho) := \frac{1}{n!}\sum_{\pi\in \Sn}P_B(\pi)^\dag\rho P_B(\pi)
    \end{eqnarray}
    Now, $\mathcal{A}$ is a CPTP map and so, by the data processing inequality,
    \begin{eqnarray}
        D(\mathcal{A}(\mathcal{E}(\nu))\|\mathcal{A}(\mathcal{F}(\nu))) \leq D(\mathcal{E}(\nu)\|\mathcal{F}(\nu))\,.
    \end{eqnarray}
    We also have
    \begin{eqnarray}
        \mathcal{A}(\mathcal{E}(\nu))&=&\frac{1}{n!}\sum_{\pi\in \Sn}P_B(\pi)^\dag\mathcal{E}(\nu) P_B(\pi) \\
        &=& \Bar{\mathcal{E}}(P_A(\pi)^\dag \nu P_A(\pi)) \\
        &=& \Bar{\mathcal{E}}(\nu)
    \end{eqnarray}
    using the permutation invariance of $\nu$. We get the analogous result for $\mathcal{A}(\mathcal{F}(\nu))$ in the same way. Hence, 
    \begin{eqnarray}
        D(\Bar{\mathcal{E}}(\nu)\|\Bar{\mathcal{F}}(\nu)) \leq D(\mathcal{E}(\nu)\|\mathcal{F}(\nu))
    \end{eqnarray}
    Now, since $\Bar{\mathcal{E}}\in\mathcal{S}$ and $\Bar{\mathcal{F}}\in\mathcal{T}$, this means that we can restrict the infimum on the left-hand side of (\ref{eq:GQSL_channel_perm_cov}) to being over permutation covariant channels. 
\end{proof}
\bigskip
As a last preliminary step before we come to the main part of the achievability argument, we require an analogue of \Autoref{lem:GQSL_converse_bound} (the strong converse-like bound) in which $\Hat{\F}_n$ is additionally averaged over all permutations.

\begin{lemma}[A Strong Converse-like Bound with a Permuted Second Argument]
Let $\E \in \CPTP(\X \to A)$ be a c-q-channel and let $(\S_n)_n$ be a sequence of sets of c-q channels that satisfy the axioms of \Autoref{def:axioms}.
\label{lem:perm_averaging_converse}
Fix any $m\in\mathbb{N}$, and $\Tilde{\mathcal{F}}_m\in\mathcal{S}_m$. \\
\\
Define $\hat{\mathcal{F}}_n:=\Tilde{\mathcal{F}}^{\otimes l}\otimes \mathcal{F}_*^{\otimes r}$, where $n=ml+r$, with $l,r\in\mathbb{Z}$, and $0\leq r<n$ (remember that $\mathcal{F}_*$ is defined to satisfy $D_{\max}(\E \| \F_*)<\infty$ in \Autoref{def:axioms}). Then, for any input state $\tau$,
\begin{eqnarray}
    \hat{\mathcal{F}_n'}(\tau):=\frac{1}{n!}\sum_{\pi\in \Sn}(\pi\cdot\hat{\mathcal{F}}_n)(\tau) \,.
\end{eqnarray}
For any $\varepsilon\in[0,1)$, we have
\begin{equation}
    \limsup_{n\rightarrow\infty}\frac{1}{n}\sup_{\nu_n\in\mathcal{D}(R\mathcal{X}^n)}{D_H^{\varepsilon}(\mathcal{E}^{\otimes n}(\nu_n)\|\hat{\mathcal{F}}_n'(\nu_n))}
    \leq \frac{1}{m}\sup_{\nu_m\in\mathcal{D}(R\mathcal{X}^m)}{D(\mathcal{E}^{\otimes m}(\nu_m)\|\Tilde{\mathcal{F}}_m(\nu_m))}
\end{equation}
\end{lemma}
\begin{proof}
    First of all note that, by convexity, closure under tensor products, and closure under permutations in \Autoref{def:axioms}, we have $\hat{\mathcal{F}}_n'\in\mathcal{S}_n$ for all $n$.
    We use the following inequality (see e.g. \cite{khatri_principles_2020}):
    \begin{equation}
        D^{\varepsilon}_H(\rho \|\sigma)\leq \Ds(\rho \|\sigma) + \frac{\alpha}{\alpha-1}\log\left(\frac{1}{1-\varepsilon}\right)
    \end{equation}
    which is satisfied for any two states $\rho, \sigma$. Using this, we get, for $\alpha>1$ (using the notation defined in Section \ref{ch:perm}) 
    \begin{eqnarray}
        &&\sup_{\nu_n\in\mathcal{D}(R\mathcal{X}^n)}{D_H^{\varepsilon}(\mathcal{E}^{\otimes n}(\nu_n)\|\hat{\mathcal{F}}'_n(\nu_n))} \nonumber\\
        &\leq & \sup_{\nu_n\in\mathcal{D}(R\mathcal{X}^n)}{\Ds(\mathcal{E}^{\otimes n}(\nu_n)\|\hat{\mathcal{F}}_n'(\nu_n))} + \frac{\alpha}{\alpha-1}\log\left(\frac{1}{1-\varepsilon}\right)\nonumber \\
        &\leq&\sup_{\nu_n\in\mathcal{D}(R\mathcal{X}^n)} \max_{\pi\in\Sn} \Ds(\mE\| \mF) + \frac{\alpha}{\alpha-1}\log\left(\frac{1}{1-\varepsilon}\right)\nonumber\\&&
    \end{eqnarray}
    since the sandwiched Rényi relative entropy $\Ds$ is jointly quasi-convex for $\alpha>1$, and $\mathcal{E}^{\otimes n}$ is permutation covariant. \\
    \\
    Now, the sandwiched Rényi relative entropy is invariant under the action of unitaries (this can be seen directly by the definition, or for example also follows from the data-processing inequality), so we obtain
    \begin{eqnarray}
       && \sup_{\nu_n\in\mathcal{D}(R\mathcal{X}^n)} \max_{\pi\in\Sn} \Ds(\mE\| \mF) + \frac{\alpha}{\alpha-1}\log\left(\frac{1}{1-\varepsilon}\right)\nonumber\\
        && \quad =\sup_{\nu_n\in\mathcal{D}(R\mathcal{X}^n)} \max_{\pi\in\Sn} \Ds(\mathcal{E}^{\otimes n}(P_{\mathcal{X}}(\pi)\nu_n P_{\mathcal{X}}(\pi)^\dag)\|\hat{\mathcal{F}}_n(P_{\mathcal{X}}(\pi)\nu_n P_{\mathcal{X}}(\pi)^\dag))\nonumber\\
        && \quad  \quad \quad + \,\, \frac{\alpha}{\alpha-1}\log\left(\frac{1}{1-\varepsilon}\right) \\
        && \quad \overset{(a)}{=}\sup_{\nu_n\in\mathcal{D}(R\mathcal{X}^n)}\Ds(\mathcal{E}^{\otimes n}(\nu_n)\|\hat{\mathcal{F}}_n(\nu_n)) + \frac{\alpha}{\alpha-1}\log\left(\frac{1}{1-\varepsilon}\right) \\
        && \quad \overset{(b)}{=}\sup_{x\in\mathcal{X}^n}{\Ds(\mathcal{E}^{\otimes n}(\ketbra{x}{x})\|\hat{\mathcal{F}}_n(\ketbra{x}{x}))} + \frac{\alpha}{\alpha-1}\log\left(\frac{1}{1-\varepsilon}\right)
    \end{eqnarray}
    Here, we use that $\mathcal{D}(R\mathcal{X}^n)$ is closed under permuting the copies of $\mathcal{X}$ in (a). Also, we use \Autoref{lem:divergence_input_reduction} in (b). We can do this since by e.g.\ \cite{khatri_principles_2020}, the sandwiched Renyi relative entropy $\Ds$ is additive under tensor products, faithful, and is jointly quasi-convex for $\alpha>1$. \\
    \\ 
    The rest of the proof follows exactly the same steps as \Autoref{lem:GQSL_converse_bound}, from equation (\ref{tensor_product_channels_samepoint}) onwards.
\end{proof}
\noindent Finally, we get to the main part of the achivability proof:
\begin{lemma}
\label{lem:main_iteration}
Let $(\mathcal{F}_n)_n$ be a sequence of channels with $\mathcal{F}_n\in\mathcal{S}_n$, and let $(\mathcal{S}_n)_n$ satisfy the assumptions in \Autoref{def:axioms}. \\
For $\varepsilon\in (0,1)$, let 
\begin{align}
    R_2 & \coloneqq\liminf_{n\rightarrow\infty}\frac{1}{n}\sup_{\nu_n\in\mathcal{D}(R\mathcal{X}^n)}D(\mathcal{E}^{\otimes n}(\nu_n)\|\mathcal{F}_n(\nu_n)) \\
    R_{1,\varepsilon} &\coloneqq \liminf_{n \to \infty} {1 \over n} D_H^{ε}(\E\n\|\S_n) = \liminf_{n\rightarrow\infty}\frac{1}{n}\sup_{\nu_n\in\mathcal{D}(R\mathcal{X}^n)}\inf_{\Bar{\mathcal{F}}_n\in \mathcal{S}_n}{D^{\varepsilon}_H(\mathcal{E}^{\otimes n}(\nu_n)\|\Bar{\mathcal{F}}_n(\nu_n))}
\end{align}
Suppose that $R_2>R_{1,\varepsilon}$. Then for any fixed $\Tilde{\varepsilon}\in(0,\varepsilon)$, there exists a sequence $(\mathcal{F}_n')_n$ with $\mathcal{F}_n'\in\mathcal{S}_n$ such that 
\begin{eqnarray}
    \liminf_{n\rightarrow\infty}\frac{1}{n}\sup_{\nu_n\in\mathcal{D}(R\mathcal{X}^n)} D(\mathcal{E}^{\otimes n}(\nu_n)\|\mathcal{F}_n'(\nu_n)) -R_{1,\varepsilon} \leq (1-\Tilde{\varepsilon})(R_2-R_{1,\varepsilon})
\end{eqnarray}
\end{lemma}
\begin{proof}
This proof follows the lines of \cite[Lemma S8]{hayashi_generalized_2024}, and the idea goes as follows: We start by fixing a sequence of input states $(\omega_n)_n$ that are permutation invariant. Then, we construct a new sequence of channels $(\F_{n,\omega_n}')_n$ where each $\F_{n, \omega_n}'$ is an average of three different channels. The $\F_{n,\omega_n}'$ depend on $\omega_n$, although this dependence will ultimately not be a problem because of \zcref{lem:rel_entropy_exchange}. Each $\F_{n,\omega_n}'$ is permutation covariant, and hence we can apply \zcref{lem:pinching_bound} to the pinching map $E_n$ which pinches with respect to the spectral projectors of $\F_{n,\omega_n}(\omega_n)$. We then split $D(E_n(\E\n(\omega_n))\| \F_{n,\omega_n}(\omega_n))$ into three terms, bounds on which are then obtained in terms of relative entropies involving the three channels that $\F_{n,\omega_n}$ is an average of.
Finally, we take a supremum over sequences of states and use \Autoref{lem:sup_over_sequences} to get the supremum over states inside the limit, and get the required bound.

    \noindent To start, fix 
    \begin{eqnarray}
        \varepsilon_0=\frac{\varepsilon-\Tilde{\varepsilon}}{1-\varepsilon}(R_2-R_{1,\varepsilon})\,,
        \label{eq:GQSL_direct_e0_def}
    \end{eqnarray}
    and choose $m$ sufficiently large such that
    \begin{eqnarray}
        \frac{1}{m}\sup_{\nu_m\in\mathcal{D}(R\mathcal{X}^m)}D(\mathcal{E}^{\otimes m}(\nu_m)\|\mathcal{F}_m(\nu_m)) \leq R_2+\varepsilon_0 \,.
        \label{S8b_mbound}
    \end{eqnarray}
    We define $\hat{\mathcal{F}_n'}\in\mathcal{S}_n$ by
    \begin{eqnarray}
        \hat{\mathcal{F}_n'}(\tau):=\frac{1}{n!}\sum_{\pi\in S_n}(\pi\cdot(\mathcal{F}_m^{\otimes l}\otimes\mathcal{F}_*^{\otimes r}))(\tau )
    \end{eqnarray}
    This is an element of $\mathcal{S}_n$ by closure under permutations, closure under tensor products, and convexity. \\
    \\
    Let $K\cong \mathcal{X}$ and $(\omega_n \in\mathcal{D}(K^n\mathcal{X}^n))_n$ be any sequence of input states such that every $\omega_n$ is permutation invariant when permuting the copies of $K\mathcal{X}$. \\
    \\
    Henceforth in this proof, when we talk about permutation invariance or permutation covariance without specifying the systems we permute, we mean with respect to permuting the $n$ copies of $K\mathcal{X}$.\\
\\
    For each $n$, define $\mathcal{F}_n^*$ as the minimizer of 
    \begin{eqnarray}
        \inf_{\Bar{\mathcal{F}}_n\in\mathcal{S}_n} {D^{\varepsilon}_H(\mathcal{E}^{\otimes n}(\omega_n)\|\Bar{\mathcal{F}}_n(\omega_n))}
    \end{eqnarray}
    This exists since $\mathcal{S}_n$ is compact, and $D_H^{\varepsilon}$ is a supremum of continuous functions and hence lower semi-continuous. Since $\omega_n$ is a permutation invariant state, we can assume $\id_{K^n}\otimes\mathcal{F}_n^*$ to be permutation covariant by \Autoref{lem:perm_cov_channels}.\\
    \\
    Furthermore, set
    \begin{equation}
        \mathcal{F}_{n,\omega_n}'
    \coloneqq\frac{1}{3}(\hat{\mathcal{F}}_n'+\mathcal{F}_*^{\otimes n}+\mathcal{F}^*_n)\,.
    \end{equation}
    Assumptions (1), (3), and (4) in \Autoref{def:axioms} then imply that $\F_{n, \omega_n}' \in \S_n$.
    Define
    \begin{eqnarray}
        \sigma_n&=&\hat{\mathcal{F}}_n'(\omega_n) \label{S8Sig_n} \\
        \sigma_{n,full}&=&\mathcal{F}_*^{\otimes n}(\omega_n) \label{S8Sig_full}\\
        \sigma_n^*&=&\mathcal{F}_n^*(\omega_n) \label{S8Sig_star}\\
        \Tilde{\sigma}_n'&=&\mathcal{F}'_{n,\omega_n}(\omega_n) \label{S8Sig_prime} \\
        \rho_n &=& \mathcal{E}^{\otimes n}(\omega_n)
    \end{eqnarray}
    The definition of $\F_{n, \omega_n}'$ means that $\Tilde{\sigma}_n'=\frac{1}{3}(\sigma_n+\sigma_{n,full}+\sigma_n^*)$. \\
    \\
    By construction, each of $\sigma_n$, $\sigma_{n,full}$  and $\sigma_n^*$ is permutation invariant, since they are the output states of a permutation covariant channel with a permutation invariant input. For this, note that $\id_{K^n}\otimes\F_n$ is permutation covariant with respect to permuting the $n$ copies of $K\mathcal{X}$ and $KA$ if and only if $\F_n$ is permutation covariant with respect to permuting the $n$ copies of $\mathcal{X}$ and $A$. In the following argument, this allows us to pinch the input state $\rho_n$ and keep the same asymptotic limit of the quantum relative entropy. \\
    \\
    Let $E_n$ be the pinching map with respect to $\Tilde{\sigma}_n'$. Then, by \Autoref{lem:pinching_bound}, we have that
    \begin{eqnarray}
        D(\rho_n\|\Tilde{\sigma}_n')-D(E_n(\rho_n)\|\Tilde{\sigma}_n')
        \leq \log|\spec(\Tilde{\sigma}_n')| = O(\log(\poly(n)))
        \label{S8Pinching}
    \end{eqnarray}
    This is because $\tilde{σ}_n'$ is permutation invariant, so $|\spec(\tilde{σ}_n')|=O(\poly(n))$. We take $\liminf_{n\rightarrow\infty}\frac{1}{n}$ in (\ref{S8Pinching}), to find
    \begin{eqnarray}
        \liminf_{n\rightarrow\infty}\frac{1}{n}D(\rho_n\|\Tilde{\sigma}_n') =\liminf_{n\rightarrow\infty}\frac{1}{n}D(E_n(\rho_n)\|\Tilde{\sigma}_n')\,.
        \label{S8Equality2}
    \end{eqnarray}    
    Now, we have 
    \begin{equation}
        \liminf_{n\rightarrow\infty}\frac{1}{n}D_H^{\varepsilon}(E_n(\rho_n)\|\Tilde{\sigma}_n') 
        = \liminf_{n\rightarrow\infty}\frac{1}{n}D_H^{\varepsilon}(E_n(\rho_n)\|E_n(\Tilde{\sigma}_n')) 
        \leq \liminf_{n\rightarrow\infty}\frac{1}{n}D_H^{\varepsilon}(\rho_n\|\Tilde{\sigma}_n')
    \end{equation}
    by the data processing inequality. We also have
    \begin{align}
        \liminf_{n\rightarrow\infty}\frac{1}{n}D_H^{\varepsilon}(\rho_n\|\Tilde{\sigma}_n') 
        &\overset{(a)}{\leq} \liminf_{n\rightarrow\infty}\frac{1}{n}D_H^{\varepsilon}(\rho_n\|{\sigma}_n^*) \nonumber\\
        &= \liminf_{n\rightarrow\infty}\frac{1}{n}  D_H^{\varepsilon}(\rho_n\|\mathcal{F}_n^*(\omega_n)) \nonumber\\
        &\overset{(b)}{=}\liminf_{n\rightarrow\infty}\frac{1}{n} \inf_{\Bar{\mathcal{F}}_n\in\mathcal{S}_n} D_H^{\varepsilon}(\rho_n\|\Bar{\mathcal{F}}_n(\omega_n)) \nonumber\\
        &\overset{(c)}{\leq} \liminf_{n\rightarrow\infty}\frac{1}{n} \sup_{\nu_n\in\mathcal{D}(R\mathcal{X}^n)}\inf_{\Bar{\mathcal{F}}_n\in\mathcal{S}_n} D_H^{\varepsilon}(\mathcal{E}^{\otimes n}(\nu_n)\|\Bar{\F}_n(\nu_n)) \nonumber\\
        &= R_{1,\varepsilon}
        \label{S8_liminf_bound1}
    \end{align}
    where we used \cite[Lemma S11]{hayashi_generalized_2024} in (a), since by (\ref{S8Sig_star}), $\tilde{\sigma}_n' \geq \frac{1}{3}\sigma_n^*$. Also, the choice of $\mathcal{F}_n^*$ gives the equality in (b). Note that it is this part of the argument that requires us to fix the sequence of input states $(\omega_n)_n$ at the beginning, and have the channels $\F_{n, \omega_n}'$ depend on this sequence. This is, because we want to obtain $R_{1, ε}$ as the upper bound, which has the supremum over input states outside of the infimum over channels, which means there need not be a channel $\F_n^{*}$ which minimizes this expression for all input states. This could be resolved if we had an exchange lemma analogous to \Autoref{lem:rel_entropy_exchange} for the hypothesis testing relative entropy, however the hypothesis testing relative entropy does not satisfy the assumptions of \zcref{lem:rel_entropy_exchange} and it is not obvious to us whether such an exchange lemma can be proven. \\
    \\
    Similarly to the previous argument, we get that, for any $\varepsilon_1\in(0,1)$,
    \begin{eqnarray}
        && \limsup_{n\rightarrow\infty}\frac{1}{n}D_H^{1-\varepsilon_1}(E_n(\rho_n)\|\Tilde{\sigma}_n') \leq \limsup_{n\rightarrow\infty}\frac{1}{n}D_H^{1-\varepsilon_1}(\rho_n\|\Tilde{\sigma}_n')
    \end{eqnarray}
    by the data processing inequality. We also have
    \begin{align}
        \limsup_{n\rightarrow\infty}\frac{1}{n}D_H^{1-\varepsilon_1}(\rho_n\|\Tilde{\sigma}_n')
        &\overset{(a)}{\leq} \limsup_{n\rightarrow\infty}\frac{1}{n}D_H^{1-\varepsilon_1}(\rho_n\|\sigma_n)  \nonumber\\
        &= \limsup_{n\rightarrow\infty}\frac{1}{n}D_H^{1-\varepsilon_1}(\mathcal{E}^{\otimes n}(\omega_n)\|\Hat{\mathcal{F}}_n'(\omega_n))  \nonumber\\
        &\leq \limsup_{n\rightarrow\infty}\frac{1}{n} \sup_{\nu_n \in\mathcal{D}(R\mathcal{X}^n)}D_H^{1-\varepsilon_1}(\mathcal{E}^{\otimes n}(\nu_n)\|\Hat{\mathcal{F}}_n'(\nu_n))  \nonumber\\
        &\overset{(b)}{\leq} \frac{1}{m}\sup_{\nu_m\in\mathcal{D}(R\mathcal{X}^m)} D(\mathcal{E}^{\otimes m}(\nu_m)\|\mathcal{F}_m(\nu_m))  \nonumber\\
        &\leq R_2+\varepsilon_0 
        \label{S8_limsup_bound1}
    \end{align}
    Again, we have by (\ref{S8Sig_n}) that $\tilde{\sigma}_n' \geq \frac{1}{3}\sigma_n$ and so we use \cite[Lemma S11]{hayashi_generalized_2024} in (a). We use \Autoref{lem:perm_averaging_converse} as well in (b). It is important to note that this application of \zcref{lem:perm_averaging_converse} is the only part of this Lemma, and indeed the only part of the proof of the direct part, that requires the channels to be c-q.\\
    \\
    Now, for any $\varepsilon_2>0$, define the following projections:
    \begin{eqnarray}
        P_{n,1}&:=&\{E_n(\rho_n)\geq e^{n(R_{1,\varepsilon}+\varepsilon_2)}\Tilde{\sigma}_n'\} \\
        P_{n,2}&:=&\{E_n(\rho_n)\geq e^{n(R_2+\varepsilon_0 + \varepsilon_2)}\Tilde{\sigma}_n'\}
    \end{eqnarray}
    where $\{A\geq B\}$ is an orthogonal projection onto the space of non-negative eigenvalues of the matrix $A-B$. \\
    \\
    Using \cite[Lemma S10]{hayashi_generalized_2024}, and equations (\ref{S8_liminf_bound1}), (\ref{S8_limsup_bound1}), we get the following limits:
    \begin{eqnarray}
        \liminf_{n\rightarrow\infty}\tr[P_{n,1}E_n(\rho_n)] &\leq& 1-\varepsilon \label{S8b_liminf_bound2}\\
        \limsup_{n\rightarrow\infty}\tr[P_{n,2}E_n(\rho_n)] &\leq& \varepsilon_1  \label{S8b_limsup_bound2}
    \end{eqnarray}
    Since $R_{1,\epsilon}+\varepsilon_2\leq R_2 +\varepsilon_0 +\varepsilon_2$, we have that $P_{n,1} \geq P_{n,2}$. So, we can define the projections
    \begin{eqnarray}
        E_{n,1} &:=& \mathds{1}-P_{n,1} \\
        E_{n,2} &:=& P_{n,1}-P_{n,2}\\
        E_{n,3} &:=& P_{n,2}
    \end{eqnarray}
    with $\sum_j E_{n,j}=\mathds{1}$.\\
    \\
    We have the following equations:
    \begin{eqnarray}
        E_{n,1}&=&\mathds{1}-P_{n,1}=\{E_n(\rho_n)< e^{n(R_{1,\varepsilon}+\varepsilon_2)}\Tilde{\sigma}_n'\} \\
        E_{n,2}&\leq& \mathds{1}-P_{n,2}= \{E_n(\rho_n)< e^{n(R_2+\varepsilon_0 + \varepsilon_2)}\Tilde{\sigma}_n'\}
    \end{eqnarray}
    All of $P_{n,1}$, $P_{n,2}$, $E_n(\rho_n)$ and $\Tilde{\sigma}_n'$ commute, so we can use $E_{n,j}$ to project onto the subspaces where inequalities involving $\log(E_n(\rho_n))-\log\Tilde{\sigma}_n'$ hold. Thus, we obtain
    \begin{eqnarray}
        \frac{1}{n}E_{n,1}(\log(E_n(\rho_n))-\log\Tilde{\sigma}_n') &\leq& E_{n,1}(R_{1,\varepsilon}+\varepsilon_2) \label{S8_proj1} \\
        \frac{1}{n}E_{n,2}(\log(E_n(\rho_n))-\log\Tilde{\sigma}_n') &\leq& E_{n,1}(R_2 +\varepsilon_0+\varepsilon_2) \label{S8_proj2}
    \end{eqnarray}
    For $E_{n,3}$ the idea is slightly different. We first note that, by assumption (4) in \Autoref{def:axioms}, $D_{\max}(\mathcal{E}\|\mathcal{F}_*)<\infty$. So, by using (\ref{S8Sig_full}), we get that 
    \begin{eqnarray}
        \Tilde{\sigma}_n'&\geq& \frac{1}{3}\sigma_{n,full} \nonumber\\
        &=& \frac{1}{3}\mathcal{F}_*^{\otimes n}(\omega_n)  \nonumber\\
        &\geq& \frac{1}{3}e^{-D_{\max}(\mathcal{E}^{\otimes n}\|\mathcal{F}_*^{\otimes n})}\mathcal{E}^{\otimes n}(\omega_n)  \nonumber\\
        &=&\frac{e^{-nD_{\max}(\mathcal{E}\|\mathcal{F}_*)}}{3}\rho_n
    \end{eqnarray}
    by the additivity of $D_{\max}$ under tensor products. \\
    \\
    Let $\Pi$ be an orthogonal projection, so $\Pi=\Pi^\dag$. This means that if $X\geq Y$ for matrices $X,Y$, then $\Pi X\Pi \geq \Pi Y \Pi$. This means that any pinching map is operator monotone. Hence,
    \begin{eqnarray}
         E_n(\Tilde{\sigma}_n' ) &\geq &\frac{e^{-nD_{\max}(\mathcal{E}\|\mathcal{F}_*)}}{3}E_n(\rho_n) \\
         \log(E_n(\rho_n)) - \log(\Tilde{\sigma}_n') &\leq &(nD_{\max}(\mathcal{E}\|\mathcal{F}_*)+ \log 3) \mathds{1} \leq nc_n'\mathds{1}
    \end{eqnarray}
    where we define 
    \begin{eqnarray}
        c_n' :=\max\left\{\frac{\log 3}{n}+D_{\max}(\mathcal{E}\|\mathcal{F}_*), \:R_2 + \varepsilon_0 + \varepsilon_2\right\}
    \end{eqnarray}
    This means that $c_n'=O(1)$ as $n\rightarrow\infty$. So, we get
    \begin{eqnarray}
        \frac{1}{n}E_{n,3}(\log(E_n(\rho_n))-\log\Tilde{\sigma}_n') &\leq& c_n'E_{n,3}
        \label{S8_proj3}
    \end{eqnarray}
    Hence, by combining (\ref{S8_proj1}), (\ref{S8_proj2}), and (\ref{S8_proj3}), we get that
    \begin{eqnarray}
        \frac{1}{n}D(E_n(\rho_n) \|\Tilde{\sigma}_n') &=& \sum_{j=1}^3 {\frac{1}{n}} \tr[(E_n(\rho_n))E_{n,j}(\log(E_n(\rho_n))-\log\Tilde{\sigma}_n')] \nonumber\\
       &\leq &(R_{1,\varepsilon}+\varepsilon_2)\tr[(E_n(\rho_n))E_{n,1}] + (R_2+\varepsilon_0+ \varepsilon_2)\tr[E_n(\rho_n)E_{n,2}] \nonumber\\
        & & \quad \quad + c_n'\tr[E_n(\rho_n)E_{n,3}] \nonumber\\
       &=& (R_{1,\varepsilon}+\varepsilon_2)+(R_2+\varepsilon_0-R_{1,\varepsilon})\tr[E_n(\rho_n) P_{n,1}]\nonumber\\
        & & \quad \quad + (c_n'-R_2-\varepsilon_0-\varepsilon_2)\tr[E_n(\rho_n)P_{n,2}]
    \end{eqnarray}
    Now, for any real sequences $(a_n)_n$ and $(b_n)_n$, we have $\liminf_{n\rightarrow\infty}(a_n+b_n)\leq\liminf_{n\rightarrow\infty}(a_n)+\limsup_{n\rightarrow\infty}(b_n)$. Hence, by taking $\liminf_{n\rightarrow\infty}$ of both sides, we get 
    \begin{eqnarray}
     && \liminf_{n\rightarrow\infty} \frac{1}{n}D(E_n(\rho_n) \|\Tilde{\sigma}_n')  \nonumber\\
        &&\quad \quad \leq (R_{1,\varepsilon}+\varepsilon_2)+(R_2+\varepsilon_0-R_{1,\varepsilon})\liminf_{n\rightarrow\infty}\tr[E_n(\rho_n)P_{n,1}]  \nonumber\\
        && \quad \quad  \quad \quad  \quad \quad + \limsup_{n\rightarrow\infty}\left((c_n'-R_2-\varepsilon_0-\varepsilon_2) \tr[E_n(\rho_n)P_{n,2}]\right)  \nonumber\\
        && \quad \quad \leq (R_{1,\varepsilon}+\varepsilon_2)+(R_2+\varepsilon_0-R_{1,\varepsilon})(1-\varepsilon)+\varepsilon_1(c'-R_2-\varepsilon_0-\varepsilon_2)
    \end{eqnarray}
    where $c':=\limsup_{n\rightarrow\infty}c_n'$. \\
    \\
    Here, $\varepsilon_1,\varepsilon_2>0$ are arbitrary so we can take the limits $\varepsilon_1,\varepsilon_2\rightarrow 0$ to give
    \begin{eqnarray}
        \liminf_{n\rightarrow\infty} \frac{1}{n}D(E_n(\rho_n) \|\Tilde{\sigma}_n') &\leq&R_{1,\varepsilon}+(R_2+\varepsilon_0-R_{1,\varepsilon})(1-\varepsilon) \\
        &=&R_{1,\varepsilon}+(R_2-R_{1,\varepsilon})(1-\Tilde{\varepsilon}) \label{S8_final2} 
    \end{eqnarray}
    where we recall the definition of $\varepsilon_0$ in (\ref{eq:GQSL_direct_e0_def}). We now combine (\ref{S8_final2}) and (\ref{S8Equality2}) to get
    \begin{eqnarray}
        \liminf_{n\rightarrow\infty}\frac{1}{n}D(\mathcal{E}^{\otimes n}(\omega_n)\|\mathcal{F}'_{n,\omega_n}(\omega_n))  \leq R_{1,\varepsilon}+(R_2-R_{1,\varepsilon})(1-\Tilde{\varepsilon})
    \end{eqnarray}
    We denote the set of all channels $\Bar{\mathcal{F}}_n$ that are permutation covariant in $\mathcal{S}_n$ by $\Bar{\mathcal{S}}_n$. Then, since $\F_{n,\omega_n}'$ is permutation covariant with respect to permuting the $n$ copies of $\X$ and $A$, we have that
    \begin{eqnarray}
        \liminf_{n\rightarrow\infty}\inf_{\Bar{\mathcal{F}}_n\in\Bar{\mathcal{S}}_n} \frac{1}{n}D(\mathcal{E}^{\otimes n}(\omega_n)\|\Bar{\mathcal{F}}_{n}(\omega_n))  &\leq& \liminf_{n\rightarrow\infty}\frac{1}{n}D(\mathcal{E}^{\otimes n}(\omega_n)\|\mathcal{F}'_{n,\omega_n}(\omega_n))\\
        &\leq& R_{1,\varepsilon}+(R_2-R_{1,\varepsilon})(1-\Tilde{\varepsilon})
    \end{eqnarray}
    Now, $(\omega_n)_n$ was an arbitrary sequence of permutation invariant states, so we can take a supremum over all such sequences $(\omega_n)_n$ and use \Autoref{lem:sup_over_sequences}, to get 
    \begin{eqnarray}
        \liminf_{n\rightarrow\infty}\sup_{\substack{\omega_n\in\mathcal{D}(K^n\mathcal{X}^n)\\\omega_n\;perm.\;invariant}} \inf_{\Bar{\mathcal{F}}_n\in \Bar{\mathcal{S}}_n} \frac{1}{n}D(\mathcal{E}^{\otimes n}(\omega_n)\|\Bar{\mathcal{F}}_{n}(\omega_n))  \leq R_{1,\varepsilon}+(R_2-R_{1,\varepsilon})(1-\Tilde{\varepsilon})
    \end{eqnarray}
    By \Autoref{lem:perm_inputs}, we can take the supremum to be over all states $\nu_n\in\mathcal{D}(R\mathcal{X}^n)$ instead, which gives the following:
    \begin{eqnarray}
        \liminf_{n\rightarrow\infty}\sup_{\nu_n\in\mathcal{D}(R\mathcal{X}^n)} \inf_{\Bar{\mathcal{F}}_n\in \Bar{\mathcal{S}}_n} \frac{1}{n}D(\mathcal{E}^{\otimes n}(\nu_n)\|\Bar{\mathcal{F}}_{n}(\nu_n))  \leq R_{1,\varepsilon}+(R_2-R_{1,\varepsilon})(1-\Tilde{\varepsilon})
    \end{eqnarray}
    Now, we can use \Autoref{lem:rel_entropy_exchange} to exchange the supremum and infimum to get
    \begin{eqnarray}
        \liminf_{n\rightarrow\infty} \inf_{\Bar{\mathcal{F}}_n\in \Bar{\mathcal{S}}_n} \sup_{\nu_n\in\mathcal{D}(R\mathcal{X}^n)}\frac{1}{n}D(\mathcal{E}^{\otimes n}(\nu_n)\|\Bar{\mathcal{F}}_{n}(\nu_n))  \leq R_{1,\varepsilon}+(R_2-R_{1,\varepsilon})(1-\Tilde{\varepsilon})
    \end{eqnarray}
    This means that there exists a sequence $(\mathcal{F}_n')_n$ that satisfies the required equation.
\end{proof}
\medskip
Finally, we prove the direct part by using \Autoref{lem:main_iteration}, again following a similar idea as in \cite[Proposition S7]{hayashi_generalized_2024}. Specifically, we find sequences of channels $(\mathcal{F}_{n,k})_{n,k}$ where 
\begin{eqnarray}
     \liminf_{n\rightarrow\infty}\frac{1}{n}\sup_{\nu_n\in\mathcal{D}(R\mathcal{X}^n)} D(\mathcal{E}^{\otimes n}(\nu_n)\|\mathcal{F}_{n,k}(\nu_n)) -R_{1,\varepsilon} \leq (1-\tilde{\varepsilon})^k(R_2-R_{1,\varepsilon})
\end{eqnarray}
by iteratively applying \Autoref{lem:main_iteration}. Then, by using parts of each of these sequences, we create a new sequence of channels such that 
\begin{eqnarray}
    \liminf_{n\rightarrow\infty}\frac{1}{n}\sup_{\nu_n\in\mathcal{D}(R\mathcal{X}^n)} D(\mathcal{E}^{\otimes n}(\nu_n)\|\mathcal{F}_{n,*}(\nu_n)) \leq R_{1,\varepsilon}
\end{eqnarray}
which proves the claim. 
\begin{proposition}[Direct Part of the Generalized Quantum Stein's Lemma]
\label{lem:GQSL_direct}
    Let $\E:\mathcal{X}\rightarrow A$ be a c-q channel, and $(\mathcal{S}_n\subset\CPTP(\mathcal{X}^n\rightarrow A^n))_n$ be a sequence of sets of c-q channels satisfying assumptions (1)-(4) in \Autoref{def:axioms}. Then, the direct part
    \begin{equation}
    \liminf_{n \to \infty} {1 \over n} D_H^{ε}(\E\n\|\S_n) \geq \lim_{n \to \infty} {1 \over n} D(\E\|\S_n)
    \end{equation}%
    holds.
\end{proposition}
\begin{proof}
    Fix some $\tilde{\varepsilon}\in(0,\varepsilon)$. Assume there exists a sequence of channels $(\mathcal{F}_{n,0})_n$, with $\mathcal{F}_{n,0}\in\mathcal{S}_n$, such that 
    \begin{eqnarray}
        R_{2,0}:= \liminf_{n\rightarrow\infty}\frac{1}{n}\sup_{\nu_n\in\mathcal{D}(R\mathcal{X}^n)}D(\mathcal{E}^{\otimes n}(\nu_n)\|\mathcal{F}_{n,0}(\nu_n))>R_{1,\varepsilon} = \liminf_{n \to \infty} {1 \over n} D_H^{ε}(\E\n\|\S_n)
    \end{eqnarray}
    We have
    \begin{eqnarray}
        \liminf_{n\rightarrow\infty}\frac{1}{n}\sup_{\nu_n\in\mathcal{D}(R\mathcal{X}^n)}D(\mathcal{E}^{\otimes n}(\nu_n)\|\mathcal{F}_{n,0}(\nu_n))
        &\geq&\liminf_{n\rightarrow\infty}\frac{1}{n}\inf_{\mathcal{F}_n \in \mathcal{S}_n} \sup_{\nu_n\in\mathcal{D}(R\mathcal{X}^n)}D(\mathcal{E}^{\otimes n}(\nu_n)\|\mathcal{F}_{n}(\nu_n))\nonumber\\
        &=&\liminf_{n\rightarrow\infty}\frac{1}{n}\sup_{\nu_n\in\mathcal{D}(R\mathcal{X}^n)}\inf_{\mathcal{F}_n \in \mathcal{S}_n}D(\mathcal{E}^{\otimes n}(\nu_n)\|\mathcal{F}_{n}(\nu_n))\nonumber\\
    \end{eqnarray}
    So, if such a sequence $(\mathcal{F}_{n,0})_n$ does not exist, then we have proved the direct part. Otherwise, we use \Autoref{lem:main_iteration} to get another sequence of channels $(\mathcal{F}_{n,1})_n$ such that
    \begin{eqnarray}
        \liminf_{n\rightarrow\infty}\frac{1}{n}\sup_{\nu_n\in\mathcal{D}(R\mathcal{X}^n)}D(\mathcal{E}^{\otimes n}(\nu_n)\|\mathcal{F}_{n,1}(\nu_n)) -R_{1,\varepsilon} \leq (1-\Tilde{\varepsilon})(R_{2,0}-R_{1,\varepsilon})
    \end{eqnarray}
    Now, repeatedly apply this procedure, so from the sequence $(\mathcal{F}_{n,k})_n$ define
    \begin{eqnarray}
        R_{2,k}:= \liminf_{n\rightarrow\infty}\frac{1}{n}\sup_{\nu_n\in\mathcal{D}(R\mathcal{X}^n)}D(\mathcal{E}^{\otimes n}(\nu_n)\|\mathcal{F}_{n,k}(\nu_n))
    \end{eqnarray}
    Then, use \Autoref{lem:main_iteration} to find a sequence $(\mathcal{F}_{n,k+1})_n$ such that 
    \begin{eqnarray}
        R_{2,k+1}-R_{1,\varepsilon} \leq (1-\Tilde{\varepsilon})(R_{2,k}-R_{1.\varepsilon}) \leq (1-\Tilde{\varepsilon})^{k+1}(R_{2,0}-R_{1.\varepsilon})
    \end{eqnarray}
    In particular we find that
    \begin{equation}
    \liminf_{n\rightarrow\infty}\frac{1}{n}\sup_{\nu_n\in\mathcal{D}(R\mathcal{X}^n)}\inf_{\F_{n} \in \S_n} D(\mathcal{E}^{\otimes n}(\nu_n)\|\mathcal{F}_{n}(\nu_n)) -R_{1,\varepsilon} \leq (1 - \tilde{ε})^k (R_{2, 0} - R_{1, ε})
    \end{equation}
    for all $k \in \naturals$, and taking the limit $k \to \infty$ on the right-hand side gives the desired result.
    \begin{EXCLUDED}%
    So, for any sequence $(\delta_k)_k$ with $\delta_k>0$ and $\delta_k\rightarrow 0$ as $k\rightarrow\infty$, there exists a subsequence $(n_k)_k$ of the natural numbers such that
    \begin{eqnarray}
        \frac{1}{n_k}\sup_{\nu_{n_k} \in\mathcal{D}(R\mathcal{X}^{n_k})}D(\mathcal{E}^{\otimes n_k}(\nu_{n_k})\|\mathcal{F}_{n_k,k}(\nu_{n_k})) -R_{1,\varepsilon} &\leq& (1-\Tilde{\varepsilon})^k(R_{2,0}-R_{1,\varepsilon})+\delta_k \nonumber\\
        &\rightarrow& 0 \; \text{as} \;k \rightarrow\infty
    \end{eqnarray}
    We use this to define a new sequence of channels $(\mathcal{F}_{n,*})_n$, with $\mathcal{F}_{n_k,*}=\mathcal{F}_{n_k,k}$ for all $k$. This means that 
    \begin{eqnarray}
        \liminf_{n\rightarrow\infty}\frac{1}{n}\sup_{\nu_n\in\mathcal{D}(R\mathcal{X}^n)}D(\mathcal{E}^{\otimes n}(\nu_n)\|\mathcal{F}_{n,*}(\nu_n)) \leq R_{1,\varepsilon}
    \end{eqnarray}
    which gives the desired result.
    \end{EXCLUDED}%
\end{proof}
%
%
\section{Quantum Resource Theories for C-Q Channels}\label{sec:QRT}

As explained in \zcref{sec:intro_QRT}, in a quantum resource theory (QRT) we consider certain resources to be free resources and all other resources to be valuable resources. The question of interest is to study the optimal conversion rate between two valuable resources under a restricted set of operations (called free operations). If a QRT is reversible, this conversion rate gives rise to a unique resource quantifier for the theory. The Generalized Quantum Stein's Lemma can be used to prove the reversibility of QRTs under a certain class of free operations called asymptotically non-resource generating (ARNG) operations. In~\cite{hayashi_generalized_2024}, the authors also took a first step at extending the the reversibility results for QRTs of states to QRTs involving dynamical resources, in particular, they considered classical-quantum (c-q) channels. In such a dynamical QRT, the objects of interests are channels, and the free operations are a set of superchannels \cite{gour_comparison_2019}.
The strategy of \cite{hayashi_generalized_2024} is to apply the GQSL for states to this problem, by considering the Choi states of the c-q channels. The distance metric between channels they consider is then the trace distance between their Choi states. For c-q channels this can easily be seen to be equivalent to the trace distance of the output states of the two channels, \emph{on average over all input states}. In particular, they consider a transformation between channels to be possible if the trace distance of the output states of the transformed channel and the target channel \emph{on average over all input states} goes to zero asymptotically. Since the size of the input system grows exponentially as $n$ increases, this means the channels might still give very different output states (potentially even orthogonal output states) on a polynomially growing fraction of the input space. In particular the transformed channel and the target channel will usually be perfectly distinguishible by just picking a suitable input state. 

Additionally, the authors of \cite{hayashi_generalized_2024} face the issue that Choi states do not behave well under the action of superchannels: When considering a superchannel on the level of Choi states (i.e., the mapping of a channel's Choi matrix to the Choi matrix of the application of a super-channel to this channel) it need not be a CPTP operation. Hence, in \cite{hayashi_generalized_2024} the authors have to further restrict the set of allowed operations to enforce that they behave sufficiently well when considered as maps between Choi states. Concretely, they have to impose the following constraint (which they refer to as \emph{asymptotic continuity}): If the Choi states of two sequences of c-q channels converge in trace distance in the asymptotic limit, then this property holds even after the channels are acted on by a sequence of allowed operations.
As explained in~\cite{hayashi_generalized_2024}, they chose this path of trying to reduce the problem to the state case, because they do not have a GQSL for c-q channels available.  

By using our GQSL for c-q channels, we are able to resolve the issues faced by the authors of \cite{hayashi_generalized_2024} and prove reversibility for QRTs involving c-q channels in a much more natural setting. In particular we compare channels in diamond norm (that is the worst-case trace distance of the output state over all input states), and so if one channel can be (asymptotically) transformed into another, the transformed and target channels' outputs will be (asymptotically) close for all inputs. Additionally, in our statement the allowed operations (i.e., superchannels) are just all resource non-generating (ARNG) operations, and not an implicitly specified subset that satisfies additional continuity properties. This comes at the small additional cost of having to impose the additional assumption of closure under permutations for our sets of free c-q channels, as we required this to prove our GQSL (\zcref{thm:GQSL}).

\subsection{Setup and Definitions}

Resource theories are generally specified by sets of free objects and free operations. We will consider QRTs where the sets of free objects are given by sets of c-q channels, i.e., sets $(\mathcal{S}_n \subset\CPTP(\mathcal{X}^n\rightarrow A^n))_n$. We require this sequence of sets satisfy the axioms of \zcref{def:axioms}, i.e.\ the same axioms we already specified for the Generalized Quantum Stein's Lemma.

Every c-q channel $\E_n \in \CQ$ that is not an element of $\S_n$ is not free, and hence called valuable. We assign to these channels a measure of how resourceful they are, called the \textit{relative entropy of resource}. We define it as follows:
\begin{eqnarray}
    R(\E_n):=D(\E_n\|\S_n) = \inf_{\F_n\in \S_n}\sup_{\nu\in\DM[R\X^n]}D(\E_n(\nu) \| \F_n(\nu))
    \label{eq:rel_entropy_resource_def}
\end{eqnarray}
where the $\inf$ and $\sup$ on the right-hand side can be exchanged due to \zcref{lem:rel_entropy_exchange}.
For a channel $\E \in \CQ$, we also define its regularized form
\begin{eqnarray}
    \R(\E):=\lim_{n\rightarrow\infty}\frac{1}{n}R(\E^{\otimes n})\,.
    \label{eq:reg_rel_entropy_resource_def}
\end{eqnarray} 
We also define the \textit{log robustness} of $\E_n\in\CPTP(\X^n\to A^n)$ as\footnote{In \cite{hayashi_generalized_2024}, they write $R_G(\E)$ for the robustness, which we can relate to the log robustness by $D_{\max}(\E \| \S_n)=\log(1+R_G(\E))$.}
\begin{eqnarray}
    D_{\max}(\E_n\|\S_n)= \sup_{\nu\in\DM[R\X^n]} \inf_{\F_n\in \S_n} D_{\max}(\E_n(\nu)\|\F_n(\nu)) = \inf_{\F_n \in \S_n} D_{\max}(\E_n\|\F_n)
\end{eqnarray}
where the last equality follows from \zcref{lem:channel_divergence_maximally_entangled_state} and the fact that $\sup \inf \leq \inf \sup$. 

\bigskip
We now define a set of operations which are allowed in order to transform resources into each other. These operations are usually called \emph{free operations}.
We take the set of free operations $\OR$ as the set of sequences of superchannels that are asymptotically resource non-generating (ARNG)\footnote{We can also define \emph{resource non-generating} operations $(\Theta_n)_n$ where for any sequence $(\F_n)_n$ of free channels we have that $D_{\max}(\Theta(\F_n)\| \S_n)=0$ for \emph{every} $n$.} in log robustness. This means that, for a sequence of superchannels $(\Theta_n)_n\in\OR$, and for any sequence $(\F_n\in\mathcal{S}_n)_n$ of free channels, we have that 
\begin{eqnarray}
    \lim_{n\rightarrow\infty}D_{\max}(\Theta_n(\F_n)\| \mathcal{S}_n)=0
\end{eqnarray}
\\
The \textit{optimal asymptotic conversion rate} between two valuable resources $\E_1$ and $\E_2$ under the set of operations $\OR$ is then given by
\begin{eqnarray}
    r(\E_1\rightarrow \E_2):=\sup\left\{r>0 \,\bigg|\, \exists(\Theta_n)_n\in\OR\, s.t. \, \limsup_{n\rightarrow\infty} \norm{\Theta_n(\E_1\n)-\E_2^{\otimes \lceil rn \rceil}}_\Diamond=0 \right\}
\end{eqnarray}

\subsection{Statement of the Result}

We prove that if the sets of free c-q channels satisfy the assumptions of \zcref{def:axioms} and the set of free operations (i.e., superchannels) are the asymptotically resource non-generating (ARNG) operations (as specified above), then the QRT is reversible. Hence also, the regularized relative entropy of resource is the unique resource quantifier.

\begin{theorem}[QRT for c-q channels]
    \label{thm:QRT}
    Let $\E_1$, $\E_2 \in \CQ$ be two c-q channels, with $\R(\E_i)>0$ for $i=1,2$. Then
    \begin{eqnarray}
        r(\E_1\rightarrow\E_2) = \frac{\R(\E_1)}{\R(\E_2)}
    \end{eqnarray}
\end{theorem}
The overall idea of the proof is to follow the established path of first relating the regularized relative entropy of resource $\R(\E)$ to the log robustness in \Autoref{cor:log_robustness}, where we prove both sides of the inequality in \Autoref{lem:log_robustness1} and \Autoref{lem:log_robustness2}. Then, we use the GQSL (\Autoref{thm:GQSL}) to show achievability of the optimal rate of resource conversion in \Autoref{lem:QRT_direct}. Finally, we prove that the relative entropy of resource is monotone under a sequence of ARNG operations in \Autoref{lem:resource_monotone}, and use this to prove the converse part in \Autoref{lem:QRT_converse}.

\subsection{Relating Relative Entropy of Resource to Log Robustness}
First, we relate the regularized relative entropy of resource of a c-q channel, $\E$, to the log robustness of a sequence of c-q channels $(\tilde{\E}_n)_n$ which is asymptotically identical to $\E\n$, in the sense of $\norm{\E\n-\tilde{\E}_n}_\Diamond\rightarrow0$ as $n\rightarrow\infty$. The goal of this section is to prove the following corollary:

\getkeytheorem{stored:cor:log_robustness}
\begin{EXCLUDED}
Note that \zcref{cor:log_robustness} can be thought of as an asymptotic equipartition property (AEP). If we define a channel-smoothed max relative entropy (for $\E, \F \in \cptp[\X \to A]$) as\footnote{Note that this is not a channel-divergence in the sense of \eqref{eq:def_channel_divergence}, i.e., it cannot -- at least not obviously -- be written as a state-divergence optimized over all input states, for more properties of this quantity see also \cite{wang_resource_2019}.}
\begin{equation}\label{eq:def_smoothed_diamond_max}
D_{\max}^{\Diamond, ε}(\E\|\F) \coloneqq \inf_{\substack{\tilde{\E} \in \cptp[\X \to A]\\ \norm{\tilde{\E} - \E}_{\Diamond} \leq ε}}D_{\max}(\tilde{\E}\|\F)
\end{equation}
and set $D_{\max}(\E\|\S) \coloneqq \inf_{\F \in \S} D_{\max}(\E\|\F)$, then \zcref{cor:log_robustness} is equivalent to\footnote{To see this equivalence, note first that for any $ε \in (0,1)$, any sequence $(\tilde{\E}_n)_n$ on the right-hand side of the statement of \zcref{cor:log_robustness} becomes (for $n$ large enough) a valid sequence of $\tilde{\E}$ for \eqref{eq:def_smoothed_diamond_max}, and hence the right-hand side of \eqref{eq:log_robustness_aep} can only be smaller. }
\begin{equation}\label{eq:log_robustness_aep}
    R^{\infty}(\E) = \lim_{n \to \infty} {1 \over n} D(\E\n\|\S_n) = \lim_{ε \to 0} \lim_{n \to \infty} {1 \over n} D_{\max}^{\Diamond, ε}(\E\n\|\S_n)\,.
\end{equation}
\end{EXCLUDED}
To prove \zcref{cor:log_robustness} we would first like to relate the relative entropy of resource of two sequences of c-q channels $(\E_n)_n$ and $(\E_n')_n$ which become identical asymptotically, in the sense of 
\begin{eqnarray}
    \lim_{n\rightarrow\infty}{\norm{\E_n-\tilde{\E_n}}}_\Diamond=0\,.
\end{eqnarray}

Intuitively, this limit suggests that the relative entropy of resource of the two sequences of channels should be the same asymptotically, as they are indistinguishable in the limit $n\rightarrow\infty$. Indeed, this is what we prove in the following lemma: 
\begin{lemma}[Asymptotically identical channels have equal regularized relative entropy of resource]
    \label{lem:resource_close_sequences}
    Let $(\E_n)_n$ and $(\tilde{\E}_n)_n$ be two sequences of c-q channels, with $\E_n,\tilde{\E}_n:\mathcal{X}^n\rightarrow A^n$, satisfying 
    \begin{eqnarray}
        \lim_{n\rightarrow\infty}{\norm{\E_n-\tilde{\E_n}}}_\Diamond=0
    \end{eqnarray}
    Then,
    \begin{eqnarray}
        \liminf_{n\rightarrow\infty}\frac{1}{n}R(\E_n)=\liminf_{n\rightarrow\infty}\frac{1}{n}R(\E_n')
    \end{eqnarray}
\end{lemma}
\begin{proof}
    We begin by using \cite[Theorem 4]{gour_how_2019}, which states that if $\E_1, \E_2: \mathcal{X}\rightarrow A$ are quantum channels, with $\frac{1}{2}\norm{\E_1-\E_2}_\Diamond=\varepsilon$, then
    \begin{eqnarray}
        \abs{R(\E_1)-R(\E_2)}\leq (1+\varepsilon)h\left(\frac{\varepsilon}{1+\varepsilon} \right) + \varepsilon \kappa
    \end{eqnarray}
    where $h(x)=-(1-x)\log(1-x)-x\log x$ is the binary entropy function and 
    \begin{eqnarray}
        \kappa := \sup_{\mathcal{N}\in\cptp[\mathcal{X}\rightarrow A]}R(\mathcal{N})
    \end{eqnarray}
    The expression on the right hand side is increasing in $\varepsilon$, so we can assume $\frac{1}{2}\norm{\E_1-\E_2}_\Diamond\leq\varepsilon$.
    \\
    \\
    Fix $0<\varepsilon<1$. Now, choose $N$ sufficiently large such that, for all $n\geq N$, we have $\frac{1}{2}\norm{\E_n-\tilde{\E}_n}_\Diamond\leq\varepsilon$. Then, we have that 
    \begin{eqnarray}
        \abs{R(\E_n)-R(\tilde{\E}_n)}\leq (1+\varepsilon)h\left(\frac{\varepsilon}{1+\varepsilon} \right) + \varepsilon \kappa_n
    \end{eqnarray}
    where 
    \begin{eqnarray}
        \kappa_n := \sup_{\mathcal{N}\in\cptp[\X^n\to A^n]}R(\mathcal{N})
    \end{eqnarray}
    Here, since we use the result of \cite[Theorem 4]{gour_how_2019}, which is phrased for general quantum channels, $\mathcal{N}$ need not, in general, be a c-q channel. \\
    \\
    We bound $\kappa_n$, by noting that, by the definition of $R(\mathcal{N})$,
    \begin{eqnarray}
        \kappa_n &=&  \sup_{\mathcal{N}\in\cptp[\X^n\to A^n]} \inf_{\F_n\in\mathcal{S}_n}\sup_{\nu_n\in\mathcal{D}(R\mathcal{X}^n)}D(\mathcal{N}_n(\nu_n)\| \F_n(\nu_n)) \\
        &\leq& \sup_{\mathcal{N}\in\cptp[\X^n\to A^n]} \inf_{\F_n\in\mathcal{S}_n}D_{\max}(\mathcal{N}_n\| \F_n) \\
        &\overset{(a)}{\leq}& \sup_{\mathcal{N}\in\cptp[\X^n\to A^n]} D_{\max}(\mathcal{N}_n\| \F_{*}^{\otimes n}) \\
        &\overset{(b)}{\leq}& \log\left(\frac{1}{λ_{\min}(\F_{*}(\Phi)\n)}\right)\\
        &=& -n \log(λ_{\min}(\F_{*}(\Phi)))
    \end{eqnarray}
    Here, (a) uses that $\mathcal{F}_{*}\n\in\mathcal{S}_n$, by assumptions (3) and (4) in \Autoref{def:axioms}, and (b) holds since the $D_{\max}$ of two channels is always achieved by a maximally entangled input state (\Autoref{lem:channel_divergence_maximally_entangled_state}), which we we wrote as $\Phi = \Phi_{\X \tilde{\X}}$, and the maximally entangled state on the tensor product system $\Phi_{\X\n \tilde{\X}\n}$ satisfies $\Phi_{\X\n \tilde{\X}\n} = \Phi_{\X \tilde{\X}}\n$. Finally, the Choi state $\F_{*}(\Phi)$ is positive definite by assumption (4) and thus has a positive minimum eigenvalue (namely $\lambda_{\min}(\F_*(\Phi))$.

    This means that 
    \begin{eqnarray}
        \limsup_{n\rightarrow\infty}\frac{1}{n}\abs{R(\E_n)-R(\tilde{\E}_n)}\leq \limsup_{n\rightarrow\infty} \varepsilon \frac{\kappa_n}{n} \leq - \varepsilon \log\left(λ_{\min}(F_{*}(\Phi))\right)
    \end{eqnarray}
    Since $\varepsilon>0$ is arbitrary, we can take the limit as $\varepsilon\rightarrow0$ to obtain
    \begin{eqnarray}
        \limsup_{n\rightarrow\infty}\frac{1}{n}\abs{R(\E_n)-R(\tilde{\E}_n)} =0
    \end{eqnarray}
    This gives our required result.
\end{proof}
We now use this lemma to relate the relative entropy of resource to the log robustness in the following two lemmas: \\
\\
\begin{lemma}
    \label{lem:log_robustness1}
    For any sequence $(\E_n\in\CPTP(\X^n\to A^n))_n$ of c-q channels, we have 
    \begin{eqnarray}
        \liminf_{n\rightarrow\infty}\frac{1}{n}R(\E_n) \leq \inf_{(\tilde{\E}_n)_n}\left\{ \liminf_{n\rightarrow\infty}\frac{1}{n}D_{\max}(\tilde{\E}_n\|\S_n) \bigg| \lim_{n\rightarrow\infty} \frac{1}{2} \norm{\E_n-\tilde{\E}_n}_\Diamond =0\right\}
    \end{eqnarray}
\end{lemma}
\begin{proof}
    This lemma is analogous to \cite[Lemma S13]{hayashi_generalized_2024}, but with a different notion of distance between channels.
    Take any sequence $(\tilde{\E}_n)_n$ satisfying
    \begin{eqnarray}
        \lim_{n\rightarrow\infty}{\norm{\E_n-\tilde{\E}_n}_\Diamond}=0
    \end{eqnarray}
    For every c-q channel $\tilde{\E}_n$, there exists a free c-q channel $\F_n$ such that 
    \begin{eqnarray}
        D_{\max}(\tilde{\E}_n\|\F_n) = D_{\max}(\tilde{\E}_n\|\S_n) 
    \end{eqnarray}
    This exists as $D_{\max}$ is jointly lower semi-continuous, and $\S_n$ is compact. So, for any input state $\nu_n$, we have that
    \begin{eqnarray}
        D(\tilde{\E}_n(\nu_n)\|\F_n(\nu_n))) &\leq& D_{\max}(\tilde{\E}_n(\nu_n)\|\F_n(\nu_n)) \\
        &\leq&  D_{\max}(\tilde{\E}_n\|\F_n) 
        \\ &=&  D_{\max}(\tilde{\E}_n\|\S_n)
        \label{eq:robustness_bound}
    \end{eqnarray}
    We now take a supremum over input states, and an infimum over the set of free channels can only reduce the quantum relative entropy on the left-hand side of (\ref{eq:robustness_bound}), so we get
    \begin{eqnarray}
        R(\tilde{\E_n})=\inf_{\mathcal{F}_n \in \mathcal{S}_n} \sup_{\nu_n\in\mathcal{D}(R\mathcal{X}^n)}D(\tilde{\E}_n(\nu_n)\|\F_n(\nu_n))) \leq D_{\max}(\tilde{\E}_n\|\S_n)
    \end{eqnarray}
    Now, by \Autoref{lem:resource_close_sequences}, we get
    \begin{eqnarray}
        \liminf_{n\rightarrow\infty}{\frac{1}{n}R(\E_n)}=\liminf_{n\rightarrow\infty}{\frac{1}{n}R(\tilde{\E}_n)} \leq \liminf_{n\rightarrow\infty}\frac{1}{n}D_{\max}(\tilde{\E}_n\|\S_n)
    \end{eqnarray}
\end{proof}
The following lemma is based on the same construction as \cite[Lemma S14]{hayashi_generalized_2024}. Here, we create a new sequence of c-q channels using their Choi states, and then show that this sequence satisfies both the required equations. The log robustness condition is proved using the pinching inequality, whilst the convergence in diamond norm is proved using the strong converse part of the GQSL for c-q channels (specifically \Autoref{lem:GQSL_converse_bound}).
\begin{lemma}
    \label{lem:log_robustness2}
    For any c-q channel $\E$, and any sequence of sets of c-q channels $(\S_n \subset \CPTP(\X^n \to A^n))_n$ satisfying the assumptions of \Autoref{def:axioms}, there exists a sequence of c-q channels $(\tilde{\E}_n)_n$, with $\tilde{\E}_n:\mathcal{X}^n\rightarrow A^n$ such that both of the following equations hold:
    \begin{eqnarray}
        \R(\E) \geq \limsup_{n\rightarrow\infty}\frac{1}{n}D_{\max}(\tilde{\E_n}\|\S_n) \\
        \lim_{n\rightarrow\infty}\frac{1}{2}\norm{\E^{\otimes n}-\tilde{\E}_n}_{\Diamond}=0
    \end{eqnarray}
\end{lemma}
\begin{proof}
    Choose any $R>\R(\E)$. Since, by \Autoref{lem:rel_entropy_exchange}, we have 
    \begin{eqnarray}
        \R(\E)= \lim_{n\rightarrow\infty}\frac{1}{n}\inf_{\mathcal{F}_n\in \mathcal{S}_n} \sup_{\nu_n\in\mathcal{D}(R\mathcal{X}^n)}{D(\mathcal{E}^{\otimes n}(\nu_n)\|\mathcal{F}_n(\nu_n))} 
    \end{eqnarray}
    there exists an integer $m>0$ and a free channel $\F_m \in \S_m$ such that
    \begin{eqnarray}
        \frac{1}{m} \sup_{\nu_m\in\mathcal{D}(R\mathcal{X}^m)}{D(\mathcal{E}^{\otimes m}(\nu_m)\|\mathcal{F}_m(\nu_m))} <R \,.
    \end{eqnarray}
    For any channel $\mathcal{N}:\mathcal{X}\rightarrow A$ we write $J(\mathcal{N}) = \mathcal{N}(\Phi_{\tilde{X}X}) \in\mathcal{D}(\mathcal{\tilde{X}A})$ for its normalized Choi state.
    Then, let $E_k$ be the pinching map with respect to the state $J(\F_m^{\otimes k}) = J(\F_m)^{\otimes k}$. This state is permutation invariant with respect to permuting the $k$ copies of $\tilde{\mathcal{X}}^mA^m$, which implies that $\abs{\spec J(\F_m^{\otimes k})}=O(\poly(k))$ as $k\rightarrow\infty$.

    One also easily sees that because the channel is c-q, we can write $J(\F_m^{\otimes k})$ as 
    \begin{equation}
    J(\F_m^{\otimes k}) = \sum_{x \in \X^{km}} \ketbra{x}{x} \otimes \F_m\k(\ketbra{x}{x})
    \end{equation}

    Now, define the projection:
    \begin{eqnarray}
        P_k:=\{E_k(J(\mathcal{E}^{\otimes km}))-e^{kmR}J(\F_m^{\otimes k}) \geq 0\}
        \label{eq:QRT_logrobustness2_Pk_def}
    \end{eqnarray}
    which commutes with $J(\F_m\k)$ by construction, and hence can also be written as
    \begin{equation}
        P_k = \sum_{x \in \X^{km}} \ketbra{x}{x} \otimes P_k^{(x)}\,.
    \end{equation}
    The intuition is that the subspace that $P_k$ projects onto contains all the \enquote{bad} parts of the channel (output), i.e., the ones which will lead to a large $D_{\max}$. We will show that by projecting onto its orthogonal complement we can remove those bad parts, and this will only incur a vanishing error asymptotically. Since projecting onto a subspace is not trace-preserving in general, for the construction of $\tilde{\E}$ we will add the output of a free channel to make $\tilde{\E}$ trace preserving. 
    Define the c-q channel $\tilde{\E}_{km}$ for all $x \in \X^{km}$ as
    \begin{multline}\label{eq:qrt_robustness_proof_def_cq_channel}
    \tilde{\E}_{km}(\ketbra{x}{x}) \coloneqq \qty(\IdentityMatrix - P_k^{(x)}) \E^{\otimes km}(\ketbra{x}{x}) \qty(\IdentityMatrix - P_{k}^{(x)}) \\+ \qty(1 - \Tr[\qty(\IdentityMatrix - P_k^{(x)}) \E^{\otimes km}(\ketbra{x}{x})]) \F_{*}^{\otimes km} (\ketbra{x}{x})
    \end{multline}
    where $\F_{*} \in \S_1$ is the channel with a full-rank Choi matrix from \Autoref{def:axioms}. It is easy to see that the Choi matrix of $J(\tilde{\E}_{km})$ can be written as
    \begin{equation}
    J(\tilde{\E}_{km}) = (\IdentityMatrix - P_k) J(\E^{\otimes km}) (\IdentityMatrix - P_k) + ω
    \end{equation}
    where $ω$ is a positive matrix coming from the second term in \eqref{eq:qrt_robustness_proof_def_cq_channel}. Since \[1 - \Tr[\qty(\IdentityMatrix - P_k^{(x)}) \E^{\otimes km}(\ketbra{x}{x})] \leq 1\] we have that $ω \leq J(\F_*^{\otimes km}) \in \S_{km}$, and so $D_{\max}(ω\|\S_{km})\leq 0$. 

    Now note that we have the following chain of inequalities,
    \begin{eqnarray}
        (\mathds{1}-P_k)J(\mathcal{E}^{\otimes km})(\mathds{1}-P_k) &\overset{(a)}{\leq}& \abs{\spec J(\F_m^{\otimes k})} (\mathds{1}-P_k)E_k(J(\mathcal{E}^{\otimes km}))(\mathds{1}-P_k) \\
        &\overset{(b)}{\leq}& \abs{\spec J(\F_m^{\otimes k})} e^{kmR}(\mathds{1}-P_k)J(\F_m^{\otimes k})(\mathds{1}-P_k)\\
        &\overset{(c)}{\leq}& \abs{\spec J(\F_m^{\otimes k})} e^{kmR} J(\F_m^{\otimes k})\,,
        \label{eq:QRT_logrobustness2_bound2}
    \end{eqnarray}
    where, (a) is by the pinching inequality, (b) follows from the definition of the projection $P_k$ in  \eqref{eq:QRT_logrobustness2_Pk_def}, and (c) holds since $P_k$ commutes with $J(\F_m^{\otimes k})$. \\
    Since $\S_n \subset \CPTP(\X^n \to A^n)$ is assumed to be convex it holds that\footnote{One easily sees that if $ρ_1 \leq λ_1 σ_1$ and $ρ_2 \leq λ_2 σ_2$ then $ρ_1 + ρ_2 \leq (λ_1 + λ_2) {λ_1 σ_1 + λ_2 σ_2 \over λ_1 + λ_2}$, and hence for a convex set of states $S$: $D_{\max}(ρ_1 + ρ_2 \| S) \leq D_{\max}(ρ_1\|S) + D_{\max}(ρ_2\|S)$. This can be applied together with \Autoref{lem:channel_divergence_maximally_entangled_state} which states that the channel max divergence is always achieved at the maximally entangled input state, and thus given by the max divergence of Choi states.}
    \begin{align}
        D_{\max}(\tilde{\E}_{km}\|\S_{km}) &\leq D_{\max}((\IdentityMatrix - P_k) J(\E^{\otimes km}) (\IdentityMatrix - P_k)\|\S_{km}) + D_{\max}(ω\|\S_{km})\\
        &\leq k m R + \log \abs{\spec J(\F_m^{\otimes k})} + 0
    \end{align}

    Now, if $n$ is not a multiple of $m$ but of the form $n = km + r$, $0 \leq r < m$, we set $\tilde{\E}_n \coloneqq \tilde{\E}_{km} \otimes \E^{\otimes r}$. 
    Then, since $(\S_n)_n$ is closed under tensor products, we get that
    \begin{align}
        {1 \over n} D_{\max}(\tilde{\E}_{n} \| \S_n) &\leq {1 \over n} \qty[D_{\max}(\tilde{\E}_{km} \| \S_{km}) + r D_{\max}(\E\|\S_1)] \\ 
        & \leq {k m R \over n} + {\log \abs{\spec(J(\F_m^{\otimes k}))} \over n} + {r \over n} D_{\max}(\E\|S_1) \xrightarrow{n \to \infty} R 
    \end{align}
    since also $D_{\max}(\E\|\S_1) < \infty$ by \Autoref{def:axioms} and $\abs{\spec J(\F_m^{\otimes k})}=O(\poly(k)) = O(\poly(n))$ as explained above. Hence, 
    \begin{equation}
    \limsup_{n \to \infty} {1 \over n} D_{\max}(\tilde{\E}_{n} \| \S_n) \leq R\,.
    \end{equation}
    \medskip

    \noindent It remains to show that $\lim_{n\rightarrow\infty}\frac{1}{2}\norm{\E^{\otimes n}-\tilde{\E}_n}_{\Diamond}=0$, for which (since $\E\n - \tilde{\E}_n = \E^{\otimes r} \otimes (\E^{\otimes km} - \tilde{\E}_{km})$) it is sufficient to show that 
    \begin{equation}
        \lim_{k \to \infty} {1 \over 2} \norm{\E^{\otimes km} - \tilde{\E}_{km}}_{\Diamond} = 0\,.
    \end{equation}

    Using the strong converse part of the GQSL (specifically \Autoref{lem:GQSL_converse_bound}) with the set of channels for the alternate hypothesis being $\tilde{\mathcal{S}}_{k}=\{\F_{m}^{\otimes k}\}$, and the null hypothesis being $(\E^{\otimes m})^{\otimes k}$, we have that for any $\varepsilon\in(0,1)$
    \begin{multline}\label{eq:qrt_steins_lemma_converse}
        \limsup_{k\rightarrow\infty}-\frac{1}{km}\log( \inf_{\nu_{km}\in\mathcal{D}(R\mathcal{X}^{km})}\inf_{\substack{0\leq\Lambda\leq\mathds{1} \\ \Tr[\Lambda\E^{\otimes mk}(\nu_{km})]\geq1-\varepsilon}} \Tr[\Lambda\F_m^{\otimes k}(\nu_{km})]) \\ \leq \frac{1}{m} \sup_{\nu_m\in\mathcal{D}(R\mathcal{X}^m)}{D(\mathcal{E}^{\otimes m}(\nu_m)\|\mathcal{F}_m(\nu_m))}
    \end{multline}
    However, from the definition of $P_k$, we have that for all $k$,
    \begin{eqnarray}
        \Tr[P_k(e^{-kmR}E_k(J(\E^{\otimes mk}))-J(\F_m^{\otimes k}))] \geq 0
    \end{eqnarray}
    In particular, due to the form of the Choi states, we have that for all $x\in\mathcal{X}^{km}$, 
    \begin{eqnarray}
        \Tr\qty{\qty(\ketbra{x}{x} \otimes P_{k}^{(x)})\qty[e^{-kmR}E_k\Big(\ketbra{x}{x}\otimes\E^{\otimes mk}(\ketbra{x}{x})\Big)-\ketbra{x}{x}\otimes\F_m^{\otimes k}(\ketbra{x}{x})]} \geq 0
    \end{eqnarray}
    Hence, 
    \begin{eqnarray}
        \Tr[P_{k}^{(x)} \F_m^{\otimes k}(\ketbra{x}{x})]&=&\Tr[\ketbra{x}{x}\otimes P_{k}^{(x)} \F_m^{\otimes k}(\ketbra{x}{x})] \\
        &\leq& e^{-kmR}\Tr[(\ketbra{x}{x} \otimes P_{k,x})E_k\Big(\ketbra{x}{x}\otimes\E^{\otimes mk}(\ketbra{x}{x})\Big)] \\
        &\leq& e^{-kmR}
    \end{eqnarray}
    Thus, for all $(x_k\in\mathcal{X}^{km})_k$,
    \begin{eqnarray}
        \limsup_{k\rightarrow\infty}{-\frac{1}{km}\log(\Tr[P_{k}^{(x_k)} \F_m^{\otimes k}(\ketbra{x_k}{x_k})])} \geq R > \frac{1}{m} \sup_{\nu_m\in\mathcal{D}(R\mathcal{X}^m)}{D(\mathcal{E}^{\otimes m}(\nu_m)\|\mathcal{F}_m(\nu_m))}\,.
    \end{eqnarray}
    Thus, for all $ε \in (0,1)$ and all sequences $(x_k\in\mathcal{X}^{km})_k$ the projectors $P_k^{(x_k)}$ can be a valid choice of $\Lambda$ in the left-hand side of \eqref{eq:qrt_steins_lemma_converse} at most for finitely many $k$. In particular this means that for every $ε \in (0,1)$, we get that 
    $\Tr[P_k^{(x_k)} \E^{\otimes mk}(\ketbra{x}{x})] < 1 - ε$ eventually. Hence,
    \begin{equation}
        \lim_{k \to \infty} \Tr[P_k^{(x_k)} \E^{\otimes mk}(\ketbra{x_k}{x_k})] = 0
    \end{equation}
    or equivalently
    \begin{equation}\label{eq:qrt_robustness_proof_second_term_vanishes}
        \lim_{k \to \infty} \Tr[\qty(\IdentityMatrix - P_k^{(x_k)}) \E^{\otimes mk}(\ketbra{x_k}{x_k})]=  1\,.
    \end{equation}
    This already implies that the second term in \eqref{eq:qrt_robustness_proof_def_cq_channel} goes to zero asymptotically. 
    To see what happens to the first term, we use the gentle measurement Lemma which states that \cite[Lemma 6.15]{khatri_principles_2020}: If $\Tr(Λρ) \geq 1 - ε$ for some $0 \leq \Lambda \leq \IdentityMatrix$, then 
    \begin{equation}
        {1 \over 2} \norm{ρ - {\sqrt{\Lambda}ρ\sqrt{\Lambda} \over \Tr(\Lambda ρ)}}_1 \leq \sqrt{ε}\,.
    \end{equation}
    Since 
    \begin{equation}
        \norm{{\sqrt{\Lambda}ρ\sqrt{\Lambda} \over \Tr(\Lambda ρ)} - \sqrt{\Lambda} ρ \sqrt{\Lambda}}_1 = (1 - \Tr(Λρ)) \norm{\sqrt{\Lambda}ρ\sqrt{\Lambda} \over \Tr(\Lambda ρ)}_1 \leq ε
    \end{equation}
    we find by the triangle inequality that
    \begin{equation}
    {1 \over 2} \norm{ρ - \sqrt{\Lambda}ρ \sqrt{\Lambda}}_1 \leq \sqrt{ε} + \frac{ε}{2}
    \end{equation}
    Applying this to $\Lambda = \sqrt{\Lambda} = \IdentityMatrix - P_k^{(x_k)}$ and $ρ = \E^{\otimes mk}(\ketbra{x_k}{x_k})$
    gives
    \begin{equation}
        \lim_{k \to \infty} \norm{\E^{\otimes km}(\ketbra{x_k}{x_k}) - (\IdentityMatrix - P_k^{(x_k)})\E^{\otimes km}(\ketbra{x_k}{x_k})(\IdentityMatrix - P_k^{(x_k)})}_1 = 0
    \end{equation}
    and hence together with \eqref{eq:qrt_robustness_proof_second_term_vanishes}:
    \begin{equation}
        \lim_{k \to \infty} \norm{\E^{\otimes km}(\ketbra{x_k}{x_k}) - \tilde{\E}_{km}(\ketbra{x_k}{x_k})}_1 = 0\,.
    \end{equation}
    Since this holds for all sequences $(x_k\in\mathcal{X}^{km})_k$, we get (using \Autoref{lem:sup_over_sequences}) 
    \begin{equation}
            \lim_{k \to \infty} \sup_{x \in \X^{km}} \norm{\E^{\otimes km}(\ketbra{x_k}{x_k}) - \tilde{\E}_{km}(\ketbra{x_k}{x_k})}_1 = \lim_{k \to \infty} \norm{\E^{\otimes km} - \tilde{\E}_{km}}_{\Diamond} = 0\,.
    \end{equation}
    where the equality to the diamond norm already for input states of the form $\ketbra{x}{x}$ follows by \Autoref{lem:input_reduce_diamond}.
\end{proof}

The above two lemmas combine to give the following corollary, relating the relative entropy of resource to the log robustness:
\begin{corollary}[store={stored:cor:log_robustness}, label={cor:log_robustness}, restate-keys={manual-num={\ref{cor:log_robustness}}}]
    For any c-q channel $\E$, and any sequence of sets of c-q channels $(\S_n)_n$ satisfying the assumptions in \Autoref{def:axioms}, we have 
    \begin{multline}
        \R(\E) = \inf_{(\tilde{\E}_n)_n}\left\{\lim_{n\rightarrow\infty} \frac{1}{n}D_{\max}(\tilde{\E}_n \| \mathcal{S}_n)\,\bigg| \lim_{n\rightarrow\infty}\norm{\E^{\otimes n}-\tilde{\E}_n}_\Diamond =0 \,
        \&\, \lim_{n\rightarrow\infty} \frac{1}{n}D_{\max}(\tilde{\E}_n \| \mathcal{S}_n)\, \mathrm{  exists } \right\}
    \end{multline}
\end{corollary}
\begin{proof}
    Using \Autoref{lem:log_robustness1}, we have that 
    \begin{eqnarray}
        \R(\E)\leq \inf_{(\tilde{\E}_n)_n}\left\{ \liminf_{n\rightarrow\infty}\frac{1}{n}D_{\max}(\tilde{\E}_n \| \mathcal{S}_n)\, \bigg|\, \lim_{n\rightarrow\infty} \frac{1}{2} \norm{\E\n-\tilde{\E_n}}_\Diamond =0\right\} \\
        \leq \inf_{(\tilde{\E}_n)_n}\left\{ \lim_{n\rightarrow\infty}\frac{1}{n}D_{\max}(\tilde{\E}_n \| \mathcal{S}_n)\, \bigg|\, \lim_{n\rightarrow\infty} \frac{1}{2} \norm{\E\n-\tilde{\E_n}}_\Diamond =0\,\&\,\lim_{n\rightarrow\infty}\frac{1}{n}D_{\max}(\tilde{\E}_n \| \mathcal{S}_n)\, \mathrm{  exists }\right\}
        \label{eq:QRT_log_robustness_cor}
    \end{eqnarray}
    The second line comes from the fact that imposing the condition that $\lim_{n\rightarrow\infty}\frac{1}{n}D_{\max}(\E_n \| \mathcal{S}_n)$ exists makes our set smaller. \\
    \\
    To show the other direction, we use \Autoref{lem:log_robustness2}. This gives us a sequence $(\tilde{\E}_n')_n$ satisfying both of the following equations:
    \begin{eqnarray}
        \limsup_{n\rightarrow\infty}\frac{1}{n}D_{\max}(\tilde{\E}_n' \| \mathcal{S}_n) \leq \R(\E) \\
        \lim_{n\rightarrow\infty}\norm{\E\n-\tilde{\E}_n'}_\Diamond=0
    \end{eqnarray}
    This means that, from (\ref{eq:QRT_log_robustness_cor}),
    \begin{eqnarray}
        \limsup_{n\rightarrow\infty}\frac{1}{n}D_{\max}(\tilde{\E}_n' \| \mathcal{S}_n) = \liminf_{n\rightarrow\infty}\frac{1}{n}D_{\max}(\tilde{\E}_n' \| \mathcal{S}_n) = \lim_{n\rightarrow\infty}\frac{1}{n}D_{\max}(\tilde{\E}_n' \| \mathcal{S}_n) = \R(\E)
    \end{eqnarray}
    Hence,
    \begin{multline}
        \R(\E)=\min_{(\tilde{\E}_n)_n}\bigg\{ \lim_{n\rightarrow\infty}\frac{1}{n}D_{\max}(\tilde{\E}_n \| \mathcal{S}_n)\, \bigg|\, \lim_{n\rightarrow\infty} \frac{1}{2} \norm{\E\n-\tilde{\E_n}}_\Diamond =0\,\\\&\,\lim_{n\rightarrow\infty}\frac{1}{n}D_{\max}(\tilde{\E}_n \| \mathcal{S}_n)\, \mathrm{  exists }\bigg\}
    \end{multline}
    with a minimizer being $(\tilde{\E}_n')_n$.
\end{proof}
\subsection{Proof of \texorpdfstring{\Autoref{thm:QRT}}{Theorem \ref{thm:QRT}}}
First we will prove the direct part, which involves similar techniques to \cite[Proposition S16]{hayashi_generalized_2024}. We use the GQSL for c-q channels (\Autoref{thm:GQSL}), and \Autoref{cor:log_robustness} to construct a sequence of operations (superchannels) that achieve the desired transformation at the desired rate. We use the bounds from the GQSL for c-q channels to prove that this sequence of operations is ARNG. Meanwhile, \Autoref{cor:log_robustness} allows us to prove that the asymptotic rate of conversion for this transformation is sufficient. 
\begin{lemma}[Direct Bound on the Optimal Asymptotic Conversion Rate]
    \label{lem:QRT_direct}
    Let $\E_1$ and $\E_2$ be two c-q channels which satisfy
    \begin{eqnarray}
        \R(\E_i)>0
    \end{eqnarray}
    Then, we have
    \begin{eqnarray}
        r(\E_1\rightarrow\E_2) \geq \frac{\R(\E_1)}{\R(\E_2)}
    \end{eqnarray}
\end{lemma}
\begin{proof}
    First choose some $\delta\in(0,\min\{\R(\E_1),\R(\E_2)\})$ and define 
    \begin{eqnarray}
        r:=\frac{\R(\E_1)-\delta}{\R(\E_2)}
        \label{eq:QRT_direct_r_def}
    \end{eqnarray}
    We will show that there exists a sequence of operations $(\Theta_n)_n\in\OR$ that achieve asymptotic rate of conversion $r$. \\
    \\
    To construct this sequence, first we apply the GQSL (\Autoref{thm:GQSL}) to the channel $\E_1^{\otimes n}$ as the null hypothesis, and the set of channels $\mathcal{S}_n$ as the alternative hypothesis, for each $n$. This gives us
    \begin{align}
        \lim_{n\rightarrow\infty}\frac{1}{n}\sup_{\nu_n\in\mathcal{D}(R\mathcal{X}^n)}\inf_{\mathcal{F}_n\in \mathcal{S}_n}{D^{\varepsilon}_H(\mathcal{E}_1^{\otimes n}(\nu_n)\|\mathcal{F}_n(\nu_n))}  &=  \lim_{n\rightarrow\infty}\frac{1}{n}\sup_{\nu_n\in\mathcal{D}(R\mathcal{X}^n)}\inf_{\mathcal{F}_n\in \mathcal{S}_n}{D(\mathcal{E}_1^{\otimes n}(\nu_n)\|\mathcal{F}_n(\nu_n))} 
    \end{align}
    By \Autoref{lem:rel_entropy_exchange}, we can exchange the supremum and infimum on the right-hand side to give us $\R(\E_1)$. We can also use  \cite[Lemma S3]{hayashi_generalized_2024} on the left-hand side, to give us
    \begin{eqnarray}
        \R(\E_1)=\lim_{n\rightarrow\infty}-\frac{1}{n}\log \left( \inf_{\nu_n\in\mathcal{D}(R\mathcal{X}^n)}\inf_{\substack{0\leq\Lambda\leq\mathds{1}\\ \Tr[\Lambda\E_1^{\otimes n}(\nu_n)]\geq1-\varepsilon}} \sup_{\mathcal{F}_n\in \mathcal{S}_n}\Tr[\Lambda\mathcal{F}_n(\nu_n)]\right)
    \end{eqnarray}
    This means that, for sufficiently large $n$, 
    \begin{eqnarray}
        \inf_{\nu_n\in\mathcal{D}(R\mathcal{X}^n)}\inf_{\substack{0\leq\Lambda\leq\mathds{1}\\ \Tr[\Lambda\E_1^{\otimes n}(\nu_n)]\geq1-\varepsilon}} \sup_{\mathcal{F}_n\in \mathcal{S}_n}\Tr[\Lambda\mathcal{F}_n(\nu_n)] \leq \exp\left(-n\R(\E_1)+\frac{n\delta}{3}\right)
    \end{eqnarray}
    This means that there exists a sequence $(\varepsilon_n)_n$ of type I error parameters, with $\varepsilon_n\rightarrow0$ as $n\rightarrow\infty$, and there exists a sequence $(\nu_n)_n$ of input states and $(\Lambda_n)_n$ of POVMs such that both of the following equations hold for sufficiently large $n$:
    \begin{eqnarray}
        \label{eq:QRT_direct_I_error_bound}
        \Tr[(\mathds{1}-\Lambda_n)\E_1^{\otimes n}(\nu_n)] \leq \varepsilon_n \\
        \sup_{\F_n\in\mathcal{S}_n}\Tr[\Lambda_n\F_n(\nu_n)]\leq \exp\left(-n\R(\E_1)+\frac{n\delta}{3}\right)
        \label{eq:QRT_direct_trace_bound}
    \end{eqnarray}
    Now, by \Autoref{cor:log_robustness}, there exists a sequence $(\E_2^{(rn)})_n$ of c-q channels satisfying both of 
    \begin{eqnarray}
        \label{eq:QRT_direct_LR_bound}
        \R(\E_2)=\lim_{n\rightarrow\infty}\frac{1}{\lceil rn \rceil}D_{\max}(\E_2^{(rn)}\| \mathcal{S}_{\lceil rn\rceil}) \\
        \lim_{n\rightarrow\infty}{\norm{\E_2^{\otimes \lceil rn \rceil}-\E_2^{(rn)}}_\Diamond}=0
        \label{eq:QRT_direct_diamond_bound}
    \end{eqnarray}
    Let $\F^{(rn)}$ be an optimal channel that minimizes $\inf_{\F\in\mathcal{S}_{\lceil rn\rceil}}D_{\max}(\E_2^{(rn)}\|\F)$, and let $\tilde{\E}^{(rn)}$ be the corresponding c-q channel such that
    \begin{eqnarray}
        \frac{\E_2^{(rn)}+(e^{D_{\max}(\E_2^{(rn)}\| \mathcal{S}_{\lceil rn\rceil})}-1)\tilde{\E}_2^{(rn)}}{e^{D_{\max}(\E_2^{(rn)}\| \mathcal{S}_{\lceil rn\rceil})}} =\F^{(rn)}
        \label{eq:QRT_direct_RG_extremizer}
    \end{eqnarray}
    holds. Now, we define our superchannels $\Theta_n$ as follows:
    \begin{eqnarray}
        \Theta_n(\mathcal{N}) := \Tr[\Lambda_n\mathcal{N}(\nu_n)]\E_2^{(rn)}+\Tr[(\mathds{1}-\Lambda_n)\mathcal{N}(\nu_n)]\tilde{\E}_2^{(rn)}
        \label{eq:QRT_direct_theta_def}
    \end{eqnarray}
    We can see this map in terms of pre and post processing channels. The pre-processing is that we input $\nu_n$ into the c-q channel $\mathcal{N}$, and then measure this with respect to the POVM $\{\Lambda_n, \mathds{1}-\Lambda_n\}$. Depending on the result of the measurement, we output the channel $\tilde{\E}_2^{(rn)}$ or $\E_2^{(rn)}$. \\
    \\
    Now, we prove that $(\Theta_n)_n$ is a sequence of ARNG operations.\\
    \\
    Take any sequence $(\F_n)_n$ of free channels $\F_n\in\mathcal{S}_n$. Define 
    \begin{eqnarray}
        s_n&:=& D_{\max}(\E_2^{(rn)}\| \mathcal{S}_{\lceil rn \rceil}) \\
        t_n&:=& \Tr[\Lambda_n\F_n(\nu_n)]
    \end{eqnarray}
    Then, by (\ref{eq:QRT_direct_LR_bound}), we have for sufficiently large $n$,
    \begin{eqnarray}
        \label{eq:QRT_direct_sn}
        nr\left(\R(\E_2)+\frac{\delta}{3r}\right) \geq s_n\geq nr\left(\R(\E_2)-\frac{\delta}{3r}\right)
    \end{eqnarray}
    The right hand side goes to infinity as $n\rightarrow\infty$, using (\ref{eq:QRT_direct_r_def}). \\
    \\
    Now, by (\ref{eq:QRT_direct_trace_bound}), we have that, for sufficiently large $n$,
    \begin{eqnarray}
        t_n\leq \sup_{\F_n\in\mathcal{S}_n}\Tr[\Lambda_n\F_n(\nu_n)]\leq \exp\left(-n\R(\E_1)+\frac{n\delta}{3}\right)
        \label{eq:QRT_direct_tn}
    \end{eqnarray}
    with the right-hand side tending towards 0 as $n\rightarrow\infty$. \\
    \\
    Now, by (\ref{eq:QRT_direct_r_def}) and (\ref{eq:QRT_direct_sn}), we have that 
    \begin{eqnarray}
        e^{-s_n}
        \geq \exp\left(-n\left(r\R(\E_2)+\frac{\delta}{3}\right)\right) 
        =\exp\left(-n\left(\R(\E_1)-\frac{2\delta}{3}\right)\right)
        \label{eq:QRT_direct_sn_2}
    \end{eqnarray}
    Therefore, by (\ref{eq:QRT_direct_tn}) and (\ref{eq:QRT_direct_sn_2}), we have that 
    \begin{eqnarray}
        e^{-s_n}-t_n \geq\exp\left(-n\left(\R(\E_1)-\frac{2\delta}{3}\right)\right)-\exp\left(-n\R(\E_1)+\frac{n\delta}{3}\right) \geq 0
    \end{eqnarray}
    Now, we have, by (\ref{eq:QRT_direct_theta_def}),
    \begin{eqnarray}
        \Theta_n(\F_n)=t_n\E_2^{(rn)}+(1-t_n)\tilde{\E}_2^{(rn)}
    \end{eqnarray}
    By (\ref{eq:QRT_direct_RG_extremizer}), we have
    \begin{eqnarray}
        \frac{\E_2^{(rn)}+(e^{s_n}-1)\tilde{\E}_2^{(rn)}}{e^{s_n}} =\F^{(rn)}\in \mathcal{S}_{\lceil rn\rceil}
    \end{eqnarray}
    This means that $s=\frac{e^{-s_n}-t_n}{(e^{s_n}-1)e^{-s_n}}$ satisfies
    \begin{eqnarray}
        \frac{\Theta_n(\F_n)+s\E_2^{(rn)}}{1+s}=\F^{(rn)} \in \mathcal{S}_{\lceil rn\rceil}
        \label{eq:QRT_direct_s_bound}
    \end{eqnarray}
    So, we get
    \begin{eqnarray}
        D_{\max}(\Theta_n(\F_n)\|\mathcal{S}_{\lceil rn \rceil})
        &\overset{(a)}{\leq}&\log(1+s)\\
        &=& \log\left(\frac{1-t_n}{1-e^{-s_n}}\right)\\
        &\overset{(b)}{\rightarrow}&0
    \end{eqnarray}
    Here, (a) is due to (\ref{eq:QRT_direct_s_bound}), and (b) is due to (\ref{eq:QRT_direct_sn}) and (\ref{eq:QRT_direct_tn}).\\
    \\
    Thus, $(\Theta_n)_n$ is a sequence of ARNG operations. \\
    \\
    To prove that $(\Theta_n)_n$ achieves the required asymptotic conversion rate, note that 
    \begin{eqnarray}
        \Theta_n(\E_1\n) = \Tr[\Lambda_n\E_1\n(\nu_n)]\E_2^{(rn)}+\Tr[(\mathds{1}-\Lambda_n)\E_1\n(\nu_n)]\tilde{\E}_2^{(rn)}
    \end{eqnarray}
    So, we have that
    \begin{eqnarray}
        &&\norm{\Theta_n(\E_1\n)-\E_2^{\otimes \lceil rn\rceil}}_\Diamond  \nonumber\\
        &\overset{(a)}{\leq}& \Tr[\Lambda_n\E_1\n(\nu_n)] \norm{\E_2^{(rn)}-\E_2^{\otimes \lceil rn \rceil}}_\Diamond+ \Tr[(\mathds{1}-\Lambda_n)\E_1\n(\nu_n)] \norm{\tilde{\E}_2^{(rn)}-\E_2^{\otimes \lceil rn \rceil}}_\Diamond \nonumber\\
        &\overset{(b)}{\leq}& \Tr[\Lambda_n\E_1\n(\nu_n)] \norm{\E_2^{(rn)}-\E_2^{\otimes \lceil rn \rceil}}_\Diamond+ \Tr[(\mathds{1}-\Lambda_n)\E_1\n(\nu_n)] \left( \norm{\tilde{\E}_2^{(rn)}}_\Diamond + \norm{\E_2^{\otimes \lceil rn \rceil}}_\Diamond\right) \nonumber\\
        &\overset{(c)}{\leq}& (1-\varepsilon_n)\norm{\E_2^{(rn)}-\E_2^{\otimes \lceil rn \rceil}}_\Diamond + 2\varepsilon_n \nonumber\\
        &\overset{(d)}{\rightarrow}& 0
    \end{eqnarray}
    Here, (a) is by the triangle inequality and homogeneity, (b) is again by the triangle inequality, (c) is by (\ref{eq:QRT_direct_I_error_bound}), and (d) is by (\ref{eq:QRT_direct_diamond_bound}) and since $\varepsilon_n\rightarrow0$ as $n\rightarrow\infty$.\\
    \\
    This gives us the asymptotic rate of conversion as $r$, as given by (\ref{eq:QRT_direct_r_def}), as required.
\end{proof}
Before we prove the converse part, we need to prove that the relative entropy of resource is non-increasing under sequences of ARNG operations. The use of the diamond norm here removes the condition of asymptotic continuity on the set of asymptotically free operations that was needed in \cite[Lemma S17]{hayashi_generalized_2024}, but otherwise the same reasoning is used. We will use \Autoref{lem:resource_monotone} to bound $R(\Theta_n(\E_1\n))$ in \Autoref{lem:QRT_converse}, in order to prove the converse bound.
\begin{lemma}[Relative Entropy of Resource Decreases under Sequences of ARNG Operations]
    \label{lem:resource_monotone}
    Let $(\Theta_n)_n\in \tilde{\mathcal{O}}$ be a sequence of ARNG operations. Then for any c-q channel $\E$, we have that
    \begin{eqnarray}
        \R(\E) \geq \liminf_{n\rightarrow\infty}\frac{1}{n}R(\Theta_n(\E^{\otimes n}))
    \end{eqnarray}
\end{lemma}
\begin{proof}
    From \Autoref{lem:log_robustness1}, we have
    \begin{eqnarray}\label{eq:qrt_converse_monotonicity}
        \liminf_{n\rightarrow\infty}\frac{1}{n}R(\Theta_n(\E\n)) \leq \inf_{(\E_n)_n}\left\{ \liminf_{n\rightarrow\infty}\frac{1}{n}D_{\max}(\tilde{\E}_n \| \mathcal{S}_n) \bigg| \lim_{n\rightarrow\infty} \frac{1}{2} \norm{\Theta_n(\E\n)-\tilde{\E_n}}_\Diamond =0\right\}
    \end{eqnarray}
    So, it suffices to show that the right-hand side is further upper-bounded by $\R(\E)$. Using \Autoref{cor:log_robustness}, there exists a sequence $(\tilde{\E}_n)_n$ of c-q channels satisfying both of the following two equations:
    \begin{eqnarray}
        \label{eq:QRT_monotone_LR}
        \lim_{n\rightarrow\infty}\frac{1}{n}D_{\max}(\tilde{\E}_n \| \mathcal{S}_n)= \R(\E) \\
        \lim_{n\rightarrow\infty}\norm{\tilde{\E}_n-\E\n}_\Diamond=0
    \end{eqnarray}
    It is well known that we can write each superchannel $\Theta_n$ as the composition of the input channel with a pre-processing channel $\mathcal{N}_n^{pre}: \X \to \X E$ and a post-processing channel $\mathcal{N}_n^{post}: AE \to A$ (where the system E is not acted on by the channel to be processed) \cite{gour_comparison_2019}. Hence
    \begin{eqnarray}
        \norm{\Theta_n(\E\n)-\Theta_n(\tilde{\E}_n)}_\Diamond &=& \norm{\mathcal{N}_n^{post}\circ\E\n\circ\mathcal{N}_n^{pre}-\mathcal{N}_n^{post}\circ\tilde{\E}_n\circ\mathcal{N}_n^{pre}}_\Diamond \nonumber\\
        &\overset{(a)}{\leq}& \norm{\E\n\circ\mathcal{N}_n^{pre}-\tilde{\E}_n\circ\mathcal{N}_n^{pre}}_\Diamond \nonumber\\
        &=& \sup_{\nu_n\in\mathcal{D}(R\mathcal{X}^n)}\norm{\E\n\circ\mathcal{N}_n^{pre}(\nu_n)-\tilde{\E}_n\circ\mathcal{N}_n^{pre}(\nu_n)}_1 \nonumber\\
        &\overset{(b)}{\leq}& \sup_{\nu_n\in\mathcal{D}(R\mathcal{X}^n)}\norm{\E\n(\nu_n)-\tilde{\E}_n(\nu_n)}_1 \nonumber\\
        &=& \norm{\E\n-\tilde{\E}_n}_\Diamond
    \end{eqnarray}
    Here, (a) is by the data processing inequality for the diamond norm, and (b) is since the inputs of the channels $\E\n$ and $\tilde{\E}_n$ are restricted by the pre-processing channel. Therefore,
    \begin{eqnarray}
        \lim_{n\rightarrow\infty}\norm{\Theta_n(\tilde{\E}_n)-\Theta_n(\E\n)}_\Diamond=0
        \label{eq:QRT_monotone_diamond}
    \end{eqnarray}
    This argument allows us to remove the assumption of asymptotic continuity that was required in \cite{hayashi_generalized_2024}. \\
    \\
    Define $r_n:=e^{D_{\max}(\tilde{\E}_n\|\mathcal{S}_n)}-1$. So, there exists a c-q channel $\E_n'$ such that 
    \begin{eqnarray}
        \frac{\tilde{\E}_n+r_n\E_n'}{1+r_n}=\F_n \in \mathcal{S}_n
    \end{eqnarray}
    is a free channel. Since $(\Theta_n)_n$ is a sequence of ARNG operations, we have $D_{\max}(\Theta_n(\F_n)\| \mathcal{S}_n)  \rightarrow0$. Let
    \begin{eqnarray}
        \label{eq:QRT_monotone_ARNG}
        r'_n:=e^{D_{\max}\left(\Theta_n(\F_n)\| \mathcal{S}_n\right)}-1
    \end{eqnarray}
    Now, there exists another c-q channel $\E_n''$ such that 
    \begin{eqnarray}
        \frac{\Theta_n\left(\frac{\tilde{\E}_n+r_n\E_n'}{1+r_n}\right)+r'_n\E_n''}{1+r'_n}= \frac{\Theta_n(\tilde{\E}_n)+r_n\Theta_n(\E_n')+(1+r_n)r_n'\E_n''}{(1+r_n)(1+r_n')}
    \end{eqnarray}
    is a free c-q channel. Thus,
    \begin{eqnarray}
        D_{\max}(\Theta_n(\tilde{\E}_n)\| \mathcal{S}_n)\leq D_{\max}(\Theta_n(\F_n)\| \mathcal{S}_n)+D_{\max}(\tilde{\E}_n\| \mathcal{S}_n) 
    \end{eqnarray}
    Hence, by (\ref{eq:QRT_monotone_LR}) and (\ref{eq:QRT_monotone_ARNG}), we have that
    \begin{eqnarray}
        \liminf_{n\rightarrow\infty}\frac{1}{n}D_{\max}(\Theta_n(\tilde{\E}_n) \| \mathcal{S}_n) \leq \liminf_{n\rightarrow\infty}\frac{1}{n}[D_{\max}(\Theta_n(\F_n)\| \mathcal{S}_n)+D_{\max}(\tilde{\E}_n\| \mathcal{S}_n) ]=\R(\E)
    \end{eqnarray}
    Since the sequence $(\Theta_n(\tilde{\E}_n))_n$ satisfies (\ref{eq:QRT_monotone_diamond}) it is a viable sequence for the right-hand side of \eqref{eq:qrt_converse_monotonicity}, and so we get that
    \begin{eqnarray}
        \inf_{(\E_n)_n}\left\{ \liminf_{n\rightarrow\infty}\frac{1}{n}D_{\max}(\tilde{\E}_n \| \mathcal{S}_n)\,\bigg| \,\lim_{n\rightarrow\infty} \frac{1}{2} \norm{\Theta_n(\E\n)-\tilde{\E_n}}_\Diamond =0\right\} \leq \R(\E)
    \end{eqnarray}
    as required.
\end{proof}
Finally we prove the converse, which is analogous to \cite[Proposition S18]{hayashi_generalized_2024}.
\begin{lemma}(Converse Bound on the Asymptotic Conversion Rate)
    \label{lem:QRT_converse}
    For any c-q channels $\E_1$ and $\E_2$, with $\R(\E_2)>0$, it holds that 
    \begin{eqnarray}
        r(\E_1\rightarrow\E_2)\leq \frac{\R(\E_1)}{\R(\E_2)}
    \end{eqnarray}
\end{lemma}
\begin{proof}
    Choose an achievable rate $r<r(\E_1\rightarrow\E_2)$. Then we have a sequence of ARNG operations $(\Theta_n)_n\in\OR$ satisfying
    \begin{eqnarray}
        \limsup_{n\rightarrow\infty}\norm{\Theta_n(\E_1\n)-\E_2^{\otimes \lceil rn \rceil}}_\Diamond=0
        \label{eq:QRT_converse_diamond}
    \end{eqnarray}
    Then, we have that
    \begin{eqnarray}
        \R(\E_1)&\overset{(a)}{\geq}&\liminf_{n\rightarrow\infty}\frac{1}{n}R(\Theta_n(\E_1^{\otimes n})) \\
        &\overset{(b)}{=}& \lim_{n\rightarrow\infty}\frac{1}{n}R(\E_2^{\otimes \lceil rn \rceil}) \\
        &\overset{(c)}{=}& r\R(\E_2)
    \end{eqnarray}
    Here, (a) is by \Autoref{lem:resource_monotone}, (b) is by (\ref{eq:QRT_converse_diamond}) and \Autoref{lem:resource_close_sequences}, and (c) can be easily seen from the definition of $R^{\infty}$.
    This gives our required bound.
\end{proof}
\noindent Hence, by combining \Autoref{lem:QRT_converse} and \Autoref{lem:QRT_direct}, we have proven \Autoref{thm:QRT}.
\subsection{Comparison to QRTs via Choi states}\label{sec:QRT_comparison}
As already explained in \zcref{sec:intro_QRT} our treatment of QRTs for C-Q channels is different from the treatment in \cite{hayashi_generalized_2024}. In this section we present some explicit examples to show how our QRTs for c-q channels differ from the construction based on Choi states from \cite{hayashi_generalized_2024}. 
Throughout this section, we fix a very simple QRT for c-q channels with the following set of free channels: 
Let $\mathcal{X}=\{0,1\}$, and let $A=\mathbb{C}^2$, with orthonormal basis $\{\ket{0}, \ket{1}\}$. Let $\F:\mathcal{X}\rightarrow A$ be the completely depolarizing channel (i.e., $\F(\nu)=\1/2$ for all $\nu\in\mathcal{D}(\mathcal{X})$). We define the sets of free channels $\mathcal{S}_n\subset\CPTP(\mathcal{X}^n\rightarrow A^n)$ to contain only the element $\F\n$. It is easy to verify that $(\mathcal{S}_n)_n$ satisfies assumptions (1) to (4) in \Autoref{def:axioms}. \zcref{lem:divergence_input_reduction} also directly implies that for every $\E \in CQ$: $R^{\infty}(\E) = R(\E)$.

We first provide an example that illustrates that using the diamond norm as a distance measure between channels (when considering whether one resource can be approximately transformed into another) leads to a relative entropy of resource ({\em{viz.}}~$R(\E) = D(\E\|\S_1)$) which is different from that considered in \cite{hayashi_generalized_2024} ({\em{viz.}} $\tilde{R}(\E) = \inf_{\F \in \S} D(J(\E)\|J(\F))$). In fact, by definition, we must have that $R(\E)\geq \tilde{R}(\E)$ for any channel $\E$, since the former has an additional optimization over input states.


\begin{example}[Relative entropy of resource can be different]
    Let $\X=\{0,1\}$, and let $A=\mathbb{C}^2$ have orthonormal basis $\{\ket{0},\ket{1}\}$. We define $\E^{(1)}$, $\E^{(2)} \in \CQ$ as follows:
    \begin{eqnarray}
        \E^{(1)}(\ketbra{x}{x})&:=& \begin{cases}
        \ketbra{0}{0}\:\,\text{if}\:\,x = 0 \\
        \IdentityMatrix/2\:\,\text{if}\:\,x = 1
        \end{cases}\\
        \E^{(2)}(\ketbra{x}{x})&:=&\ketbra{0}{0}\:\forall\, x\in\mathcal{X} \\
    \end{eqnarray}
    Then for $\S_1=\{\F\}$, where $\F:\X\to A$ is the completely depolarizing channel, we easily find that
    \begin{equation}
    R(\E^{(1)}) = D(\E^{(1)}\|\F) = D(\E^{(1)}(\ketbra{0}{0})\|\F(\ketbra{0}{0})) = D(\ketbra{0}{0} \| \IdentityMatrix /2) = \log(2)
    \end{equation}
    whereas by the direct sum property of $D$
    \begin{equation}
    {\tilde{R}}(\E^{(1)}) = (D(J(\E^{(1)})\|J(\F)) = {1 \over 2} D(\ketbra{0}{0} \| \IdentityMatrix / 2) + {1 \over 2} D(\IdentityMatrix/2 \| \IdentityMatrix /2) = {1 \over 2} \log(2)\,.
    \end{equation}
    Additionally, these two different resource quantifiers cannot be scalar multiples of each other since
    \begin{equation}
    R(\E^{(2)}) = D(\E^{(2)}\|\F) = \log(2)
    \end{equation}
    and
    \begin{equation}
        \tilde{R}(\E^{(2)})= D(J(\E^{(2)})\|J(\F)) = D(\ketbra{0}{0}\|\IdentityMatrix/2) = \log(2)\,
    \end{equation}
    and so for the channel $\E^{(2)}$ both constructions give the same relative entropy of resource.
\end{example}

In the next example we show that the restriction to what they consider {\em{asymptotically continuous}} free operations is a real restriction, i.e., there are superchannels that do not satisfy this property. This demonstrates that our result can be applied to a wider range of operations than the analogous result in \cite[Lemma S12]{hayashi_generalized_2024}. 

The authors of~\cite{hayashi_generalized_2024} define asymptotic continuity for superchannels $(\Theta_n)_n$ as follows: For every sequence of channels $\E_n$ and $\tilde{\E}_n$ such that 
\begin{equation}
    \lim_{n \to \infty} \norm{J(\E_n) - J(\tilde{\E}_n)}_1 = 0
\end{equation}
it holds that
\begin{equation}
    \lim_{n \to \infty} \norm{J(\Theta_n(\E_n)) - J(\Theta_n(\tilde{\E}_n))}_1 = 0\,.
\end{equation}
The idea of our example is that we consider two sequences of channels (say, $(\E_n^{(1)})_n$ and $(\E_n^{(2)})_n$) which we choose in such a way that their outputs for each $n$ are orthogonal for one specific input state, but are equal for all other inputs. This will imply that the trace distance between their Choi states (which can be seen as an average-case trace distance between their outputs) vanishes in the asymptotic limit ($n \to \infty$). However, if we apply a superchannel that picks out this specific input, then the trace distance between the Choi states of the channels which result from the action of the superchannel on $\E_n^{(1)}$ and $\E_n^{(2)}$ no longer vanishes in the asymptotic limit. We show below that, for a certain choice of free channels, such a superchannel is resource non-generating for all $n$, and hence a sequence of such superchannels is ARNG. Hence, asymptotic continuity does not hold.   

\begin{example}[Not every ARNG sequence of superchannels is asymptotically continuous in the sense of \cite{hayashi_generalized_2024}]
We now define two sequences of c-q channels in the following way:
    \begin{eqnarray}
        \E_n^{(1)}(\ketbra{x}{x})&:=&\ketbra{0}{0}\n\:\forall\, x\in\mathcal{X}^n \\
        \E_n^{(2)}(\ketbra{x}{x})&:=& \begin{cases}
        \ketbra{0}{0}\n\:\text{if}\:x\neq (1,1,1,..., 1) \\
        \ketbra{1}{1}\n\:\text{if}\:x=(1,1,1,..., 1)
        \end{cases}
    \end{eqnarray}
    These two c-q channels have the following Choi states:
    \begin{eqnarray}
        J(\E_n^{(1)})&=&\tau_n\otimes\ketbra{0}{0}\n \\
        J(\E_n^{(2)})&=&(\tau_n-2^{-n}\ketbra{1}{1}\n)\otimes\ketbra{0}{0}\n+2^{-n}\ketbra{1}{1}\n\otimes\ketbra{1}{1}\n 
    \end{eqnarray}
    where $\tau_n:=\1_{\mathcal{X}^n}/2^n$. Hence,
    \begin{eqnarray}
        \norm{ J(\E_n^{(1)})- J(\E_n^{(2)})}_1=\norm{2^{-n}(\ketbra{1}{1}\n\otimes\ketbra{0}{0}\n-\ketbra{1}{1}\n\otimes\ketbra{1}{1}\n)}_1=2^{1-n}\rightarrow0
    \end{eqnarray}
    as $n\rightarrow\infty$. \\
    \\
    We introduce the channel $\mathcal{N}^{pre}_n:\mathcal{X}^n\rightarrow\mathcal{X}^n$, where $\mathcal{N}^{pre}_n(\ketbra{x}{x})=\ketbra{1}{1}\n$ for all $x\in\mathcal{X}^n$. We define the corresponding superchannel $\Theta_n:\CPTP(\mathcal{X}^n\rightarrow A^n)\rightarrow \CPTP(\mathcal{X}^n\rightarrow A^n)$ by $\Theta_n(\mathcal{N}):=\mathcal{N}\circ\mathcal{N}_n^{pre}$. \\
    \\
    Since $\Theta_n(\F\n)=\F\n\circ\mathcal{N}_n^{pre}=\F\n$, we have that $D_{\max}(\Theta_n(\F\n) \|\S_n)=0$, and hence $(\Theta_n)_n$ is a sequence of ARNG operations (in fact, each $\Theta_n$ is resource non-generating). \\
    \\
    However, we have that 
    \begin{eqnarray}
        J(\Theta_n(\E_n^{(1)}))=\tau_n\otimes\ketbra{0}{0}\n \\
        J(\Theta_n(\E_n^{(1)}))=\tau_n\otimes\ketbra{1}{1}\n
    \end{eqnarray}
    So, 
    \begin{eqnarray}
        \norm{J(\Theta_n(\E_n^{(1)}))-J(\Theta_n(\E_n^{(2)}))}_1 = \norm{\tau_n\otimes(\ketbra{0}{0}\n-\ketbra{1}{1}\n)}_1 = 2
    \end{eqnarray}
    for all $n$. \\
    \\
    Hence, the sequence of superchannels $(\Theta_n)_n$ is ARNG but not asymptotically continuous.
\end{example}

\subsection{Examples of QRTs for C-Q Channels}\label{sec:QRT_examples}
We now present some examples of QRTs for c-q channels to which our result, \Autoref{thm:QRT}, applies.\\
One natural example of free c-q channels satisfying the assumptions in \Autoref{def:axioms} is the class of c-q channels whose outputs are bipartite quantum states which are separable. It is defined as follows:
\begin{example}[c-q channels yielding separable output states]
    Let $\S_n\subset\CPTP(\mathcal{X}^n\rightarrow A^nB^n)$ be the set of c-q channels 
    \begin{eqnarray}
        \S_n:=\Set{\F\in\CPTP(\mathcal{X}^n\rightarrow A^nB^n) | \F(\ketbra{x}{x})\in\SEP(A^n:B^n)\,\forall x\in\mathcal{X}^n }
    \end{eqnarray}
    Here, $\SEP(A:B)$ is the set of all quantum states that are separable across the bipartition $A:B$ of the system $AB$, i.e.
    \begin{eqnarray}
        \rho_{AB}\in\SEP(A:B) \iff \rho_{AB}=\sum_{y\in\mathcal{Y}}{p(y)\sigma_A^{(y)}\otimes\varphi_B^{(y)}}
    \end{eqnarray}
    for some $\sigma_A^{(y)}\in\mathcal{D}(A)$, $\varphi_B^{(y)}\in\mathcal{D}(B)$, and  some probability distribution $p:\mathcal{Y}\rightarrow[0,1]$. \\
    \\
    It is easy to verify that $(\S_n)_n$ satisfies all the assumptions in \Autoref{def:axioms}. \\
    
    If $\abs{\mathcal{X}}=1$, then this construction reduces to a QRT for entanglement of states. This is because a channel with only a single possible input state is equivalent to a quantum state (namely, the corresponding unique output state). Moreover, the reference system does not play a role in this case, since  if $\X=\{0\}$, then for any input $\nu_n=\nu_R\otimes\ketbra{0}{0}\n\in\mathcal{D}(R\mathcal{X}^n)$ and any c-q channel $\E_n\in\CPTP(\X^n\to A^nB^n)$, we have that $\E_n(\nu_n)=\nu_R\otimes\E(\ketbra{0}{0}\n)$. Since $D$ and $D_{\max}$ are faithful and additive under tensor products, we can thus ignore the reference system. Hence, this example of a QRT can be seen as an extension of the QRT for entanglement to c-q channels. 
    
    Lami and Regula proved  in~\cite{lami_no_2023} that for mixed states entanglement is not reversible under the class of RNG operations\footnote{Recall that these are operations that map free objects to free objects for every $n$.}. This result, along with the above-mentioned relation between the QRTs, imply that also the QRT for c-q channels cannot be reversible under just RNG operations. Hence, ARNG operations seem to be again the smallest meaningful class of operations under which reversibility holds.
    
    \end{example}

In fact this construction can be applied to any resource theory for quantum states that satisfies the original Brandao-Plenio axioms \cite{brandao_reversible_2015}, in particular this includes the resource theories of coherence, magic, athermality and many others (see e.g. \cite{chitambar_quantum_2019} for an overview):
    \begin{example}
        Let $(S_n \subset\DM[\HS\n])_n$ be any sequence of sets that satisfy the following conditions: 
        \begin{enumerate}
            \item Each set $S_n$ is closed and convex as a subset of $\DM[\HS\n]$.
            \item The sets $(S_n)$ are closed under tensor products, i.e., if $σ_n \in S_n$ and $σ_m \in S_m$ then $σ_n \otimes σ_m \in S_{n + m}$. 
            \item The set $S_1$ contains a full-rank state.
            \item For any state $σ_n \in S_{n}$ and a permutation $\pi \in \Sn$ with corresponding representation $P_{\HS}(\pi)$ that permutes the $n$ copies of $\HS$, also $P_{\HS}(\pi)σ_n P_{\HS}(\pi)^\dagger \in S_n$.
        \end{enumerate}
        Then, the following set of c-q channels satisfies the assumptions of \zcref{def:axioms}, and hence serves as the set of free channels for a reversible QRT in the sense described above:
        \begin{equation} \S_n \coloneqq \Set{\F\in\CPTP(\mathcal{X}^n\rightarrow A^nB^n) | \F(\ketbra{x}{x})\in S_n\,\forall x\in\mathcal{X}^n }
        \end{equation}
    \end{example}

    Note that the applicability of our statements is not limited to such constructions. To illustrate this, 
    let us consider the free set to consist of all replacer channels. This leads us to the QRT for the (classical) capacity of c-q channels, that is a QRT in which the (regularized) relative entropy of resource of a c-q channel is given by its capacity to transmit classical information:
    \begin{example}[Resource theory of classical capacity for c-q channels]
        Consider the resource theory where the free channels are all replacer channels, i.e.
        \begin{equation}
            \S_n \coloneqq \Set{\F \in \CQ[\X^n \to A^n] | \exists\, σ_n \in \DM[A^n]: \F(\ketbra{x}{x}) = σ_n \,\forall x \in \X^n}
        \end{equation}
        For any probability distribution $p_{\X}:=\{p(x)\}_{x \in \X}$, let $\omega_{R\X}^{(p)} \coloneqq \sum_{x \in \X} p(x) \ketbra{xx}{xx}_{R \X}$.
        One then finds that for any channel $\E:\X\to A$, by \zcref{lem:optimal_cq_input_state}, the relative entropy of resource (see (\ref{eq:rel_entropy_resource_def})) satisfies:
        \begin{align}
            R(\E) := D(\E\|\S_1) &= \sup_{\nu_{R\X} \in \DM[R\X]} \inf_{\F \in \S_1} D(\E(\nu)\|\F(\nu)) = \sup_{p_{\X}} \inf_{\F \in \S_1} D(\E(\omega_{R\X}^{(p)})\|\F(\omega_{R\X}^{(p)})) \nonumber\\ 
            &= \sup_{p_{\X}} \inf_{σ_1 \in \DM[A]} D(\E(\omega_{R\X}^{(p)})\| \omega_{R}^{(p)} \otimes σ_1) = \sup_{p_{\X}} I(R:A)_{\E(\omega_{R \X}^{(p)})} \equiv C(\E),
        \end{align}
        where $C(\E)$ is the capacity of $\E$ \cite{holevo_capacity_1998,schumacher_sending_1997}. The intuition behind this relation is that every c-q channel with zero capacity is a replacer channel, and so the sets $(\S_n)_n$ are exactly the free sets in the QRT of capacity for c-q channels. Moreover, the capacity of c-q channels is known to be additive \cite{wilde_strong_2014}, and so $R^{\infty}(\E) = R(\E)$.
        In particular, our \zcref{thm:QRT} then implies that under ARNG operations this QRT is reversible. This is essentially a reverse Shannon theorem for c-q channels. To see this, note that the capacity can be defined as the rate $r(\E \to \id_{\X})$ under local pre- and post-processing operations (which can be easily seen to preserve replacer channels). The statement that the reverse is possible, i.e.\ the rate $r(\id_{\X} \to \E)$ is given by the inverse of the capacity, is known as a reverse Shannon theorem. This is usually studied and shown in a setting where the set of allowed (i.e.\ free) operations for this reverse task is given by local operations together with pre-shared entanglement \cite{berta_quantum_2011, bennett_quantum_2014}, which are also easily seen to be a subclass of RNG operations for the free set of replacer channels.
        
        Note that in \cite{hayashi_generalized_2024} the authors also consider these sets of free channels $(\S_n)_n$ as a particular example. However, their construction based on Choi states leads to a different theory, where the input distribution $p_{\X}$ from above is fixed to be a uniform distribution. In this case, the corresponding relative entropy of resource ($\tilde{R}(\E)$) is also no longer the capacity of the c-q channel $\E$.
    \end{example}

\appendix
\section{Proofs of Technical Lemmas} \label{Appendix}
\begin{lemma}[name=Input reduction for c-q channels with divergences, store={stored:lem:divergence_input_reduction}, label={lem:divergence_input_reduction}, restate-keys={manual-num={\ref{lem:divergence_input_reduction}}}]
Let $\mathcal{E},\mathcal{F}:\mathcal{X}\rightarrow A$ be c-q channels. Let $\mathbf{D}$ be a quantum divergence that is jointly quasi-convex, sub-additive under tensor products, and faithful. Then
\begin{equation}
    \sup_{\nu\in\mathcal{D}(R\mathcal{X})}{\mathbf{D}(\mathcal{E}(\nu)\|\mathcal{F}(\nu))} =  \sup_{x\in\mathcal{X}}{\mathbf{D}(\mathcal{E}(\ketbra{x}{x})\|\mathcal{F}(\ketbra{x}{x}))}
\end{equation}
\end{lemma}

\begin{proof}
Clearly, by elementary properties of suprema,
\begin{equation}
    \sup_{\nu\in\mathcal{D}(R\mathcal{X})}{\mathbf{D}(\mathcal{E}(\nu)\|\mathcal{F}(\nu))} \geq  \sup_{x\in\mathcal{X}}{\mathbf{D}(\mathcal{E}(\ketbra{x}{x})\|\mathcal{F}(\ketbra{x}{x}))}
\end{equation}
So, it remains to show the reverse inequality. \\
\\
Fix some (arbitrary) $\nu\in\mathcal{D}(R\mathcal{X})$. Then, we define
\begin{equation}
    \tr_{\mathcal{X}}\left[\nu\ketbra{x}{x}\right] := p(x,\nu)\rho_R^{(x)}
\end{equation}
where $\rho_R^{(x)}$ is a state in system $R$ and $p(x,\nu)$ is a normalization constant satisfying $p(x,\nu)\geq0$ and $\sum_{x\in \mathcal{X}}p(x,\nu)=1$, since 
\begin{equation}
    1=\tr[\nu]=\tr_{R}\left[\sum_{x\in\mathcal{X}}\tr_{\mathcal{X}}[\nu\ketbra{x}{x}]\right] = \sum_{x\in\mathcal{X}}\tr_{R}[p(x,\nu)\rho_R^{(x)}]=\sum_{x\in\mathcal{X}}p(x,\nu)
\end{equation}
Now, we have 
\begin{eqnarray}
    \mathcal{E}(\nu)=\sum_{x\in\mathcal{X}}p(x,\nu)\rho_R^{(x)}\otimes \mathcal{E}(\ketbra{x}{x}) \\
    \mathcal{F}(\nu)=\sum_{x\in\mathcal{X}}p(x,\nu)\rho_R^{(x)}\otimes \mathcal{F}(\ketbra{x}{x})
\end{eqnarray}
So, we get 
\begin{eqnarray}
    \mathbf{D}(\mathcal{E}(\nu)\|\mathcal{F}(\nu)) 
    &=& \mathbf{D}\left( \sum_{x\in\mathcal{X}}p(x,\nu)\rho_R^{(x)}\otimes \mathcal{E}(\ketbra{x}{x}) \:\middle\|\:\sum_{x\in\mathcal{X}}p(x,\nu)\rho_R^{(x)}\otimes \mathcal{F}(\ketbra{x}{x})\right) \nonumber\\
    &\overset{(a)}{\leq}&\max_{x\in\mathcal{X}}\mathbf{D}(\rho_R^{(x)}\otimes \mathcal{E}(\ketbra{x}{x})\|\rho_R^{(x)}\otimes \mathcal{F}(\ketbra{x}{x})) \nonumber\\
    &\overset{(b)}{\leq}& \max_{x\in\mathcal{X}}\left(\mathbf{D}(\rho_R^{(x)}\|\rho_R^{(x)})+\mathbf{D}\left(\mathcal{E}(\ketbra{x}{x})\|\mathcal{F}(\ketbra{x}{x})\right)\right)\nonumber\\
    &\overset{(c)}{=}& \max_{x\in\mathcal{X}}\mathbf{D}\left(\mathcal{E}(\ketbra{x}{x})\|\mathcal{F}(\ketbra{x}{x})\right)
\end{eqnarray}
Here, (a) is by joint quasi-convexity, (b) is by sub-additivity under tensor products, and (c) uses that the divergence is faithful.\\
\\
Now, we take the supremum over all input states $\nu$, to obtain
\begin{equation}
\begin{split}
    \sup_{\nu\in\mathcal{D}(R\mathcal{X})} \mathbf{D}(\mathcal{E}(\nu)\|\mathcal{F}(\nu))
    &\leq \sup_{x\in\mathcal{X}}\mathbf{D}\left(\mathcal{E}(\ketbra{x}{x})\|\mathcal{F}(\ketbra{x}{x})\right)
\end{split}
\end{equation}
which gives our desired result.
\end{proof}

\begin{lemma}[name=Input reduction for c-q channels with diamond norm, 
store={stored:lem:diamond_input_reduction}, label={lem:input_reduce_diamond}, restate-keys={manual-num={\ref{lem:input_reduce_diamond}}}]
Let $\mathcal{E},\mathcal{F}:\mathcal{X}\rightarrow A$ be c-q channels. Then
\begin{equation}
    \norm{\mathcal{E}-\mathcal{F}}_{\Diamond} =  \sup_{x\in\mathcal{X}}\norm{\mathcal{E}(\ketbra{x}{x}) - \mathcal{F}(\ketbra{x}{x})}_1
\end{equation}
\end{lemma}
\begin{proof}
Clearly, by elementary properties of suprema,
\begin{equation}
    \norm{\mathcal{E}-\mathcal{F}}_{\Diamond} =\sup_{\nu\in\mathcal{D}(R\mathcal{X})}\norm{\mathcal{E}(\nu) - \mathcal{F}(\nu)}_1 \geq  \sup_{x\in\mathcal{X}}\norm{\mathcal{E}(\ketbra{x}{x}) - \mathcal{F}(\ketbra{x}{x})}_1
\end{equation}
So, it remains to show the reverse inequality. \\
\\
Fix some (arbitrary) $\nu\in\mathcal{D}(R\mathcal{X})$. Then, we define
\begin{equation}
    \tr_{\mathcal{X}}\left[\nu\ketbra{x}{x}\right] := p(x,\nu)\rho_R^{(x)}
\end{equation}
where $\rho_R^{(x)}$ is a state in system $R$ and $p(x,\nu)$ is a normalization constant satisfying $p(x,\nu)\geq0$ and $\sum_{x\in \mathcal{X}}p(x,\nu)=1$, since 
\begin{equation}
    1=\tr[\nu]=\tr_{R}\left[\sum_{x\in\mathcal{X}}\tr_{\mathcal{X}}[\nu\ketbra{x}{x}]\right] = \sum_{x\in\mathcal{X}}\tr_{R}[p(x,\nu)\rho_R^{(x)}]=\sum_{x\in\mathcal{X}}p(x,\nu)
\end{equation}
Now, we have 
\begin{eqnarray}
    \mathcal{E}(\nu)=\sum_{x\in\mathcal{X}}p(x,\nu)\rho_R^{(x)}\otimes \mathcal{E}(\ketbra{x}{x}) \\
    \mathcal{F}(\nu)=\sum_{x\in\mathcal{X}}p(x,\nu)\rho_R^{(x)}\otimes \mathcal{F}(\ketbra{x}{x})
\end{eqnarray}
So, we get 
\begin{eqnarray}
    &&\norm{\mathcal{E}(\nu) - \mathcal{F}(\nu)}_1 \\
    &=& \norm{ \sum_{x\in\mathcal{X}}p(x,\nu)\rho_R^{(x)}\otimes \mathcal{E}(\ketbra{x}{x}) -\sum_{x\in\mathcal{X}}p(x,\nu)\rho_R^{(x)}\otimes \mathcal{F}(\ketbra{x}{x})}_1 \\
    &\overset{(a)}{\leq}&\sum_{x\in\mathcal{X}} p(x,\nu)\norm{\rho_R^{(x)}\otimes \mathcal{E}(\ketbra{x}{x})-\rho_R^{(x)}\otimes \mathcal{F}(\ketbra{x}{x})}_1 \\
    &\overset{(b)}{=}& \sum_{x\in\mathcal{X}} p(x,\nu)\norm{\rho_R^{(x)}}_1 \norm{\left(\mathcal{E}(\ketbra{x}{x}) -\mathcal{F}(\ketbra{x}{x})\right)}_1\\
    &\overset{(c)}{=}& \sum_{x\in\mathcal{X}} p(x,\nu)\norm{\mathcal{E}(\ketbra{x}{x}) -\mathcal{F}(\ketbra{x}{x})}_1
\end{eqnarray}
Here, (a) is by the triangle inequality, (b) is by multiplicativity under tensor products, and (c) uses that the trace norm of a quantum state is 1.\\
\\
Let $\Prob(X)$ be the set of all probability distributions on a set $X$. Now, we take the supremum over all input states $\nu$, to obtain
\begin{equation}
\begin{split}
    \sup_{\nu\in\mathcal{D}(R\mathcal{X})} \norm{\mathcal{E}(\nu) - \mathcal{F}(\nu)}_1
    &\leq \sup_{p\in \Prob(\mathcal{X})} \sum_{x\in\mathcal{X}}p(x)\norm{\mathcal{E}(\ketbra{x}{x}) -\mathcal{F}(\ketbra{x}{x})}_1 \\
    &= \sup_{x\in\mathcal{X}}\norm{\mathcal{E}(\ketbra{x}{x}) -\mathcal{F}(\ketbra{x}{x})}_1
\end{split}
\end{equation}
which gives our desired result.
\end{proof}

\begin{lemma}[Input reduction for divergences with infimum]\label{lem:optimal_cq_input_state}
Let $\D$ be any quantum divergence that satisfies the data-processing inequality, let $\E \in \CQ$ be a c-q channel and let $\S \subset \CQ$ be a set of CQ channels. Then the supremum over states $\nu \in \DM[R \X]$ in
\begin{equation}
    \sup_{\nu \in \DM[R \X]} \inf_{\F \in \S} \D(\E(\nu)\|\F(\nu))
\end{equation}
is achieved for $R \cong \X$ and at a state of the form
\begin{equation}
    \nu_{R X} = \sum_{x \in \X} p(x) \ketbra{x}{x}_{R} \otimes \ketbra{x}{x}_{\X}\,.
\end{equation}
\end{lemma}
\begin{proof}
Let $\nu \in \DM[R \X]$ be any state. Since all channels are C-Q and thus start with a measurement in a fixed basis of $\X$, the state $\nu$ and the state 
\begin{equation}
    \tilde{\nu}_{R \X} \coloneqq \sum_{x \in \X} \Tr_{\X}[\nu_{R\X}\ketbra{x}{x}] \otimes \ketbra{x}{x}
\end{equation}
are equivalent inputs for all C-Q channels. With $p(x) = \Tr[\nu_{R\X}\ketbra{x}{x}]$ let $\nu_{R}^{(x)} \coloneqq \Tr_{\X}[\nu_{R\X}\ketbra{x}{x}]/p(x)$, and then
\begin{equation}
    \tilde{\nu}_{R \X} = \sum_{x \in \X} p(x) \nu_{R}^{(x)} \otimes \ketbra{x}{x}\,.
\end{equation}
Define the channel $\Lambda: \tilde{X} \to R$, via $\ketbra{x}{x} \mapsto \nu_{R}^{(x)}$. Then, with
\begin{equation}
\omega_{\tilde{\X} \X} = \sum_{x \in \X} p(x) \ketbra{xx}{xx}_{\tilde{\X}\X}
\end{equation}
we find that
\begin{equation}
    D(\E(\nu_{R\X})\|\F(\nu_{R\X})) = D((\Lambda \otimes \E)(\omega_{\tilde{\X}\X})\|(\Lambda \otimes \F)(\omega_{\tilde{\X}\X})) \leq D(\E(\omega_{\tilde{\X}{\X}})\|\F(\omega_{\tilde{\X}X}))
\end{equation}
for every $\F \in \S$  by the data-processing inequality, which implies the desired statement.
\end{proof}

\begin{lemma}[store={stored:lem:perm_inputs}, label={lem:perm_inputs}, restate-keys={manual-num={\ref{lem:perm_inputs}}}]
    Let $\mathcal{S}$, $\mathcal{T}$ be two compact sets of quantum channels in $\CPTP(A^n\rightarrow B^n)$ that are permutation covariant. Then
    \begin{eqnarray}
        \sup_{\nu\in\mathcal{D}(RA^n)} \inf_{\substack{\mathcal{E}\in  \mathcal{S} \\ \mathcal{F}\in \mathcal{T}}} D(\mathcal{E}(\nu)\|\mathcal{F}(\nu)) = \sup_{\substack{\nu\in\mathcal{D}(K^nA^n) \\ K\cong A \\ \nu \,perm.\, invariant}} \inf_{\substack{\mathcal{E}\in  \mathcal{S} \\ \mathcal{F}\in \mathcal{T}}} D(\mathcal{E}(\nu)\|\mathcal{F}(\nu))
    \end{eqnarray}
    where permutation invariant on the right-hand side means with respect to permuting the n $KA$ subsystems.
\end{lemma}
\begin{proof}
    Let $\nu=\nu_{R_0A^n}$ be an arbitrary state in $\mathcal{D}(R_0A^n)$, where $R_0$ is an arbitrary reference system. Let $\pi\in S_n$ be an arbitrary permutation with unitary representation $P_A(\pi)$ and $P_B(\pi)$ on the systems $A^n$ and $B^n$ respectively. 
    \begin{eqnarray}
    \label{Perm1_eq1}
        D(\mathcal{E}(\nu)\|\mathcal{F}(\nu)) &=& D(P_B(\pi)\mathcal{E}(\nu)P_B(\pi)^\dag\|P_B(\pi)\mathcal{F}(\nu)P_B(\pi)^\dag) \\
        &=& D(\mathcal{E}(P_A(\pi)\nu P_A(\pi)^\dag)\|\mathcal{F}(P_A(\pi)\nu P_A(\pi)^\dag))
        \label{Perm1_eq2}
    \end{eqnarray}
    Here, equation (\ref{Perm1_eq1}) is because $D$ is invariant under the action of unitary matrices, and (\ref{Perm1_eq2}) is due to the permutation covariance of the channels. \\
    \\
    Now, define 
    \begin{eqnarray}
        \omega_{CR_0A^n}:=\frac{1}{n!}\sum_{\pi\in \Sn} \ketbra{\pi}{\pi}\otimes (P_A(\pi)\nu P_A(\pi)^\dag)
    \end{eqnarray}
    where the system $C$ is classical and stores the permutation (so $C$ has orthonormal basis $\{\ket{\pi}\,|\,\pi\in \Sn\}$). By the direct sum property,
    \begin{eqnarray}
        D(\mathcal{E}(\nu)\|\mathcal{F}(\nu)) &=& \frac{1}{n!} \sum_{\pi\in \Sn}D(\mathcal{E}(P_A(\pi)\nu P_A(\pi)^\dag)\|\mathcal{F}(P_A(\pi)\nu P_A(\pi)^\dag)) \\
        &=& D(\mathcal{E}(\omega_{CR_0A^n})\|\mathcal{F}(\omega_{CR_0A^n}))
    \end{eqnarray}
    Let $\omega_{WCR_0A^n}$ be a purification of $\omega_{CR_0A^n}$. Since
    \begin{eqnarray}
        \omega_{A^n}=\frac{1}{n!}\sum_{\pi\in \Sn}P_A(\pi)\nu_{A^n} P_A(\pi)^\dag
    \end{eqnarray}
    is permutation invariant, we have that there exists a system $K$, which is isomorphic to $A$, and a permutation invariant purification $\omega_{(KA)^n}\in \mathcal{D}(K^nA^n)$, by \cite[Lemma II.5]{christandl_one-and--half_2007}. Here, the permutations permute the copies of $KA$. Now the two purifications are related by a partial isometry $V:K^n\rightarrow WCR_0$, which commute with $\mathcal{E}$ and $\mathcal{F}$. Hence,
    \begin{eqnarray}
        D(\mathcal{E}(\omega_{CR_0A^n})\|\mathcal{F}(\omega_{CR_0A^n})) &\leq& D(\mathcal{E}(\omega_{WCR_0A^n})\|\mathcal{F}(\omega_{WCR_0A^n})) \\
        &=& D(\mathcal{E}(\omega_{(KA)^n})\|\mathcal{F}(\omega_{(KA)^n}))
    \end{eqnarray}
    using the data processing inequality and since the quantum relative entropy is invariant under isometries.
    Now, take an infimum over all the permutation covariant channels to get
    \begin{eqnarray}
        \inf_{\substack{\mathcal{E}\in  \mathcal{S} \\ \mathcal{F}\in \mathcal{T}}} D(\mathcal{E}(\omega_{CR_0A^n})\|\mathcal{F}(\omega_{CR_0A^n})) &\leq& \inf_{\substack{\mathcal{E}\in  \mathcal{S} \\ \mathcal{F}\in \mathcal{T}}}D(\mathcal{E}(\omega_{(KA)^n})\|\mathcal{F}(\omega_{(KA)^n}))
    \end{eqnarray}
So, we can restrict the supremum to only being over permutation invariant input states as required.
\end{proof}
\begin{EXCLUDED}
    
\section{Main New Ideas}
\subsection{GQSL Lemmas}
\Autoref{lem:divergence_input_reduction} allows us to restrict to tensor product states. It allows us to use additivity of $D$ to show that the sequence is subadditive, to use Fekete in \Autoref{lem:lim_exists}. \Autoref{lem:rel_entropy_exchange} also allows us to swap sup and inf at end, which allows us to pick a minimal $\F_n$. This is something that we regularly do in the GQSL proof. Last prop is unremarkable. 
\subsection{GQSL Strong Converse}
In \Autoref{lem:GQSL_converse_bound} we need \Autoref{lem:divergence_input_reduction} to use additivity under tensor products. We also need the finite set to take a sup over to swap limit and supremum. \\
\\
In strong converse proof we need \Autoref{lem:rel_entropy_exchange} to show we can choose an optimum $\F_k$ to apply \Autoref{lem:GQSL_converse_bound} to, to get the strong converse.
\subsection{GQSL Direct}
\Autoref{lem:sup_over_sequences} is useful as we only need to deal with fixed states - can use second half of Hayashi Lemma S8 without much alteration. Mainly necessary as we don't have an exchange lemma for $D_H^\varepsilon$.
\\
Min evalue trick in Hayashi Lemma S8 hard to do as with reference system, can have evalue arbitrarily close but not equal to zero, so don't get right rate of decay needed. Instead, we use perm invariance to bound spectrum appropriately. \Autoref{lem:perm_cov_channels} and \Autoref{lem:perm_averaging_converse} allow us to guarantee that the state we deal with is permutation invariant, given a permutation invariant input. \Autoref{lem:perm_inputs} allows us to restore sup to over all states at end, which allows this initial restriction. Finally need \Autoref{lem:rel_entropy_exchange} to get final sequence of channels.\\
\\
Final part of proof has nothing new (except notation is better).
\subsection{QRT Relative entropy of resource to Log Robustness}

$\R$ is a more natural resource measure and diamond norm convergence is stronger than convergence of choi states. \\
\\
\Autoref{lem:resource_close_sequences} uses Gour\&Winter bound instead of Winter bound. I have checked it and I am pretty sure we can apply in this circumstance. It is needed in \Autoref{lem:log_robustness1} and \Autoref{lem:QRT_converse}. \\
\\
\Autoref{lem:log_robustness1} otherwise pretty much same but assume diamond norm initially. \\
\\
\Autoref{lem:log_robustness2} needs stronger GQSL to deal with non-choi states. Construction and log robustness bound is otherwise the same. Use c-q property to break up condition on pinched states (8.63) to on individual parts of reference systems. Finiteness allows us to swap sup and lim. We then need \Autoref{lem:input_reduce_diamond} to change this sup to a sup over all states. Need this/trace out other parts \& data processing to get $n\neq mk$ part of the diamond norm limit. \\
\\
\Autoref{cor:log_robustness} nothing new.
\subsection{QRT Direct}
Need full GQSL to get bound for log robustness. \Autoref{cor:log_robustness} is necessary to get the diamond norm bound.
\subsection{QRT Converse}
Diamond norm allows us to remove asymptotic continuity condition - helps as this was an unnatural condition anyway. \\
\\
Rest is pretty much nothing new.

\end{EXCLUDED}

\printbibliography
\end{document}